\documentclass[acmsmall,screen]{acmart}

\setcopyright{rightsretained}
\acmPrice{}
\acmDOI{10.1145/3290351}
\acmYear{2019}
\copyrightyear{2019}
\acmJournal{PACMPL}
\acmVolume{3}
\acmNumber{POPL}
\acmArticle{38}
\acmMonth{1}

\usepackage{booktabs} 
\usepackage{subcaption} 

\usepackage{mathpartir}
\usepackage{graphicx}
\usepackage{wrapfig}
\usepackage[all]{xy}
\usepackage{bussproofs}
\usepackage{amsmath}
\usepackage{braket}
\usepackage{algorithm}
\usepackage[noend]{algpseudocode}
\usepackage{color}

\newcommand{\Sets}{\mathbf{Set}}

\newcommand{\Pred}{\mathbf{Pred}}

\newcommand{\Meas}{\mathbf{Meas}}

\newcommand{\QBS}{\mathbf{QBS}}

\newcommand{\labeleq}[1]{\mathbin{\ooalign{$=$\crcr\hss\raisebox{1.0ex}[0pt][0pt]{\scalebox{0.7}[0.7]{#1}}\hss}}}

\newcommand{\inverse}[1]{{#1}^{\hspace{-0.15em} - \hspace{-0.1em}1}\hspace{-0.1em}}
\newcommand{\defeq}{\stackrel{\mathrm{def}}{=}}

\newcommand{\lrangle}[1]{\langle #1\rangle}
\newcommand{\interpret}[1]{{[\![ {#1} ]\!]}}
\newcommand{\Rinterpret}[1]{{\mathopen{(\!|} {#1} \mathclose{|\!)}}}

\newcommand{\MLET}{\mathop{\mathtt{mlet}}}
\newcommand{\LET}{\mathop{\mathtt{let}}}
\newcommand{\IN}{\mathbin{\mathtt{in}}}
\newcommand{\RETURN}{\mathop{\mathtt{return}}}
\newcommand{\LETREC}{\mathop{\mathtt{letrec}}}
\newcommand{\REAL}{\mathord{\mathtt{real}}}
\newcommand{\BOOL}{\mathord{\mathtt{bool}}}
\newcommand{\IF}{\mathop{\mathtt{if}}}
\newcommand{\THEN}{\mathbin{\mathtt{then}}}
\newcommand{\ELSE}{\mathbin{\mathtt{else}}}
\newcommand{\INT}{\mathord{\mathtt{nat}}}
\newcommand{\LIST}{\mathord{\mathtt{list}}}
\newcommand{\CASE}{\mathop{\mathtt{case}}}
\newcommand{\WITH}{\mathbin{\mathtt{with}}}
\newcommand{\SCALE}{\mathop{\mathtt{scale}}}

\newcommand{\BIND}{\mathop{\mathtt{bind}}}
\newcommand{\OBSERVE}{\mathop{\mathtt{query}}}
\newcommand{\AS}{\mathbin{\Rightarrow}}

\newcommand{\NameOfRelationalLogic}{RPL}
\newcommand{\NameOfUnaryLogic}{UPL}
\newcommand{\NameOfUnderlyingLogic}{PL}
\newcommand{\PL}{\NameOfUnderlyingLogic}
\newcommand{\UPL}{\NameOfUnaryLogic}
\newcommand{\RPL}{\NameOfRelationalLogic}
\newcommand{\PCFP}{HPProg}

\newcommand{\THESYSTEM}{PPV}

\begin{document}

\title{Formal Verification of Higher-Order Probabilistic Programs}
\subtitle{Reasoning about Approximation, Convergence, Bayesian
  Inference, and Optimization}      


\author{Tetsuya Sato}
\affiliation{
 \department{CSE}    
 \institution{University at Buffalo, SUNY}   
 \country{USA}     
}
\email{tetsuyas@buffalo.edu}   

\author{Alejandro Aguirre}
\affiliation{
 \institution{IMDEA Software Institute}   
 \country{Spain}     
}
\email{alejandro.aguirre@imdea.org}   

\author{Gilles Barthe}
\affiliation{
 \institution{IMDEA Software Institute}   
 \country{Spain}     
}
\email{gjbarthe@gmail.com}   

\author{Marco Gaboardi}
\affiliation{
 \department{CSE}    
 \institution{University at Buffalo, SUNY}   
 \country{USA}     
}
\email{gaboardi@buffalo.edu}   

\author{Deepak Garg}
\affiliation{
 \institution{Max Planck Institute for Software Systems}   
 \country{Germany}     
}
\email{dg@mpi-sws.org}   

\author{Justin Hsu}
\affiliation{
 \department{CS}    
 \institution{University of Wisconsin--Madison}   
 \country{USA}     
}
\email{email@justinh.su}   


\renewcommand{\shortauthors}{T. Sato, A. Aguirre, G. Barthe, M. Gaboardi, D. Garg and J. Hsu}

\begin{abstract}
Probabilistic programming provides a convenient \emph{lingua franca} for writing
succinct and rigorous descriptions of probabilistic models and inference tasks.
Several probabilistic programming languages, including Anglican, Church or
Hakaru, derive their expressiveness from a powerful combination of continuous
distributions, conditioning, and higher-order functions. Although very important
for practical applications, these features raise fundamental challenges for
program semantics and verification. Several recent works offer promising answers
to these challenges, but their primary focus is on foundational semantics
issues.

In this paper, we take a step further by developing a suite of logics,
collectively named \THESYSTEM, for proving properties of programs written in an
expressive probabilistic higher-order language with continuous sampling
operations and primitives for conditioning distributions. Our logics mimic the
comfortable reasoning style of informal proofs using carefully selected
axiomatizations of key results from probability theory. The versatility of our
logics is illustrated through the formal verification of several intricate
examples from statistics, probabilistic inference, and machine learning. We
further show expressiveness by giving sound embeddings of existing logics. In
particular, we do this in a parametric way by showing how the semantics idea of
(unary and relational) $\top\top$-lifting can be internalized in our logics.
The soundness of \THESYSTEM\ follows by interpreting programs and assertions in
quasi-Borel spaces (QBS), a recently proposed variant of Borel spaces with a
good structure for interpreting higher order probabilistic programs.
\end{abstract}



\maketitle
\section{Introduction}
Probabilistic programming is \emph{en vogue} in statistics and machine
learning, where modern probabilistic programming languages are viewed
as a convenient \emph{lingua franca} for writing classical statistical
estimators, and for describing probabilistic models and performing
probabilistic inference. A key strength of many modern probabilistic
programming languages is their expressiveness, which allows
programmers to give succinct descriptions for a broad range of
probabilistic models, and to program specialized inference algorithms
when generic algorithms do not perform well. This
expressiveness has led to significant theoretical
challenges. Specifically, many probabilistic programming languages adopt a
combination of features that goes beyond standard program
semantics and program verification. In this paper, we consider functional probabilistic programming languages and focus on the following
elements:
\begin{itemize}
\item \emph{sampling}: the first key ingredient of a probabilistic programming
  language is a construct to sample from (continuous) distributions.
  A popular way to expose this mechanism is the monadic approach, where
  probabilities are modelled as effects. Languages feature a type constructor
  $M$ for probability measures and monadic operations for sampling from
  continuous distributions or composing probabilistic computations.

\item \emph{conditioning}: the second key ingredient of probabilistic
  programming languages is a conditioning operator, which can be used
  to build a conditional distribution that incorporates observations
  from the real world. Conditioning is often performed through specific
  constructs, such as $\mathop{\mathtt{observe}}$ or $\OBSERVE$, which scale
  a distribution to a measure according to a likelihood function,
  and then normalize the resulting measure back to a
  distribution.

\item \emph{higher-order functions}: probabilistic models and
  statistical tasks are often described
  in a natural way by means of functional higher-order
  programs. The modularity that higher-order functions provide is
  useful for writing likelihood functions, weighting functions,
  parametric models, etc. These components facilitate writing concise
  and expressive probabilistic computations.

\end{itemize}
Examples of probabilistic programming languages that
incorporate the features above are Anglican,
Church, and Hakaru.
For example, Anglican~\citep{DBLP:conf/aistats/WoodMM14} extends
Clojure with constructs for
basic probability distributions and an operation $\mathop{\mathtt{observe}}$,
which is used to build conditional distributions with respect to a
predicate representing an observation of random
variables. Church~\citep{DBLP:conf/uai/GoodmanMRBT08} supports in a simply typed lambda calculus a
similar conditioning operation named $\mathop{\mathtt{query}}$, Hakaru~\citep{DBLP:conf/flops/NarayananCRSZ16}
supports these features as a domain-specific language embedded in
Haskell. 

Despite their popularity, higher-order probabilistic programming languages pose
significant challenges for semantics and verification. In particular, a
classical result~\citep{aumann1961} shows that the category of measurable spaces
is not Cartesian closed, and thus it cannot be used to give denotational models
for higher-order probabilistic languages. Aumann's negative result has triggered
a long line of research, culminating in several recent proposals for semantic
models of higher-order probabilistic languages.  One such proposal, relevant to
our work, is the notion of the quasi-Borel space (QBS)~\citep{HeunenKSY17}, which
has a rich categorical structure and yields an elegant denotational model for
higher-order probabilistic programs.

While a denotational model facilitates formal reasoning about probabilistic
programs, the resulting style of reasoning is typically hard to use. As with
more standard programming languages, we would prefer to use other
techniques, such as equational methods and program logics, to
structure the arguments at a higher level. Several recent papers have started
to look at this. For instance, \citet{Staton17ESOP} and~\citet{CulpepperC17}
have recently proposed equational methods for proving equivalences between
higher-order probabilistic programs.
\citet{CulpepperC17} propose an equational framework based on
observational equivalence and logical relations, while
\citet{Staton17ESOP} proposes a semantics method for equational
reasoning which can be used for program equivalence. 
These two methods  are important steps towards more
general high-level reasoning techniques. However, their main focus is program
equivalence and they do not directly support arbitrary program
properties. Moreover, their approach is based on techniques which
are difficult to directly apply to complex examples. As a result, for
more complex examples the only currently viable approach is to resort directly
to the denotational semantics; for instance,
\citet{Scibior:2017:DVH:3177123.3158148} use semantic methods to prove the
correctness of higher-order Bayesian inference.

\paragraph{Our work: }
The long-term goal of our research is to build practical verification
tools for higher-order probabilistic programs, and to leverage these
tools for building libraries of formally verified algorithms from
machine learning and statistics. This paper makes an initial step
towards this goal and justifies its feasibility by introducing a framework, called the Probabilistic Programming Verification
framework (\THESYSTEM), for proving (unary and relational) properties of
probabilistic higher-order programs with discrete and continuous
distributions. \THESYSTEM\ is:
\begin{itemize}
\item \emph{expressive}: it can reason about different properties of
  probabilistic programs, including approximation, convergence,
  probabilistic inference and optimization.

\item \emph{practical}: it supports lean derivations that abstract away from
  lower-level concerns, like measurability.
\item \emph{sound}: it can be soundly interpreted in the category of quasi-Borel
  spaces.
\end{itemize}

\THESYSTEM's design is based on three different logics: \PL, \UPL\ and \RPL. These logics are presented in the style
of~\citet{Aguirre:2017:RLH:3136534.3110265}: \PL\ is an
\emph{intuitionistic logic} for
reasoning about higher-order programs using a style inspired by
HOL~\cite{JacobsM93} based on judgments of the form $\Gamma \mid \Psi \vdash_{\mathrm{PL}}
\phi$. \UPL\ is a \emph{unary program logic} which manipulates judgments of the form $\Gamma
\mid \Psi \vdash_{\mathrm{UPL}} e \colon \tau \mid \phi$. Finally,
\RPL\ is a \emph{relational program logic} which manipulates judgments of the form
$\Gamma \mid \Psi \vdash_{\mathrm{RPL}} e \colon \tau \sim e' \colon \tau' \mid
\phi'$. Here $\Gamma$ is a simple typing context; $\tau$ and
$\tau'$ are the simple types of the expressions $e$ and $e'$; $\Psi$ is a
set of assumed assertions; $\phi$ is a postcondition; and $\phi'$ is a
relational postcondition. The proof systems are equi-expressive, but
the \UPL\ and \RPL\ are closer to the syntax-directed
style of reasoning generally favored in unary and relational program verification, respectively. We
define an interpretation of assertions in the category of QBS
predicates and prove that the logics are sound with respect to the
interpretation. This interpretation guarantees that every subset of a
quasi-Borel space yields an object in the category.  As a consequence, assertions of the logic are interpreted
set-theoretically, and extensionality is valid. This facilitates
formal reasoning and formal verification.

To further ease program verification, we define carefully crafted
axiomatizations of fundamental probabilistic definitions and results,
including expectations as well as concentration bounds. Following
\citet{Scibior:2017:DVH:3177123.3158148}, we validate the soundness of
these axiomatizations using synthetic measure theory for the QBS
framework. This ensures that a derivation based on our proof system
and axioms is valid in quasi-Borel spaces. A consequence of
this approach is that, in order to verify programs, a 
user of \THESYSTEM\ can focus on higher-level reasoning about probabilistic
programs, rather than the specific details of QBS.

We validate our design through a series of examples from statistics,
Bayesian inference and machine learning. We also demonstrate that our
systems can be used as a framework where other
program logics can be embedded. We show this in a parametrized way by
using \THESYSTEM\ to define a family of \emph{graded
  $\top\top$-liftings}, a logical relation-like technique to construct
predicates/relations over probability distributions, starting from
predicates/relations over values.  As a concrete application, we
embed two recent probabilistic logics: a \emph{union bound} logic for
reasoning about accuracy~\citep{BartheGGHS16ICALP}, and a logic for reasoning
about probability distributions through 
couplings~\citep{DBLP:conf/esop/0001BBBG018}.

Overall, our work provides a fresh, verification-oriented perspective
on quasi-Borel spaces, and contributes to establish their status as a
sound, simple and natural theoretical framework for practical verification of higher-order
probabilistic programs.

In technical terms,
our framework follows the presentation style introduced by 
\citet{Aguirre:2017:RLH:3136534.3110265} to reason about deterministic
higher-order programs. We extend this approach to higher-order
probabilistic programs with continuous random variables. A similar 
approach has been used for discrete random variables
by~\citet{DBLP:conf/esop/0001BBBG018} in order to reason about unary and
relational properties of Markov chains.  Our contribution differs
significantly from the one by~\citet{DBLP:conf/esop/0001BBBG018}.  Assertions in their framework 
  are non-probabilistic and are interpreted first over deterministic values, and
  then over distributions over values by probabilistic lifting.
Instead, in  \THESYSTEM\ we can reason about
  (monadic) probabilistic expressions directly in assertions. This is
  a key component in expressing probabilistic properties such as
  the convergence of the  expectation of an
  expression directly. Moreover,~\citet{DBLP:conf/esop/0001BBBG018} support
  analysis of probabilistic programs via
  coupling arguments only. \THESYSTEM's proof rules are more expressive: they allow reasoning 
  about probabilities within the logic.

\section{\THESYSTEM\ by example}
\label{sec:section2}
In this section we introduce the general ideas behind \THESYSTEM\ through
two examples.
%
\paragraph{Continuous Observations: Two Uniform Samples}
This warm-up example serves as an introduction to Bayesian conditioning
and how we can reason about it in our system.
Let us consider the following program $\mathtt{twoUs}$:
\[
\begin{array}{r@{}l}
\mathtt{twoUs}\equiv\
&\LET u_1 =\mathtt{Uniform}(0,1)\ \IN\ \LET u_2 =\mathtt{Uniform}(0,1)\ \IN\\
&\LET y = u_1\otimes u_2\ \IN\\
&\OBSERVE \ y \AS \lambda{x}.(\IF \pi_1(x)<.5 \lor
\pi_2(x)>.5 \THEN 1 \ELSE 0)
\end{array}
\]
The first line defines two uniform distributions  $u_1$ and
$u_2$. The second line pairs the two distributions together using the product measure of
$u_1$ and $u_2$ which we denote $u_1\otimes u_2$ (this is defined
formally in Section \ref{sec:PCFP}). 
Then, the third line performs
Bayesian conditioning on this product measure using the construction $\OBSERVE$. The \emph{prior} $y$ gets conditioned by the
\emph{likelihood function} corresponding to the observation $\pi_1(x)<.5 \lor
\pi_2(x)>.5$, and a \emph{posterior} is computed. In this simple example, this
is morally equivalent to giving score $1$ to the traces that do satisfy the
assertion, and score $0$ to the ones that do not satisfy it, and rescaling the
distribution. In general, we can use the conditioning construct with an
arbitrary likelihood function to perform more general inference. After the
observation, the posterior is a uniform distribution over the set $\{(x_1,x_2)
\mid x_1 < .5 \lor x_2 > .5\}$.

The simple property we will show is that $\Pr_{(x_1,x_2) \sim \mu}[x_1 > .5 ] =
1/3$, where $\mu$ is the posterior after the observation and the pair
$(x_1,x_2)$ is distributed by $\mu$.  This is expressed in the unary logic
\UPL---since this is a unary property---through the following judgment:
\[
{\textstyle
\vdash_{\mathrm{\UPL}}\mathtt{twoUs}: M[\REAL \times \REAL] \mid \Pr_{z \sim \mathbf{r}}[\pi_1(z) > .5] = 1/3
}
\]
where the distinguished variable $\mathbf{r}$ in the logical assertion
represents the given term $\mathtt{twoUs}$ and the variable $z$ is bound by
$\Pr_{z\sim \mathbf{r}}[\ldots]$ and it is used to represent the value sampled from the probability
distribution $\mathbf{r}$.
We show informally how to derive this assertion. The system
\UPL\ allows us to reason in a syntax-directed manner. Since the
program starts with three let bindings, the first step will be to apply the rule for let bindings
three times. This rule, which we will present formally in
Section~\ref{sec:logic:program}, moves $u_1,u_2$ and $y$ plus the
logical assertions about them into the context. The resulting
judgement is:
\[
\begin{array}{r@{}l}
&
\Pr_{z \sim u_1}[z > .5] = 1/2,  
\Pr_{z \sim u_2}[z < .5] = 1/2,
\Pr_{z \sim u_1}[\top] = 1,  
\Pr_{z \sim u_2}[\top] = 1,
y = u_1\otimes u_2
\\
& \vdash_{\mathrm{\UPL}} \OBSERVE y \AS \lambda{x}.(\IF \pi_1(x)<.5 \lor
\pi_2(x)>.5 \THEN 1 \ELSE 0)\colon
M[\REAL \times \REAL] \mid\\
&\quad \Pr_{z \sim \mathbf{r}}[\pi_1(z) > .5] = 1/3
\end{array}
\]
where for simplicity we omitted the typing context.  The logical assertions on
$u_1$ and $u_2$ can be easily discharged using the assumption that they are
distributed uniformly as $\mathtt{Uniform}(0,1)$, i.e. uniformly
between $0$ and $1$.  To finish the proof, we want to use the
fact that $\OBSERVE$ corresponds to conditioning. In \UPL\ we can do this using
the following special rule that internalizes the Bayesian properties of
$\OBSERVE$:\footnote{%
We introduce the rule here to give some intuition, but this is also discussed in
Section~\ref{sec:logic:program} after introducing \THESYSTEM.}
\[
{
\RightLabel{[Bayes]}
\AxiomC{
$\Gamma, x\colon \tau \vdash e' \colon \mathtt{bool}$
\quad
$\Gamma, x \colon \tau \vdash e'' \colon \mathtt{bool}$
\quad
$\Gamma \vdash e \colon M[\tau]$
}
\UnaryInfC{
$
\Gamma\mid\Psi\vdash_{\mathrm{\NameOfUnaryLogic}}
\OBSERVE e \AS \lambda{x}.(\IF e'\THEN 1\ELSE 0)
{\colon} M[\tau] \mid
\Pr_{y \sim \mathbf{r}}[e''[y/x]]
=
\frac{\Pr_{x \sim e}[e' \land e'' ]}{\Pr_{x \sim e}[e']}
$}
\DisplayProof
}
\]
This rule corresponds to a natural reasoning principle (derived by
Bayes' theorem) for $\OBSERVE$ when we have a boolean condition
as the likelihood function: the probability of an event $e''$ under the
\emph{posterior} distribution is equal to the probability of the
intersection of the event $e''$ and the observation $e'$, under the
\emph{prior} distribution $e$, divided by the probability of $e'$
under the prior distribution $e$.

To apply this rule we need to rewrite the postcondition into
the appropriate shape: a fraction that has the
probability of a conjunction of events in the numerator and the
probability of the observed event in the denominator. This can be done in \UPL\ through
\emph{subtyping} which lets us reason directly in the logic \PL, where we can
prove the following judgment:
\[
\begin{array}{r@{}l}
&
\Pr_{z \sim u_1}[z > .5] = 1/2,  
\Pr_{z \sim u_2}[z < .5] = 1/2,
\Pr_{z \sim u_1}[\top] = 1,  
\Pr_{z \sim u_2}[\top] = 1,
y = u_1\otimes u_2\\
&\vdash_{\mathrm{\PL}} \frac{\Pr_{z \sim y}[(\pi_1(z)<.5 \lor\pi_2(z)>.5) \land (\pi_1(z)>.5)]}
             {\Pr_{z \sim y}[\pi_1(z)<.5 \lor\pi_2(z)>.5]} = \frac{1/4}{3/4} =
             1/3
\end{array}
\]
Using this equivalence and subtyping we can rewrite the judgment we need to prove as follows:
\[
\begin{array}{r@{}l}
&
\Pr_{z \sim u_1}[z > .5] = 1/2,  
\Pr_{z \sim u_2}[z < .5] = 1/2,
\Pr_{z \sim u_1}[\top] = 1,  
\Pr_{z \sim u_2}[\top] = 1,
y = u_1\otimes u_2
\\
&\vdash_{\mathrm{\UPL}} \OBSERVE y \AS  \lambda{x}.(\IF \pi_1(x)<.5 \lor \pi_2(x)>.5 \THEN 1 \ELSE 0)
: M[\REAL \times \REAL] \mid
\\
&\quad \Pr_{z \sim \mathbf{r}}[\pi_2(z)>.5] = 
\frac{\Pr_{z \sim y}[(\pi_1(z)<.5 \lor\pi_2(z)>.5) \land (\pi_1(z)>.5)]}
             {\Pr_{z \sim y}[\pi_1(z)<.5 \lor\pi_2(z)>.5]}
\end{array}
\]
and this can be proved by applying the [Bayes] rule above, concluding the proof.
We saw different components of \THESYSTEM at work here: unary rules, subtyping, and a special rule for $\OBSERVE$. All these components can be assembled in more complex examples, as we show in Section~\ref{sec:examples}.

\paragraph{Monte Carlo Approximation}
As a  second example, we show how to use \THESYSTEM\ to reason about other classical applications that do not use
observations. We consider reasoning about expected value and
variance of distributions.
Concretely, we show convergence in probability of an implementation of
the naive Monte Carlo approximation.  This algorithm considers
a distribution $d$ and a real-valued function $h$, and tries to
approximate the expected value of $h(x)$ where $x$ is sampled from $d$
by sampling a number $i$ of values from $d$ and computing their mean.

Consider the following implementation of Monte Carlo approximation:
\[
\begin{array}{r@{}l}
\mathtt{MonteCarlo}
\equiv&
 \LETREC f(i \colon \INT ) = \IF (i \leq 0) \THEN { \RETURN (0)}\\
&\ELSE \MLET m=f(i-1) \IN \MLET x=d \IN \RETURN ((1/i)\ast(h(x)+m\ast(i-1)))
\end{array}
\]
Our goal is to prove the convergence in probability of this algorithm, that is, 
the result can be made as accurate as desired by increasing the sample
size (denoted by $i$ above and $n$ below). This is formalized in the following \UPL\ judgment (we omit the typing context for simplicity):
\begin{equation}
\label{eq:pre_example:MonteCarlo:0}
\begin{array}{r@{}l}
& (\mathbb{E}_{x \sim d}[1] = 1), (\sigma^2 = \mathrm{Var}_{x \sim d}[h(x)]), (\mu = \mathbb{E}_{x \sim d}[h(x)]), (\varepsilon > 0) \vdash_{\mathrm\UPL} \\
& \mathtt{MonteCarlo}\colon \INT \to M[ \REAL ] \mid
	\forall n, (n > 0) \implies \Pr_{y \sim \mathbf{r}n}[|y-\mu| \geq \varepsilon] \leq {\sigma^2}/{n\varepsilon^2}
\end{array}
\end{equation}
Formally, we are showing that the probability that the computed mean $y$ differs
from the actual mean $\mu$ by more than $\varepsilon$ is upper bounded by a
value that depends inversely on $n$---more samples lead to better estimates.
To derive (\ref{eq:pre_example:MonteCarlo:0}) in \UPL\ we need to perform two steps:
\begin{itemize}
\item Calculating the mass, mean, and variance of $\mathtt{MonteCarlo}$ in \NameOfUnaryLogic:
\begin{equation}
\label{eq:pre_example:MonteCarlo:00}
\begin{array}{rl}
&(\mathbb{E}_{x \sim d}[1] = 1), (\sigma^2 = \mathrm{Var}_{x \sim
  d}[h(x)]),(\mu = \mathbb{E}_{x \sim d}[h(x)]), (\varepsilon > 0) \vdash_{\mathrm\UPL} \\
&\mathtt{MonteCarlo}\colon \INT \to M[ \REAL ] \mid \\
& \forall n \colon \INT. (n > 0) \implies 
(\mathbb{E}_{y \sim \mathbf{r}n}[1] = 1)
\land
(\mathbb{E}_{y \sim \mathbf{r}n}[y] = \mu)
\land
(\mathrm{Var}_{y \sim \mathbf{r}n}[y] = {\sigma^2}/{n})
\end{array}
\end{equation}
\item Applying the Chebyshev inequality (formula \ref{eq:MHOL:lemma:Chebyshev_inequality} in Section \ref{sec:markov_chebyshev}) to (\ref{eq:pre_example:MonteCarlo:00}) using subtyping.
\end{itemize}
We focus on the proof of (\ref{eq:pre_example:MonteCarlo:00}), which
is carried out
by induction on $n$. In our system, the rule for $\LETREC$ lets us prove inductive
properties of (terminating) recursive
functions by introducing an inductive hypothesis into
the set of assertions that can only be instantiated for smaller arguments. After applying this
rule, the new goal is:
\[ \phi_{\mathrm{IH}} \equiv
\forall n \colon \INT. (n < i){\implies}(n > 0){\implies} 
(\mathbb{E}_{y \sim f(n)}[1] = 1)
\land
(\mathbb{E}_{y \sim f(n)}[y] = \mu)
\land
(\mathrm{Var}_{y \sim f(n)}[y] = {\sigma^2}/{n})
\]
On this, we can apply a rule for case distinction according to the two branches of the
if-then-else, which gives us the following two premises:
\[
\begin{array}{rl}
& \Psi, (i\leq 0) \vdash
 \RETURN (0) 
\mid \psi
\\
&\Psi, (i > 0) \vdash
\MLET m = f(i-1) \IN (\MLET x=d \IN
\RETURN (\frac{1}{i}(h(x)+m\ast(i-1))))
\mid \psi
\end{array}
\]
where $\Psi = (\mathbb{E}_{x \sim d}[1] = 1), (\mu = \mathbb{E}_{x \sim d}[h(x)]), (\sigma^2 = \mathrm{Var}_{x \sim d}[h(x)]), (i > 0),
\phi_{\mathrm{IH}}$ and
$\psi = (\mathbb{E}_{y \sim \mathbf{r}i}[1] = 1)
\land
(\mathbb{E}_{y \sim \mathbf{r}i}[y] = \mu)
\land
(\mathrm{Var}_{y \sim \mathbf{r}i}[y] = {\sigma^2}/{i})
$.
The first premise is obvious since the assumptions $(i > 0)$ and $(i \leq 0)$ are
contradictory. The second premise follows from subtyping applied to a
\NameOfUnderlyingLogic-judgment that is proved by instantiating the induction hypothesis
with $i-1$ and applying axioms on expected values. This concludes the proof.

Again, we have seen here several different components of \THESYSTEM: unary
rules (including the rule for inductive reasoning), subtyping, and the use of equations and axioms. We further illustrate these components of \THESYSTEM\ as well as others in verifying more
involved examples (including relational examples) in Section~\ref{sec:examples}.

\paragraph{Remark: }
In this work, we assume that $\OBSERVE$ is always defined: we don't
consider programs ``observing'' events with zero probability.  We make
this simplification to focus here on program verification without the
need to reason about whether a $\OBSERVE $ statement is defined or
not. This approach was used, for example, in~\citet{BartheFGAGHS16} to
reason about differential privacy for Bayesian processes.  We believe
that the problem of identifying ways to reason about when a $\OBSERVE
$ statement is defined is an important one, but it is orthogonal to
the formal reasoning we consider here. Other work has focused on this
problem~\citep{BorgstromGGMG11,Scibior:2017:DVH:3177123.3158148,ShanR17,HeunenKSY17}. In
a similar way, we consider only programs that terminate, without
stipulating a specific method to prove termination.

\section{\PCFP: a higher-order probabilistic programming language}
\label{sec:PCFP}
We present the probabilistic language \PCFP\ we use in this paper. The language is an extension of the simply-typed lambda calculus 
with products, coproducts, natural numbers, lists, (terminating) recursion and \emph{the monadic type for probability}.
The types of \PCFP\ are defined by the following grammar.
\[
\begin{array}{rlll}
\tilde{\tau}
&\mathbin{::=} & \mathtt{unit}
\mid \mathtt{bool}
\mid \INT
\mid \REAL
\mid \mathtt{pReal}
\mid \tilde{\tau} \times \tilde{\tau}
\mid \mathtt{list}(\tilde{\tau})
&\text{(Basic Types)}\\
\tau &\mathbin{::=} &\tilde{\tau}
\mid M[\tau] 
\mid \tau \to \tau
\mid \tau \times \tau
\mid \mathtt{list}(\tau)
&\text{(Types)}
\end{array}
\]
We distinguish two sorts of types: Basic Types and Types.
The former, as the name suggests, include standard basic types (where
$\mathtt{pReal}$ is the type of positive real numbers), and products and lists of them. The latter include a monadic type $M[\tau]$ for general measures
on $\tau$, as well as function and product types. As we will
see later in Section~\ref{sec:semantics}, Basic Types will be interpreted in standard Borel spaces, while for general Types we will need quasi-Borel
spaces. 
The language of \PCFP\ expressions is defined by the following grammar.
\[
\begin{array}{rll}
e
&\mathbin{::=} & x
\mid c
\mid f
\mid e~e
\mid \lambda x.e
\mid \lrangle{e,e}
\mid \pi_i(e)
\mid \mathtt{case}~e~\mathtt{with}~[ d_i \overline{x_i} \Rightarrow e_i ]_i
\mid \mathtt{letrec}~f x = e
\\
& &
\mid \mathtt{return}~e
\mid \mathtt{bind}~e~e
\mid \OBSERVE ~e~\AS~e
\mid \mathtt{Uniform} (e,e)
\mid \mathtt{Bern} (e)
\mid \mathtt{Gauss} (e,e)
\end{array}
\]
Most of the constructs are standard. We use $c$ to
range over a set of basic constants and $f$ to range over a set of
primitive functions. We have monadic constructions $\mathtt{return}~e$ and
$\mathtt{bind}~e_1~e_2$ for the monadic type $M[\tau]$, a conditioning
construction $\OBSERVE ~e_1~\AS~e_2$ for computing the posterior
distribution given a prior distribution $e_1$, and a likelihood
function $e_2$, and primitives representing basic probability distributions.

\PCFP\ expressions are simply typed, using rules that are mostly standard. We show only selected rules here:
\begin{gather*}
{
\AxiomC{
$
\Gamma  \vdash e \colon M[\tau]\qquad
\Gamma \vdash e' \colon \tau \to \mathtt{pReal} 
$
}
\UnaryInfC{$\Gamma \vdash \OBSERVE ~e~\AS~e' \colon  M[\tau] $}
\DisplayProof
}
\qquad
{
\AxiomC{
$
\Gamma  \vdash e \colon M[\tau]\qquad
\Gamma \vdash e' \colon \tau \to M[\tau'] 
$
}
\UnaryInfC{$\Gamma \vdash \BIND e~ e' \colon M[\tau']$}
\DisplayProof
}
\\
{
\AxiomC{
$
\Gamma, f \colon I \to \sigma, x \colon I\vdash e \colon \sigma
$
\quad
$
I \in \{\INT,\LIST(\tau)\}
$
\quad
$
\mathit{Terminate}(f,x,e)
$
}
\UnaryInfC{$\Gamma \vdash \LETREC f x = e \colon I \to \sigma$}
\DisplayProof
}
\end{gather*}
Here, $\mathit{Terminate}(f,x,e)$ is any termination criterion which
ensures that all recursive calls are on smaller arguments.  We also
consider a basic equational theory for expressions based on
$\beta$-reduction, extensionality and monadic rules. These are also
standard and we omit them here. We enrich this equational theory with
axioms and equations reflecting common reasoning principles for
probabilistic programming in Section~\ref{sec:axioms}.

For convenience, we  use some syntactic sugar: 
$(\LET x = e_1 \IN e_2) \equiv (\lambda{x}.~e_2)e_1$,
$(\MLET x = e_1 \IN e_2) \equiv \BIND e_1~ (\lambda{x}.~e_2)$, and
$e_1 \otimes e_2 \equiv \BIND e_1~(\lambda{x}.~\BIND e_2~ (\lambda{y}.\RETURN\lrangle{x,y}))$.
Thanks to the commutativity of $M$ (the Fubini-Tonelli equality; see Section~\ref{sec:axioms}), the semantics of $e_1 \otimes e_2$ is exactly the \emph{product measure} of $e_1$ and $e_2$.

\section{\PL: A Logic for Probabilistic Programs}
\label{sec:MHOL}
In this section we introduce a logic, named \PL, for reasoning about
higher-order
probabilistic programs. This logic forms the basis of \THESYSTEM. \PL\
contains basic predicates over expressions of \PCFP.  To support
more natural verification in \THESYSTEM, we enrich \PL\ with a set of axioms
encompassing a wide variety of reasoning principles over probabilistic programs.


\PL\ follows the style of higher-order simple predicate logic
(HOL)~\cite{JacobsM93}, where quantified variables can be of arbitrary
types, but extends HOL with assertions about
probabilistic constructions.
Terms and formulas of \PL\ are defined by the following 
grammar:
\[
\begin{array}{rlll}
t   &{::=}& e \mid \mathbb{E}_{x \sim t}[t(x)] \mid \mathtt{scale}(t,t)
      \mid \mathtt{normalize}(t) \qquad\qquad\qquad\quad\qquad\text{(enriched expressions)}\\
\phi &{::=}& (t=t)\mid (t<t) \mid \top \mid \bot \mid \phi \land \phi
       \mid \phi \implies \phi \mid \neg \phi \mid \forall{x \colon
       \tau}.\phi \mid \exists{x \colon \tau}.\phi\quad \text{(logical formulas)}
\end{array}
\]
Enriched expressions enrich \PCFP\ expressions
 $e$ with constructions for 
expectations $\mathbb{E}_{x \sim t}[t(x)]$, rescaling of measures
$\mathtt{scale}(t,t)$, and normalization $\mathtt{normalize}(t)$.
A logical formula $\phi$ is a formula built over
equalities and inequalities between enriched expressions. 
%

Similar to expressions in \PCFP, we consider only well-typed enriched
expressions. Typing rules for the additional
constructs of \PL\ are the following.
\[
{
\AxiomC{$\Gamma \vdash t_1 \colon M[\tau]$ \ \ $\Gamma \vdash t_2 \colon \tau \to \mathtt{pReal}$}
\UnaryInfC{$\Gamma \vdash \mathbb{E}_{x \sim t_1}[t_2(x)] \colon\mathtt{pReal}$}
\DisplayProof
}
\ \
{
\AxiomC{$\Gamma \vdash t_1 \colon M[\tau]$ \ \ $\Gamma \vdash t_2 \colon \tau \to \mathtt{pReal}$ }
\UnaryInfC{$\Gamma \vdash \mathtt{scale}(t_1,t_2) \colon M[\tau]$}
\DisplayProof
}
\ \
{
\AxiomC{$\Gamma \vdash t \colon M[\tau]$}
\UnaryInfC{$\Gamma \vdash \mathtt{normalize}(t) \colon M[\tau]$}
\DisplayProof
}
\]
Intuitively, $\mathbb{E}_{x \sim t_1}[t_2(x)] $ is the expected value
of the function $t_2$ over the distribution $t_1$;
$\mathtt{scale}(t_1,t_2)$ is a distribution obtained from an
underlying measure $t_1$ by rescaling its components by means of the density function $t_2$;
$\mathtt{normalize}(t)$ is the normalization of a measure $t$ to a
probability distribution (a measure with mass $1$).
%
Expectations of real-valued functions are defined by the difference
of positive and negative parts.
Precisely, for given $\Gamma \vdash t_1 \colon M[\tau]$ and $\Gamma \vdash t_2 \colon \tau \to \REAL$,
we define the expectation as the following syntactic sugar:\footnote{%
We use absolute values $| - | \colon \REAL \to \mathtt{pReal}$ to adjust the typing.
The right-hand side is undefined if both expectations are infinity.
We could avoid this kind of undefinedness by stipulating
$\infty-\infty = -\infty$, but we leave it undefined since this
actually never shows up in our concrete examples.}
\[
\mathbb{E}_{x \sim t_1}[t_2(x)]
\equiv 
\mathbb{E}_{x \sim t_1}[\IF t_2(x) > 0 \THEN | t_2(x) | \ELSE 0]
- \mathbb{E}_{x \sim t_1}[\IF t_2(x) < 0 \THEN | t_2(x) | \ELSE 0]
\]
We can also define variance and probability in terms of expectation:
\[
\Pr_{x \sim e}[e']
\equiv \mathbb{E}_{x \sim e}[\mathtt{if}~e'~\mathtt{then}~1~\mathtt{else}~0]
\qquad
\label{eq:variance:definition}
\mathrm{Var}_{x \sim e_1}[e_2]
\equiv \mathbb{E}_{x \sim e_1}[(e_2)^2] - (\mathbb{E}_{x \sim e_1}[e_2])^2
\]
A \PL\ judgment has the form $\Gamma \mid \Psi \vdash_{\mathrm{\PL}}
\phi$ where $\Gamma$ is a context assigning types to variables, $\Psi$
is a set of formulas well-formed in the context $\Gamma$, and $\phi$
is a formula also well-formed in $\Gamma$. Rules to derive
well-formedness judgments $\Gamma \vdash \phi\;\mathsf{wf}$ are rather
standard and we omit them here. We often refer to $\Psi$ as the
\emph{precondition}.
The proof rules of \PL\ are rather standard, so we show only a
selection in Figure~\ref{fig:PL-rules}. 

We extend equational rules and axioms for standard expressions to enriched expressions. We also
introduce some axioms specific to enriched expressions in Section~\ref{sec:axioms}.
\begin{figure}
\begin{gather*}
{
\AxiomC{$\Gamma \vdash t \colon \tau$ \quad $\Gamma \vdash t' \colon \tau$ \quad $t =_{\beta\iota\mu} t'$}
\RightLabel{[CONV]}
\UnaryInfC{$\Gamma \mid \Psi \vdash_{\mathrm{\PL}} t = t'$}
\DisplayProof
}
\qquad
{
\AxiomC{$\Gamma \mid \Psi \vdash_{\mathrm{\PL}} \phi[t/x]$ \quad $\Gamma \mid \Psi \vdash_{\mathrm{\PL}} t = u$ }
\RightLabel{[SUBST]}
\UnaryInfC{$\Gamma \mid \Psi \vdash_{\mathrm{\PL}} \phi[u/x]$}
\DisplayProof
}
\\
{
\AxiomC{$\phi \in \Psi$}
\RightLabel{[AX]}
\UnaryInfC{$\Gamma \mid \Psi \vdash_{\mathrm{\PL}} \phi$}
\DisplayProof
}
\qquad
{
\AxiomC{$\Gamma \mid \Psi,\psi \vdash_{\mathrm{\PL}} \phi$ }
\RightLabel{[$\Rightarrow_I$]}
\UnaryInfC{$\Gamma \mid \Psi \vdash_{\mathrm{\PL}}\psi \implies \phi$}
\DisplayProof
}
\qquad
{
\AxiomC{$\Gamma \mid \Psi \vdash_{\mathrm{\PL}} \psi \implies \phi$  \quad $\Gamma \mid \Psi \vdash_{\mathrm{\PL}} \psi$ }
\RightLabel{[$\Rightarrow_E$]}
\UnaryInfC{$\Gamma \mid \Psi \vdash_{\mathrm{\PL}} \phi$}
\DisplayProof
}
\end{gather*}
\caption{Selection of rules for the \PL\ logic.}
\label{fig:PL-rules}
\end{figure}

\section{Axioms and Equations of Assertions for Statistics}
\label{sec:axioms}
%
In this section, we introduce axioms and equations for the logic \PL.
First, we have the standard equational theory for (enriched) expressions covering $\alpha$-conversion, $\beta$-reduction, extensionality, and the monadic rules of the monadic type $M$. We omit these standard rules here.
The monadic type $M$ also has commutativity (the Fubini-Tonelli equality), represented by the following equation:
%
\begin{equation}
\label{eq:MHOL:commutativity}
(\BIND e_1 ~\lambda{x}.(\BIND e_2 ~\lambda{y}.e(x,y)))
=
(\BIND e_2 ~\lambda{y}.(\BIND e_1 ~\lambda{x}.e(x,y)) \quad (x, y \text{fresh})
\end{equation}

We introduce some equalities pertaining to expectation.
We have the monotonicity and linearity of expectation (axioms \ref{ineq:MHOL:expected_monotonicity}, \ref{eq:MHOL:expected:linearity}), and we also have the Cauchy-Schwarz inequality (axiom \ref{ineq:expectation:Cauchy-Schwarz}). 
Finally, we can transform variables in expressions related to expectation by substitution (axiom \ref{eq:MHOL:expected:variable_transformation}).
\begin{gather}
\label{ineq:MHOL:expected_monotonicity}
(\forall{x \colon \tau}.~e' \geq 0) \implies \mathbb{E}_{x \sim e}[e'] \geq 0
\\
\label{eq:MHOL:expected:linearity}
\mathbb{E}_{x \sim e}[e_1 \ast e_2] = e_1 \ast \mathbb{E}_{x \sim e}[e_2]
\quad (x\notin \mathrm{FV}(e_1)),
\qquad
\mathbb{E}_{x \sim e}[e_1 + e_2] = \mathbb{E}_{x \sim e}[e_1] + \mathbb{E}_{x \sim e}[e_2] 
\\
\label{ineq:expectation:Cauchy-Schwarz}
(\mathbb{E}_{x \sim e}[e_1 \ast e_2])^2 \leq \mathbb{E}_{x \sim e}[e_1^2] \ast \mathbb{E}_{x \sim e}[e_2^2]
\\
\label{eq:MHOL:expected:variable_transformation}
\mathbb{E}_{x \sim \BIND e~\lambda{y}.\RETURN(e')}[e'']
=\mathbb{E}_{y \sim e}[e''[e'/x]]
\end{gather}
We also introduce some basic equalities pertaining to observation, rescaling,
and normalization. 
\begin{gather}
\label{eq:MHOL:expected:scaling}
\mathbb{E}_{x \sim d'}[h(x)\cdot g(x)] = \mathbb{E}_{x \sim \mathtt{scale}(d',g)}[h(x)].
\\
\label{eq:MHOL:scaling1}
(\mathtt{scale}(\mathtt{scale}(e_1, e_2),~ e_3) =
(\mathtt{scale}(e_1, \lambda{x}.(e_2(x) \ast e_3(x))),
\quad
e = \mathtt{scale}(e,\lambda{\_}.1)
\\
\label{eq:MHOL:scaling2}
(\MLET x = \mathtt{scale}(e_1, e_2) \IN e_3(x))
=
(\MLET x = e_1 \IN \mathtt{scale}(e_3(x), \lambda{u}.e_2(x)))
\\
\label{eq:MHOL:scaling3}
\mathtt{scale}(e_1, e_2)\otimes \mathtt{scale}(e_3, e_4)
=
\mathtt{scale}( e_1 \otimes e_2, \lambda{w}. e_2(\pi_1(w)) \ast e_4(\pi_2(w)))
\\
\label{eq:MHOL:scaling4}
\mathbb{E}_{y\sim e}[1] < \infty
\implies
(\BIND e'~\lambda{x}.e)=\mathtt{scale}(e,\mathbb{E}_{y\sim e'}[1])
\quad
(x \notin \mathrm{FV}(e))
\\
\label{eq:MHOL:observe}
(\OBSERVE e_1 \AS e_2) = \mathtt{normalize}(\SCALE (e_1,e_2))
\\
\label{eq:MHOL:normalize1}
\mathtt{normalize}(e) = \SCALE(e,\lambda{u}.1/\mathbb{E}_{x \sim e}[1]) \quad (u \notin \mathrm{FV}(\mathbb{E}_{x \sim e}[1]))
\\
\label{eq:MHOL:normalize3}
0 < \alpha < \infty
\implies \mathtt{normalize}(\mathtt{scale}(e_1, e_2))=\mathtt{normalize}(\mathtt{scale}(e_1, \alpha \ast e_2))
\end{gather}
We could also introduce the axioms for particular distributions such as
$\mathbb{E}_{x \sim \mathtt{Bern}(e)}[\IF x \THEN 1 \ELSE 0] = e$ ($0 \leq e \leq 1$) and
$\mathbb{E}_{x \sim \mathtt{Gauss}(e_1,e_2)}[x]=e_1$, but we do not do this here.
%
\subsection{Markov and Chebyshev Inequalities}
\label{sec:markov_chebyshev}
The axioms we introduced above are not only very basic but also very expressive. 
For instance, we can prove the Markov inequality (\ref{eq:MHOL:lemma:Markov_inequality}) and the Chebyshev inequality (\ref{eq:MHOL:lemma:Chebyshev_inequality}) in \PL\ using these axioms.
\begin{gather}
\label{eq:MHOL:lemma:Markov_inequality}
d \colon M[\REAL],~
a \colon \REAL
\vdash_{\mathrm{\NameOfUnderlyingLogic}}
(a > 0) \implies \Pr_{x \sim d}[ |x| \geq a] \leq \mathbb{E}_{x \sim d}[|x|]/a.\\
\label{eq:MHOL:lemma:Chebyshev_inequality}
\begin{array}{rl}
d \colon M[\REAL],
b \colon \REAL,
\mu \colon \REAL
&\vdash_{\mathrm{\NameOfUnderlyingLogic}}
\mathbb{E}_{x \sim d}[1] = 1 \land \mu = \mathbb{E}_{x \sim d}[x] \land b^2 > 0\\
&\implies \Pr_{x \sim d}[ |x - \mu| \geq b] \leq \mathrm{Var}_{x\sim d}[x]/b^2.
\end{array}
\end{gather}

\section{Unary/Relational Logic}
\label{sec:logic:program}
%

In this section, we introduce two specializations of \PL. The first one, \UPL{}, is
a unary logic to verify \emph{unary properties} of probabilistic programs.
More concretely, \UPL{} can be considered as a collection of inference rules derivable in \PL{}
specialized in proving formulas of form $\phi(e)$ by following the syntactic structure
of the distinguished subterm $e$ rather than the syntactic structure of $\phi$ itself.

The second one, \RPL{}, is a relational logic to verify \emph{relational properties} of probabilistic
programs. Similarly, it can be seen as a collection of inference rules derivable in \PL{}
to prove formulas of the form $\phi(e_1, e_2)$ by following the syntactic structure of
$e_1$ and $e_2$.


\subsection{The Unary Logic \UPL}

Judgments in the unary logic \UPL\ have the shape
$
\Gamma \mid \Psi \vdash_{\mathrm{\NameOfUnaryLogic}} e \colon \tau \mid \phi
$
where $\Gamma$ is a context, $\Psi$ is a set of assertions on the context variables, $e$ is a \PCFP\
expression, $\tau$ a type, and $\phi$ is an assertion (possibly)
containing
a distinguished variable $\mathbf{r}$ of type $\tau$ which is used to refer to the
expression $e$ in the formula $\phi$.

We give in Figure~\ref{fig:unary-pure-rules} a selection of proof rules in
\UPL. We have two groups of rules, one for pure computations
and the other for probabilistic computations. The rules are mostly \emph{syntax-directed}, 
with the exception
of the rule [u-SUB]. 
\begin{figure}
{\fbox{Rules for pure constructions.}}
{\small
\begin{gather*}
{
\AxiomC{$\Gamma \vdash x \colon \tau$ \quad $\Gamma \mid \Psi \vdash_{\mathrm\PL} \phi[x/\mathbf{r}]$}
\RightLabel{\tiny [u-VAR]}
\UnaryInfC{$\Gamma \mid \Psi \vdash_{\mathrm\UPL} x \colon \tau \mid \phi$}
\DisplayProof
}
\qquad 
{
\AxiomC{$\Gamma,x\colon\tau\mid\Psi,\phi'\vdash_{\mathrm\UPL} t \colon \sigma \mid \phi$}
\RightLabel{\tiny [u-ABS]}
\UnaryInfC{$\Gamma\mid\Psi \vdash_{\mathrm\UPL} \lambda{x\colon\tau}.t \colon \tau \to \sigma \mid \forall{x}.\phi' {\implies} \phi[\mathbf{r}x/\mathbf{r}]$}
\DisplayProof
}
\\
{
\AxiomC{$\Gamma\mid\Psi \vdash_{\mathrm\UPL} t \colon \sigma \mid \phi'$ \quad $\Gamma\mid\Psi \vdash_{\mathrm\PL} \phi'[t/\mathbf{r}] \implies \phi[t/\mathbf{r}]$}
\RightLabel{\tiny [u-SUB]}
\UnaryInfC{$\Gamma\mid\Psi \vdash_{\mathrm\UPL} t \colon \sigma \mid \phi$ }
\DisplayProof
}
\\ 
{
\AxiomC{$\Gamma\mid \Psi \vdash_{\mathrm\UPL}t\colon\tau\to\sigma\mid\forall{x}.\phi' {\implies} \phi[\mathbf{r}x/\mathbf{r}]$ \quad $\Gamma\mid\Psi\vdash_{\mathrm\UPL}u\colon\tau\mid\phi' $}
\RightLabel{\tiny [u-APP]}
\UnaryInfC{$\Gamma\mid\Psi\vdash_{\mathrm\UPL}tu\colon\sigma\mid\phi[u/x]$}
\DisplayProof
}
\\
{
\AxiomC{
$
\begin{array}{c@{}}
\mathit{Terminate}(f,x,e)\\
\Gamma, x\colon I, f \colon I \to \sigma \mid \Psi,\phi',\forall{m}.|m| < |x| \implies \phi'[m/x] \implies \phi[m/x][fm/\mathbf{r}]\vdash_{\mathrm\UPL} e \colon \sigma\mid\phi
\end{array}
$}
\RightLabel{\tiny [u-LETREC]}
\UnaryInfC{$\Gamma\mid \Psi \vdash_{\mathrm\UPL}\LETREC fx=e \colon I \to \sigma\mid \forall{x}.\phi'\implies\phi[\mathbf{r}x/\mathbf{r}]$}
\DisplayProof
}
\end{gather*}

{
\fbox{Rules for probabilistic constructions.}
}
\begin{gather*}
{
\RightLabel{\tiny [u-RET]}
\AxiomC{$\Gamma \mid \Psi \vdash_{\mathrm\UPL} e \colon \tau \mid \phi[\RETURN(\mathbf{r})/\mathbf{r}]$}
\UnaryInfC{$\Gamma \mid \Psi \vdash_{\mathrm\UPL} \RETURN(e) \colon M[\tau] \mid \phi$}
\DisplayProof
}
%
\\
%
{
\RightLabel{\tiny [u-BIND]}
\AxiomC{
$\begin{array}{c}
\Gamma \mid \Psi \vdash_{\mathrm\UPL} e \colon M[\tau_1] \mid \phi_1
\quad
\Gamma \mid \Psi \vdash_{\mathrm\UPL} e' \colon \tau_1 \to M[\tau_2] \mid
\forall{s\colon M[\tau_1]}.(\phi_1[s/\mathbf{r}]{\implies}
\phi_2[\BIND s~\mathbf{r}/\mathbf{r}])
\end{array}
$
}
\UnaryInfC{$\Gamma \mid \Psi \vdash_{\mathrm\UPL} \BIND e~e' \colon M[\tau_2] \mid \phi_2$}
\DisplayProof
}
%
\\
{ 
\infer
{
\begin{array}{c}\Gamma \mid \Psi \vdash_{\mathrm\UPL} e \colon M[\tau] \mid \phi_1\quad
\Gamma \mid \Psi \vdash_{\mathrm\UPL} e' \colon \tau \to \mathtt{pReal} \mid
\forall{s\colon M[\tau]}.(\phi_1[s/\mathbf{r}] {\implies} \phi_2[\OBSERVE s \AS \mathbf{r}/\mathbf{r}])
\end{array}
} {\Gamma \mid \Psi \vdash_{\mathrm\UPL} \OBSERVE e \AS e' \colon  M[\tau] \mid
  \phi_2}
{\text{\tiny [u-QRY]}}
}
\end{gather*}
}
\caption{A selection of \UPL\ rules. }
\label{fig:unary-pure-rules}
\end{figure}
We present a selection of the pure rules, the rest of them are as in UHOL
\citep{Aguirre:2017:RLH:3136534.3110265}. 
The rule [u-ABS] turns an assertion about the bound variable into a precondition
of its lambda abstraction. The rule [u-APP] proves a postcondition of a function
application provided that the argument satisfies the precondition of the function.
The rule [u-LETREC] allows proving properties
of terminating recursive functions by introducing an induction hypothesis in the context.

In the case of monadic computations, we have rules for monadic return, bind and $\OBSERVE$. 
It is worth noticing that in the second premise of both
the rules [u-BIND] and [u-QRY], the assertion quantifies over elements
in $M[\tau_1]$, while the input type of the function is just
$\tau_1$. This follows the spirit of the interpretation (see Section~\ref{sec:semantics}), where the
Kleisli lifting $(-)^\#$ is used to lift a function $\tau_1\to M[\tau_2]$ to a
function $M[\tau_1]\to M[\tau_2]$. The quantification over
distributions, rather than over elements, is essential to establish a
connection with the assertion on the first premise. This will be
useful to simplify the verification of our examples. 

We can prove that, despite being syntax directed, \UPL{} does not lose expressiveness relative to
\PL: The following theorem shows that the unary logic \UPL{} is sound and complete with respect to the underlying logic \PL.
\begin{theorem}[Equi-derivability of \PL{} and \UPL]
\label{thm:equiPLUPL}
The judgment 
$\Gamma\mid\Psi\vdash_{\mathrm{\NameOfUnderlyingLogic}} \phi[ e /
\mathbf{r}]$ is derivable if and only if
the judgment $\Gamma\mid\Psi\vdash_{\mathrm{\NameOfUnaryLogic}} e \colon\tau \mid
\phi$ is derivable.
\end{theorem}

\subsection{The Relational Logic \RPL}
Judgments in the relational logic \RPL\ have the shape
$
\Gamma \mid \Psi \vdash_{\mathrm{\NameOfUnaryLogic}} e_1 \colon \tau_1 \sim e_2 \colon \tau_2 \mid \phi
$,
where $\Gamma$ is a context, $\Psi$ is a set of assertions on the context, $e_1$ and $e_2$ are  \PCFP\
expressions, $\tau_1$ and $\tau_2$ are types, and $\phi$ is an assertion (possibly)
containing
two distinguished variables $\mathbf{r}_1$ of type $\tau_1$ and
$\mathbf{r}_2$ of type $\tau_2$ which are used to refer to the
expressions $e_1$ and $e_2$ in the formula $\phi$.
We give in Figure~\ref{fig:RPL-rules} a selection of proof rules in
\RPL. We present three groups of rules. The first group consists of
relational rules for pure computations. The second group consists of
two-sided relational rules for probabilistic computations,
meaning that the terms on both sides of the judgment have the same top-level
constructor. Finally, the third group consists of one-sided relational rules for
probabilistic computations, meaning that one of terms has
a specific top-level constructor while the other is arbitrary (not analyzed by the rule).
Here we show just the left-sided rules that have the constructor on the left;
right-sided rules are symmetrical. 
As in the unary case, we use an approach that is mostly \emph{syntax-directed} except for
the [r-SUB] rule.
\begin{figure}
{ 
\fbox{Relational rules for pure constructions - two-sided}
}
{
\begin{gather*}
{ 
\RightLabel{\tiny [r-ABS]}
\AxiomC{
$\Gamma,x_1 \colon \tau_1,x_2\colon \tau_1 \mid \Psi,\phi' \vdash_{\mathrm\RPL}t_1 \colon \sigma_1 \sim t_2 \colon \sigma_2 \mid \phi $
}
\UnaryInfC{
$\Gamma \mid \Psi \vdash_{\mathrm\RPL} \lambda{x_1}.t_1 \colon \tau_1 \to \sigma_1 \sim \lambda{x_2}.t_2 \colon \tau_2 \to \sigma_2 \mid \forall{x_1}.\forall{x_2}\phi'\implies \phi[\mathbf{r}_1x_1/\mathbf{r}_1,\mathbf{r}_2x_2/\mathbf{r}_2] $}
\DisplayProof
}
\\
{ 
\RightLabel{\tiny [r-APP]}
\AxiomC{
$\begin{array}{c@{}}
\Gamma \mid \Psi \vdash_{\mathrm\RPL}u_1 \colon \tau_1 \sim u_2\colon \tau_2 \mid \phi'
\\
\Gamma \mid \Psi \vdash_{\mathrm\RPL}t_1 \colon \sigma_1 \to \tau_1 \sim t_2 \colon \sigma_2 \to \tau_2 \mid
\forall{x_1}.\forall{x_2}.\phi'[x_1/\mathbf{r}_1,x_2/\mathbf{r}_2] {\implies} \phi[\mathbf{r}_1x_1/\mathbf{r}_1,\mathbf{r}_2x_2/\mathbf{r}_2] 
\end{array}
$
}
\UnaryInfC{
$\Gamma \mid \Psi \vdash_{\mathrm\RPL}t_1 u_1 \colon \sigma_1 \sim t_2 u_2\colon \sigma_2 \mid \phi $}
\DisplayProof
}
\\
{ 
\RightLabel{\tiny [r-SUB]}
\AxiomC{
$
\Gamma \mid \Psi \vdash_{\mathrm\RPL} t_1 \colon \sigma_1 \sim t_2 \colon \sigma_2  \mid \phi'
$
\quad
$
\Gamma \mid \Psi \vdash_{\mathrm\PL} \phi'[t_1/\mathbf{r}_1, t_2/\mathbf{r}_2] \implies \phi[t_1/\mathbf{r}_1, t_2/\mathbf{r}_2] 
$
}
\UnaryInfC{$\Gamma \mid \Psi \vdash_{\mathrm\RPL} t_1 \colon \sigma_1 \sim t_2 \colon \sigma_2  \mid \phi$}
\DisplayProof
}
\end{gather*}
}

{ \fbox{Relational rules for probabilistic constructions - two-sided}}
{\small
\begin{gather*}
{ 
\RightLabel{\tiny [r-RET]}
\AxiomC{$\Gamma \mid \Psi \vdash_{\mathrm\NameOfRelationalLogic} e_1 \colon \tau_1 \sim  e_2 \colon \tau_2 \mid \phi[\RETURN(\mathbf{r}_1)/\mathbf{r}_1,~\RETURN(\mathbf{r}_2)/\mathbf{r}_2]$}
\UnaryInfC{$\Gamma \mid \Psi \vdash_{\mathrm\NameOfRelationalLogic} \ \RETURN(e_1) \colon M[\tau_1]\sim\RETURN(e_2) \colon M[\tau_2] \mid \phi$}
\DisplayProof
}
\\
{ 
\RightLabel{\tiny [r-BIND]}
\AxiomC{
$
\begin{array}{c@{}}
\phi=\forall{s_1\colon M[\tau_1]}.\forall{s_2\colon M[\tau_2]}.~(\phi_1[s_1/\mathbf{r}_1,s_2/\mathbf{r}_2]\implies \phi_2[\BIND s_1~\mathbf{r}_1/\mathbf{r}_1, \BIND s_2~\mathbf{r}_2/\mathbf{r}_2])\\
\Gamma \mid \Psi \vdash_{\mathrm\NameOfRelationalLogic} e_1 \colon M[\tau_1] \sim e_2 \colon M[\tau_2] \mid \phi_1\quad
\Gamma \mid \Psi \vdash_{\mathrm\NameOfRelationalLogic}
e'_1 \colon \tau_1 \to M[\tau_3]
\sim
e'_2 \colon \tau_2 \to M[\tau_4]
\mid\phi
\end{array}
$
}
\UnaryInfC{$\Gamma \mid \Psi \vdash_{\mathrm\NameOfRelationalLogic} \BIND e_1~e'_1 \colon M[\tau_3] \sim \BIND e_2~e'_2  \colon  M[\tau_4] \mid \phi_2$}
\DisplayProof
}
\\
%
{ 
\RightLabel{\tiny [r-QRY]}
\AxiomC{
$
\begin{array}{c@{}}
\phi=\forall{s_1\colon M[\tau_1]},{s_2\colon M[\tau_2]}.\phi_1[s_1/\mathbf{r}_1,s_2/\mathbf{r}_2] 
\implies \phi_2[\OBSERVE s_1 \AS \mathbf{r}_1/\mathbf{r}_1,\OBSERVE s_2 \AS \mathbf{r}_2/\mathbf{r}_2]
\\
\Gamma \mid \Psi \vdash_{\mathrm\NameOfRelationalLogic} e_1 \colon M[\tau_1] \sim e_2 \colon M[\tau_2] \mid \phi_1\quad
\Gamma \mid \Psi \vdash_{\mathrm\NameOfRelationalLogic}
e'_1 \colon \tau_1 \to \mathtt{pReal} \sim e'_2 \colon \tau_2 \to
  \mathtt{pReal} \mid \phi
\end{array}
$
}
\UnaryInfC{$\Gamma \mid \Psi \vdash_{\mathrm\NameOfRelationalLogic} \OBSERVE e_1 \AS e'_1 \colon  M[\tau_1]
\sim  \OBSERVE e_2 \AS e'_2 \colon  M[\tau_2] \mid \phi_2$}
\DisplayProof
}
\end{gather*}
}

{ \fbox{Relational rules for probabilistic constructions - one-sided}}
{\small
\begin{gather*}
{ 
\RightLabel{\tiny [r-RET-L]}
\AxiomC{$\Gamma \mid \Psi \vdash_{\mathrm\NameOfRelationalLogic} e_1 \colon \tau_1 \sim  e_2 \colon \tau_2 \mid \phi[\RETURN(\mathbf{r}_1)/\mathbf{r}_1]$}
\UnaryInfC{$\Gamma \mid \Psi \vdash_{\mathrm\NameOfRelationalLogic} \ \RETURN(e_1) \colon M[\tau_1]\sim e_2 \colon \tau_2 \mid \phi$}
\DisplayProof
}
\\
{ 
\RightLabel{\tiny [r-BIND-L]}
\AxiomC{
$
\begin{array}{c@{}}
\Gamma \mid \Psi \vdash_{\mathrm\NameOfUnaryLogic} e_1 \colon M[\tau_1] \mid \phi_1\\
\Gamma \mid \Psi \vdash_{\mathrm\NameOfRelationalLogic}
e'_1 \colon \tau_1 \to M[\tau_3] \sim e_2 \colon M[\tau_2]
\mid\forall{s_1\colon M[\tau_1]}.\phi_1[s_1/\mathbf{r}]\implies \phi_2[\BIND s_1\IN \mathbf{r}_1/\mathbf{r}_1]
\end{array}
$
}
\UnaryInfC{$\Gamma \mid \Psi \vdash_{\mathrm\NameOfRelationalLogic} \BIND e_1~e'_1  \colon M[\tau_3] \sim e_2\colon  M[\tau_2] \mid \phi_2$}
\DisplayProof
}
\\
{ 
\RightLabel{\tiny [r-QRY-L]}
\AxiomC{
$
\begin{array}{c@{}}
\Gamma \mid \Psi \vdash_{\mathrm\UPL} e_1 \colon M[\tau_1] \mid \phi_1\\
\Gamma \mid \Psi \vdash_{\mathrm\RPL}
e'_1 \colon \tau_1 \to \mathtt{pReal} \sim e_2 \colon M[\tau_2] \mid
  \forall{s_1\colon M[\tau_1]}.\phi_1[s_1/\mathbf{r}_1] \implies \phi_2[\OBSERVE s_1 \AS \mathbf{r}_1/\mathbf{r}_1]
\end{array}
$
}
\UnaryInfC{$\Gamma \mid \Psi \vdash_{\mathrm\NameOfRelationalLogic} \OBSERVE e_1 \AS e'_1 \colon  M[\tau_1]
\sim  e_2 \mid \phi_2$}
\DisplayProof
}
\end{gather*}
}
\caption{A selection of \RPL\ rules.}
\label{fig:RPL-rules}
\end{figure}

The rules for pure computations are similar to the ones from RHOL~\citep{Aguirre:2017:RLH:3136534.3110265}
and we present only a selection. 
For the probabilistic constructions, we have relational rules for the monadic
return and bind, and for $\OBSERVE$. These rules are the natural
generalization of the unary rules to the relational case. 
In particular, in all the rules for $\BIND$ and $\OBSERVE$ we use assertions
quantifying over distributions, similarly to what we have in \UPL, to establish a
connection between the different assertions. 

The equi-derivability result for \UPL{} can be lifted to the relational setting: \RPL{} is also sound and
complete with respect to the logic \PL.
\begin{theorem}[Equi-derivability of \PL{} and \RPL]
\label{thm:equivPLRPL}
The judgment $\Gamma\mid\Psi\vdash_{\mathrm{\NameOfUnderlyingLogic}} \phi[ e_1/\mathbf{r}_1,~e_2/\mathbf{r}_2]$ is derivable if and only if $\Gamma\mid\Psi\vdash_{\mathrm{\NameOfRelationalLogic}} e_1\colon\tau_1 \sim e_2\colon\tau_2 \mid \phi$ is derivable.
\end{theorem}

\paragraph{A comment on product types and \RPL{} }
One could effectively embed the whole of \RPL{} into \UPL{} by replicating the
set of rules of \RPL{} as \UPL{} rules for every possible product type, and
rewriting the distinguished $\mathbf{r}_1, \mathbf{r}_2$ in the refinements to
$\pi_1(\mathbf{r}),\pi_2(\mathbf{r})$.  For instance, the two-sided [r-ABS]
rule would be rewritten as a \UPL{} rule [u-ABSxABS] for the product of two
arrow types. Similarly, all the one-sided rules would be written as unary rules
directed by only one side of the product type, while ignoring the other.

However, we believe that this style of presentation would be significantly more
cumbersome, and, moreover, it would hide the fact that we are trying to express
and prove a relational property of two programs that execute independently.

\subsection{Special Rules}
As already discussed in the introduction, we enrich \THESYSTEM\ with special rules that can ease verification. One example is the  use  the following Bayesian law
expressing a general fact about the way we can reason about
probabilistic inference when the observation is a boolean:
\[
{ 
\RightLabel{[Bayes]}
\AxiomC{
$\Gamma, x\colon \tau \vdash e' \colon \mathtt{bool}$
\quad
$\Gamma, x \colon \tau \vdash e'' \colon \mathtt{bool}$
\quad
$\Gamma \vdash e \colon M[\tau]$
}
\UnaryInfC{
$
\Gamma\mid\Psi\vdash_{\mathrm{\NameOfUnaryLogic}}
\OBSERVE ~ e  \AS \lambda{x}.(\IF e' \THEN 1 \ELSE 0)
\colon M[\tau] \mid
\Pr_{y \sim \mathbf{r}}[e''[y/x]]
=
\frac{\Pr_{x \sim e}[e' \land e'' ]}{\Pr_{x \sim e}[e']}
$}
\DisplayProof
}
\]
This rule can be derived by first using [u-QRY],
and then reasoning in \PL\ through the [u-SUB] rule, which is
why the premises are just simply typed assumptions.
In particular, in \PL\ we use the characterization of $\OBSERVE$ given in Section \ref{sec:axioms}.

We also introduce a [LET] rule, which can be derived by desugaring the \texttt{let} notation:
\[
{ 
\RightLabel{[LET]}
\AxiomC{
$\begin{array}{c@{}}
\Gamma \mid \Psi \vdash e \colon \tau_1 \mid \phi_1
\quad
\Gamma, x:\tau_1 \mid \Psi, \phi_1[x/\mathbf{r}] \vdash e' \colon \tau_2 \mid
   \phi_2
\end{array}
$
}
\UnaryInfC{$\Gamma \mid \Psi \vdash \LET~x= e~\IN~e' \colon \tau_2 \mid \phi_2$}
\DisplayProof
}
\]

Notice that Theorem~\ref{thm:equiPLUPL} can be used to convert \UPL{} derivation trees into \PL{} ones and vice versa. 
Similarly, Theorem \ref{thm:equivPLRPL} is used to convert \RPL{} proofs to \PL proofs.
These conversions are useful to switch between the different levels of our system
and to reason in whichever one is more convenient.
To this end, we introduce the following \emph{admissible} rules:
\[
{ 
\AxiomC{$\Gamma\mid\Psi\vdash_{\mathrm{\NameOfUnaryLogic}} e \colon\tau \mid \phi$}
\RightLabel{[conv-\NameOfUnaryLogic]}
\UnaryInfC{$\Gamma\mid\Psi\vdash_{\mathrm{\NameOfUnderlyingLogic}} \phi[ e / \mathbf{r}]$}
\DisplayProof
}
\qquad
{ 
\AxiomC{$\Gamma\mid\Psi\vdash_{\mathrm{\NameOfRelationalLogic}} e_1\colon\tau_1 \sim e_2\colon\tau_2 \mid \phi$}
\RightLabel{[conv-\NameOfRelationalLogic]}
\UnaryInfC{$\Gamma\mid\Psi\vdash_{\mathrm{\NameOfUnderlyingLogic}} \phi[ e_1/\mathbf{r}_1,~e_2/\mathbf{r}_2]$}
\DisplayProof
}
\]

\section{Semantics}
\label{sec:semantics}
\subsection{Background}
In this section we present the semantic foundation of \THESYSTEM.
We start by recalling the definition of
quasi-Borel spaces~\citep{HeunenKSY17} and by showing how we can use
them to define monads for probability measures~\citep{Scibior:2017:DVH:3177123.3158148}. We use these
constructions in the next section to give the
semantics of programs on which we will build our logic.
\paragraph{Quasi-Borel Spaces}
We introduce here the category $\QBS$ of \emph{quasi-Borel spaces}.
Intuitively, the category $\QBS$ is a relaxation of the category
$\Meas$ of measurable spaces. $\QBS$ has a nice categorical structure---it is  Cartesian closed and retains important properties
coming from measure theory.
%
Before introducing quasi-Borel spaces, we fix some notation.
We use $\mathbb{R}$ to denote the real line equipped with the standard Borel algebra.
We use 
$\coprod_{i \in \mathbb{N}} S_i$ to denote the coproduct of a countable family
of sets $\{S_i\}_{i\in \mathbb{N}}$, and  
$[\alpha_i]_{i \in \mathbb{N}}$ for the copairing of functions
$\alpha_i$ for $i\in\mathbb{N}$.

\begin{definition}[{\citet{HeunenKSY17}}]
The category $\QBS$ is the category of quasi-Borel spaces and morphisms
between them, where a  \emph{quasi-Borel space} $(X,M_X)$
(with respect to $\mathbb{R}$)
is a set $X$ equipped with a subset $M_X$ of functions in $\mathbb{R} \to X$ such that
(1) If $\alpha \colon \mathbb{R} \to X$ is constant then $\alpha \in M_X$.
(2) If $\alpha \in M_X$ and $f \colon \mathbb{R} \to \mathbb{R}$ is measurable then $\alpha \circ f \in M_X$.
(3) If the family $\{S_i\}_{i\in \mathbb{N}}$ is a countable partition of $\mathbb{R}$, i.e. $\mathbb{R}=\coprod_{i \in \mathbb{N}} S_i$, with each set $S_i$ Borel, and if $\alpha_i \in M_X$ ($\forall i\in\mathbb{N}$) then the copairing $[\alpha_i|_{S_i}]_{i \in \mathbb{N}}$ of $\alpha_i|_{S_i} \colon S_i \to X$ belongs to $M_X$.

A  morphism from a quasi-Borel space $(X,M_X)$ to a quasi-Borel space $(Y,M_Y)$ is a function $f \colon X \to Y$ such that 
$f \circ \alpha \in M_Y$ holds for any $\alpha \in M_X$.
\end{definition}

As shown by~\citet{HeunenKSY17}, the category $\QBS$ has a convenient
structure to interpret probabilistic programs. That is, it is
well-pointed and Cartesian closed and we have the usual structure for currying and uncurrying functions; it has products and coproducts with
distributivity between them;
every standard Borel space $\Omega$ can be converted to a quasi-Borel
space and every measurable function $f\colon\Omega_1\to\Omega_2$ is  
a morphism $f\colon\Omega_1\to\Omega_2$ in $\QBS$.
Hence, a $\QBS$ can be used to interpret a probabilistic functional language. See
\citet{HeunenKSY17,Scibior:2017:DVH:3177123.3158148} for more details.

The category $\QBS$ has also a convenient
structure to \emph{reason} about probabilistic programs.
In particular, the forgetful functor $|{-}| \colon \QBS \to \Sets$
erasing the quasi-Borel structure does not change the underlying structure of functions. This property is fundamental for the design of the category $\Pred(\QBS)$ of predicates over quasi-Borel spaces.
\paragraph{Measures on quasi-Borel spaces.}
Quasi-Borel spaces were introduced to support measure theory in a
Cartesian closed category. In particular, given a measure on some
standard Borel space $\Omega$ we can define a measure over
quasi-Borel spaces.

\begin{definition}[{\citet{Scibior:2017:DVH:3177123.3158148}}]
A measure on a quasi-Borel space $(X,M_X)$ is a triple
$(\Omega,\alpha,\nu)$ where $\Omega$ is a standard Borel space,
$\alpha \colon \Omega \to X$ is a morphism  in $\QBS$, and $\nu$  is a $\sigma$-finite measure over $\Omega$.
\end{definition}
For a measure $\mu = (\Omega,\alpha,\nu)$ on $X$ and a
function $f:X \to \mathbb{R}$ in $\QBS$,
we define  integration over quasi-Borel spaces in terms of integration
over Borel spaces:
$\int_X f~d{\mu} \defeq \int_\Omega (f \circ \alpha)~d\nu$.
Equivalence of measures in $\QBS$ is defined in terms of equality of
integrations:
\begin{center}
$
(\Omega,\alpha,\nu) \approx (\Omega',\alpha',\nu')\defeq \forall{f: X \to \mathbb{R} \text{ in } \QBS}.~ \int_\Omega (f \circ \alpha)~d\nu = \int_{\Omega'} (f \circ \alpha')~d\nu'.
$
\end{center}
In the following, it will be convenient to work with equivalence
classes of measures which we denote by $[\Omega,\alpha,\nu]$.
Every equivalence class for a measure $(\Omega,\alpha,\nu)$ also contains a measure over $\mathbb{R}$ defined in the appropriate way \citep{HeunenKSY17}. 
%
%
We are now ready to
define a monad for measures. 
\begin{definition}[\citet{Scibior:2017:DVH:3177123.3158148}]
The monad of $\sigma$-finite measures $\mathfrak{M}$ is defined as follows.
\begin{itemize}
\item For any $X$ in $\QBS$, $\mathfrak{M}X$ is the set of equivalence
  classes of $\sigma$-finite measures equipped with the quasi-Borel
  structure given by the following definition
\[
M_{\mathfrak{M}X} = \Set{\lambda r.[D_r,\alpha(r,-),\mu_r] | \begin{aligned}
&D\subseteq_{\text{measurable}} \mathbb{R} \times \Omega,~ \mu  \colon \text{$\sigma$-finite measure on } \Omega, \\
& \alpha \colon D \to X,~ D_r = \Set{\omega | (r,\omega) \in D},~ \mu_r = \mu|_{D_r}
\end{aligned}
}
\]
\item The unit $\eta_X \colon X \to \mathfrak{M}X$ is defined by $\eta_X(x) = [1, \lambda \ast.~x, \mathbf{d}_\ast]$.
\item The Kleisli lifting is defined for any $f \colon X \to \mathfrak{M}Y$ and $[\Omega,\alpha,\nu] \in \mathfrak{M}X$ as
\[
f^\sharp [\Omega,\alpha,\nu] = [D,\beta,(\nu \otimes \nu')|_D]
\]
where $D=\{(r,\omega)\ |\ \omega\in D_r\}$ and $\beta(-)=\lambda
r.\beta(r,-)$ are defined for every $\gamma:\Omega\to \mathbb{R}$ and
$\gamma^*:\mathbb{R}\to \Omega$ satisfying $\gamma^\ast \circ \gamma
= \mathrm{id}_\Omega$ through $(f \circ \alpha)(\gamma^\ast(r)) =
[D_r,\beta(r,-),\nu']$.
\end{itemize}
\end{definition}
Let us unpack in part this definition. The set of functions
$M_{\mathfrak{M}X}$ can be seen as a set of (uncountable) families of measures,
indexed by $r$, supporting infinite measures. The Kleisli lifting uses
the fact that each $(f \circ \alpha)(\gamma^\ast(-))$ is a function in
$M_{\mathfrak{M}Y}$, that $D$ built as a product measure starting from
$r$ and $D_r$ is measurable, and $\beta$ is a morphism from $D$ to $Y$.

Thanks to the Fubini-Tonelli theorem, the monad $\mathfrak{M}$ on $\QBS$
is commutative strong with respect to Cartesian products. 
We can also use the structure of $\QBS$ to define the product measure
of $[\Omega,\alpha,\nu]$ and $[\Omega',\alpha',\nu']$ as $[(\Omega\times\Omega'),(\alpha\times\alpha'),(\nu\otimes\nu')]$.
Using the isomorphism $\mathfrak{M}1 \cong [0,\infty]$, usual integration $\int f ~d\mu$ for $f \colon \mathbb{R} \to [0,\infty]$ and $\mu \in \mathfrak{M}(\mathbb{R})$ corresponds to $f^\sharp(\mu)$.
We can define the \emph{mass} $|\mu|$ of measure $\mu = [\Omega,\alpha,\nu]$
by $\int_X 1 d\mu$ which is the same as the mass $|\nu|$ of base measure $\nu$.
The monad $\mathfrak{M}$ captures general measures. For example, we can
define a \emph{null measure} as $\mathbf{0}=[\Omega,\alpha,0]$.

In the sequel, we will also use a commutative monad $\mathfrak{P}$ on
$\QBS$ obtained by restricting the monad $\mathfrak{M}$ to \emph{subprobability} measures.
We have the canonical inclusion $\mathfrak{P}X \subseteq \mathfrak{M}X$.
\subsection{Semantics for \THESYSTEM}
In order to give meaning to the logical formulas of \PL, we first need
to give meaning to expressions in \PCFP\ and to enriched expressions
in \PL. We do this by interpreting types as $\QBS$ objects as shown below:
\[
\begin{array}{c@{}}
\interpret{\mathtt{unit}} \defeq 1,\quad
\interpret{\mathtt{bool}} \defeq 1+1,\quad
\interpret{\INT} \defeq \mathbb{N},\quad
\interpret{\REAL} \defeq \mathbb{R},\quad
\interpret{\mathtt{pReal}} \defeq [0,\infty],\\
\interpret{\tau_1 \to \tau_2}  \defeq \interpret{\tau_1}{\Rightarrow}\interpret{\tau_2},
\interpret{\tau_1 \times \tau_2}  \defeq \interpret{\tau_1} \times \interpret{\tau_2},
\interpret{\mathtt{list}(\tau)} \defeq \coprod_{n \in \mathbb{N}}\!\! \interpret{\tau}^n,
\interpret{M[\tau]} \defeq \mathfrak{M}(\interpret{\tau})
\end{array}
\]
where $1$ is the terminal object in $\QBS$;
$\coprod_{n \in \mathbb{N}} \interpret{\tau}^n$ is the coproduct of the countable family $\interpret{\tau}^n = \interpret{\tau} \times \cdots \times \interpret{\tau}$ ($n$ times);
$(\interpret{\tau_1} \Rightarrow \interpret{\tau_2})$ is the exponential object in $\QBS$.
%
We interpret each term $\Gamma \vdash e \colon \tau$ as a
morphism $\interpret{\Gamma} \to \interpret{\tau}$ in $\QBS$, where, as
usual, the interpretation $\interpret{\Gamma}$ of a context $\Gamma$ is
the product of the interpretations of its components.
Pure computations are interpreted using the Cartesian closed structure of $\QBS$
%
%
%
%
%
where we can interpret recursive terms based on recursive data types ($I = \mathtt{list}(\tau),\INT$) by means of a fixed point operator $\mathrm{fix}$ iterating functions until termination.
Since the termination criterion $\mathrm{Terminate}(f,x,e)$ ensures that all recursive calls are on smaller arguments, the operator $\mathrm{fix}$ is well-defined: for any $n \in I$ with $|n| < k$, $\mathrm{fix}(\lambda f.\lambda x.e)(n)$ is defined within $k$ steps.
\[
\interpret{\Gamma \vdash \mathtt{letrec}~f x = e \colon I \to \sigma}
\defeq \mathrm{fix}(\interpret{\Gamma \vdash \lambda f\colon I \to \sigma. \lambda x \colon I.~ e \colon (I \to \sigma) \to (I \to \sigma)})
\]

We interpret $\mathtt{return}$ and
$\mathtt{bind}$ using the structure of the monad
$\mathfrak{M}$ of measures on $\QBS$.
\[
\begin{array}{rl}
\interpret{\Gamma \vdash \mathtt{return}~e \colon M[\tau]}
&\defeq \eta_{\interpret{\tau}} \circ \interpret{\Gamma \vdash e \colon \tau}\\
\interpret{\Gamma \vdash \mathtt{bind}~e_1~e_2 \colon M[\tau_2]}
&\defeq \interpret{\Gamma \vdash e_2 \colon \tau_1 \to M[\tau_2] }^\sharp \circ 
\mathrm{st}_{\interpret{\Gamma},\interpret{\tau_1}}(\lrangle{\mathrm{id}_{\interpret{\Gamma}},\interpret{\Gamma \vdash e_1 \colon M[\tau_1]}})
\end{array}
\]
where $\eta$, $(-)^\sharp$, and $\mathrm{st}$ are the unit, the Kleisli lifting, and the tensorial strength of the commutative monad $\mathfrak{M}$.
To interpret the other constructions we first introduce two semantics
constructions for scaling and normalizing:
\[
\mathrm{scale}(\nu,f) \defeq (\mathfrak{M}(\pi_2) \circ \mathrm{dst}_{1,X}\circ\lrangle{f,\eta_X})^\sharp(\nu).
\qquad
\mathrm{normalize}(\nu) \defeq
\begin{cases}
\mathbf{0} & |\nu| = 0, \infty\\
\nu/|\nu| & \text{ (otherwise) }
\end{cases}
.
\]
where $\mathrm{dst}$ is the double strength of the
commutative monad $\mathfrak{M}$, and $|\nu|$ is the mass of
$\nu$. 
In the definition of $\mathrm{scale}(\nu,f)$, the construction 
$\mathfrak{M}(\pi_2) \circ \mathrm{dst}_{1,X}\circ\lrangle{f,\eta_X}$
corresponds to a function mapping an element $x\in X$ to  a Dirac distribution centered at $x$ and scaled by $f(x)$, whose domain is then lifted to measures using the Kleisli lifting.  To achieve this, we use the equivalence $\interpret{\mathtt{pReal}} =
[0,\infty] \cong \mathfrak{M}1$, and pairing and projection constructions to manage the duplication of $x$. The definition of $\mathrm{scale}(\nu)$ is more straightforward and reflects the semantics we described before.

Using these constructions we can interpret
the corresponding syntactic constructions.
\[
\begin{array}{rl}
\interpret{\Gamma\vdash\mathtt{scale}(t,t')\colon M[\tau]}&
\defeq \mathrm{scale}(\interpret{\Gamma \vdash t \colon M[\tau]}, \interpret{\Gamma \vdash t' \colon \tau \to \mathtt{pReal}})\\
\interpret{\Gamma\vdash\mathtt{normalize}(t) \colon M[\tau]} &\defeq \mathrm{normalize}(\interpret{\Gamma \vdash t \colon M[\tau]})
\end{array}
\]
We can now interpret $\OBSERVE$ as follows:
\[
\interpret{\Gamma\vdash\OBSERVE \ e~\AS~e' \colon M[\tau]}
\defeq \mathrm{normalize}(\mathrm{scale}(\interpret{\Gamma \vdash e \colon M[\tau]},\interpret{\Gamma \vdash e' \colon \tau \to \mathtt{pReal}}))
\]

%
%
Using the equivalence $\interpret{\mathtt{pReal}} =
[0,\infty] \cong \mathfrak{M}1$ again, we interpret expectation as:
\[
\interpret{\Gamma \vdash \mathbb{E}_{x\sim t}[t'(x)] \colon \mathtt{pReal}}
\defeq
\lambda{\gamma \in \interpret{\Gamma}}.~\big (\interpret{\Gamma \vdash t' \colon \tau \to
  \mathtt{pReal}} (\gamma)\big )^\sharp(\interpret{\Gamma \vdash t \colon
  M[\tau]} (\gamma)).
\]
The primitives of basic probability distributions $\mathtt{Uniform}$, $\mathtt{Bern}$, $\mathtt{Gauss}$ are interpreted by rescaling a measure (given as a constant) with density functions (cf. Section \ref{subsection:example:NormalLearning}), and the usual operations on real numbers are given by embedding measurable real functions in $\QBS$.
%

To interpret formulas in \PL\ we use the category $\Pred(\QBS)$ of predicates on quasi-Borel spaces. This will be useful to see these formulas as assertions in the unary logic \UPL\ and the relational logic \RPL. This is actually the main reason why we use quasi-Borel spaces: we want an assertion logic whose predicates support both higher-order computations and continuous probability.
The structure of the category $\Pred(\QBS)$ is the following:
\begin{itemize}
\item An object is a pair $(X,P)$ where $X \in \QBS$ and $P \subseteq X$.
\item A morphism  $f\colon(X,P)\to (Y,Q)$ is $f \colon X \to Y\in\QBS$ such that $\forall x\in P. f(x) \in Q$.
\end{itemize}
An important property of this category is that \emph{every arbitrary subset}
$P$ of a quasi-Borel space $X$ forms an object $(X,P)$ in $\Pred(\QBS)$.
This allows us to interpret all logical operations, including universal quantifiers, in a set-theoretic way.

Notice that this category can be seen as the total category
of the fibration $q\colon\Pred(\QBS) \to \QBS$ given by the following change-of-base
of the fibration $p \colon \Pred \to \Sets$ along the forgetful functor $|-| \colon \QBS \to \Sets$.
Here $\Pred$ is the category of predicates and all predicate-preserving maps, and
the fibration $p$ extracts underlying sets of predicates.
Then, all fibrewise properties of the fibration $p$ are inherited by the fibration $q$.
For detail, see \citep[Section 1.5--1.8]{JacobsCLTT}.

We are now ready to interpret formulas in \PL.
We interpret a typed formula $\Gamma \vdash \phi\;\mathsf{wf}$ as an object $\interpret{\Gamma \vdash \phi\;\mathsf{wf}} \defeq (\interpret{\Gamma},\Rinterpret{\Gamma \vdash \phi\;\mathsf{wf}})$ in $\Pred(\QBS)$
where the predicate part $\Rinterpret{\Gamma \vdash \phi\;\mathsf{wf}}$ is interpreted inductively. We give here a selection of the inductive rules defining the interpretation:
\[
\begin{array}{c}
\Rinterpret{\Gamma \vdash \top\;\mathsf{wf}}
\defeq \interpret{\Gamma},
\quad 
\Rinterpret{\Gamma \vdash \forall{x \colon \tau } \phi\;\mathsf{wf}}
\defeq {\textstyle\bigcap_{y \in \interpret{\tau}}} \Set{\gamma \in \interpret{\Gamma} | (\gamma,y) \in \Rinterpret{\Gamma,x \colon \tau \vdash \phi\;\mathsf{wf}}},
\\
\Rinterpret{\Gamma \vdash \bot\;\mathsf{wf}}
\defeq \emptyset,
\quad
\Rinterpret{\Gamma \vdash t_1 = t_2 \;\mathsf{wf}}
\defeq \Set{\gamma \in \interpret{\Gamma} | \interpret{\Gamma\vdash t_1 \colon \tau}(\gamma) =  \interpret{\Gamma\vdash t_2 \colon \tau}(\gamma)},
\\
\Rinterpret{\Gamma \vdash \phi_1 \land \phi_2\;\mathsf{wf}}
\defeq \Rinterpret{\Gamma \vdash \phi_1\;\mathsf{wf}} \cap \Rinterpret{\Gamma \vdash \phi_2\;\mathsf{wf}},
\qquad
\Rinterpret{\Gamma \vdash \neg \phi\;\mathsf{wf}}
\defeq \interpret{\Gamma} \setminus \Rinterpret{\Gamma \vdash \phi\;\mathsf{wf}},
\end{array}
\]
This interpretation is well-behaved with respect to substitution. In particular, the substitution $\phi[t/x]$ of $x$ by an enriched expression $t$ can be interpreted by the inverse image
$
\Rinterpret{\Gamma \vdash \phi[t/x] \;\mathsf{wf}}
= \inverse{\lrangle{\mathrm{id}_{\interpret{\Gamma}}, \interpret{\Gamma \vdash t \colon \tau}}}\Rinterpret{\Gamma,x\colon\tau \vdash \phi \;\mathsf{wf}}.
$
Using this property, we can show that the logic \PL\ is sound with respect to the semantics that we defined above.
\begin{theorem}[\PL\ Soundness]
\label{thm:PLsound}
If a judgment  $\Gamma \mid \Psi \vdash_{\mathrm{\PL}} \phi$ is derivable then we have the inclusion
$(\bigcap_{\psi \in \Psi} \Rinterpret{\Gamma \vdash \psi\;\mathsf{wf}}) \subseteq \Rinterpret{\Gamma \vdash \phi\;\mathsf{wf}}$ of predicates, which is equivalent to having a morphism
$
\mathrm{id}_{\interpret{\Gamma}} \colon \interpret{\Gamma \vdash {\textstyle\bigwedge_{\psi \in \Psi}} \psi\;\mathsf{wf}}
 \to \interpret{\Gamma \vdash \phi\;\mathsf{wf}}
$ in the category $\Pred(\QBS)$.
\end{theorem}
Here, the soundness of \PL\ axioms introduced in Section \ref{sec:axioms} is proved from the basic facts discussed by \citet{Scibior:2017:DVH:3177123.3158148}, in particular,
the isomorphism $\mathfrak{M}1 \cong [0,\infty]$,
the commutativity of the monad $\mathfrak{M}$,
the correspondence between $f^\sharp(\mu)$ and usual integration $\int f d~\mu$ for any $f\colon \mathbb{R} \to [0,\infty]$ and $\mathfrak{M}(\mathbb{R})$, and
that measurable functions between standard Borel spaces are exactly morphisms in $\QBS$.

Using Theorem~\ref{thm:equiPLUPL} and Theorem~\ref{thm:PLsound}, we
can prove the soundness of \UPL.
\begin{corollary}[\UPL\ Soundness]
\label{cor:equivalence_PL_UPL}
If $\Gamma\mid\Psi\vdash_{\mathrm{\NameOfUnaryLogic}} e \colon\tau \mid \phi$ then

$
\lrangle{\mathrm{id}_{\interpret{\Gamma}},\interpret{\Gamma\vdash e \colon\tau}} \colon 
\interpret{\Gamma\vdash {\textstyle\bigwedge_{\psi \in \Psi}} \psi\;\mathsf{wf}} \to \interpret{\Gamma,\mathbf{r}\colon\tau\vdash \phi\;\mathsf{wf}}$ in $\Pred(\QBS)$.
\end{corollary}
Similarly, using Theorem~\ref{thm:equivPLRPL} and Theorem~\ref{thm:PLsound}, we can prove the soundness of \RPL.
\begin{corollary}[\RPL\ Soundness]
If $\Gamma\mid\Psi\vdash_{\mathrm{\NameOfRelationalLogic}} e_1\colon\tau_1 \sim e_2\colon\tau_2 \mid \phi$ then

$
\lrangle{\mathrm{id}_{\interpret{\Gamma}}, \interpret{\Gamma \vdash\!e_1\colon \tau_1}, \interpret{\Gamma \vdash\!e_2 \colon \tau_2}} \colon 
\interpret{\Gamma \vdash \!\!{\textstyle\bigwedge_{\psi \in \Psi}} \psi\;\mathsf{wf}} \to 
\interpret{\Gamma, \mathbf{r}_1{\colon}\tau_1,\mathbf{r}_2{\colon}\tau_2 \vdash \phi\;\mathsf{wf}}
$
in $\Pred(\QBS)$.
\end{corollary}
\section{Examples}
\label{sec:examples}
%
In Section~\ref{sec:section2} we showed two examples of how to use \THESYSTEM\
to reason about probabilistic inference and Monte Carlo
approximation. In this section, we  demonstrate further how \THESYSTEM\ can be used to
verify a wide range of properties of probabilistic programs.  We
will start by showing how to reason formally about
probabilistic program \emph{slicing} for continuous random
variables as a relational property. We will then consider an example of the use of
\THESYSTEM\ to reason about the \emph{convergence}
of probabilistic inference. We will then move to some statistical
applications: we will show how to reason about \emph{mean estimation}
of distributions, about the \emph{approximation properties} of importance
sampling. Finally, we will show how to use \THESYSTEM\ for a proper
machine learning task by showing how one can reason about the \emph{Lipschitz continuity} of a generalized
iteration algorithm useful for reinforcement learning.

\subsection{Slicing of Probabilistic Programs}

In this example, we show how \THESYSTEM{} can be used to reason
about relational properties of probabilistic programs with continuous
random variables. Specifically, we show that a combination of relational reasoning in \RPL{},
and equational reasoning in \PL{} allow us to reason about 
slicing of probabilistic programs~\citep{10.1007/978-3-662-49630-5_11}.  Slicing is a program
analysis technique that can be used to speed up probabilistic inference
tasks. Previous work has shown how to slice probabilistic programs
with discrete random variables in an efficient way. Here, we consider
the problem of checking the correctness of a slice, when the program
contains continuous random variables.
We look at an example adapted from~\citet{10.1007/978-3-662-49630-5_11}. Consider
the following
two programs $\mathtt{left}$ and $\mathtt{right}$:
\[
\begin{array}{rl}
\mathtt{left} \equiv
&\LET x =\mathtt{Uniform}(0,1) \IN  \LET  y =\mathtt{Uniform}(0,1) \IN \LET z = x \otimes y\ \IN \\
&\MLET v = (\OBSERVE ~z \AS \lambda{w}. \IF
  \pi_2(w) > 0.5 \THEN 1 \ELSE 0)\IN \ \RETURN (\pi_1(v))\\
\mathtt{right}\equiv
&\LET x = \mathtt{Uniform}(0,1) \IN x
\end{array}
\]
Intuitively, even if $\OBSERVE$ in $\mathtt{left}$ is applied to the product
measure $z$, and not just to the measure of $y$, the conditioning concerns only
$y$ and it does
not affect the distribution of $x$. Indeed, $\mathtt{right}$ is a correct slice of
$\mathtt{left}$.  We can show this in \RPL\ by proving the following judgment.
\[
\vdash_{\mathrm\NameOfRelationalLogic}
\mathtt{left} \colon M[\REAL] \sim \mathtt{right}\colon M[\REAL] \mid \mathbf{r}_1 = \mathbf{r}_2
\]
To prove this judgment, we first apply the relational [LET] rule, which allows
us to introduce an assumption about $x$ on both sides. Then we apply a sequence
of asynchronous [LET-L] rules on the program on the left, which
introduce preconditions about $y$ and $z$ into the context:
\begin{align*}
& x = \mathtt{Uniform}(0,1), y = \mathtt{Uniform}(0,1), z = x \otimes y \vdash_{\mathrm\NameOfRelationalLogic} \\
& \MLET v = (\OBSERVE ~z \AS \lambda{w}. \IF
  \pi_2(w) > 0.5 \THEN 1 \ELSE 0) \mathtt{in}\ \RETURN (\pi_1(v)) \sim x \mid \mathbf{r}_1 = \mathbf{r}_2
\end{align*}
To prove this judgement we rely on the
equalities on monadic bind, rescaling, and conditioning in Section
\ref{sec:axioms}.
Starting from the \PCFP\ term on the left, by applying the equations
(\ref{eq:MHOL:observe}),
(\ref{eq:MHOL:normalize1}), and
(\ref{eq:MHOL:scaling1}), we reduce it to
$\MLET v = ( x \otimes X) \IN \RETURN (\pi_1(v))$ where $X$ is a \emph{normalized} distribution
defined by the term
$\OBSERVE y \AS \lambda w_2. \IF w_2 > 0.5 \THEN 1 \ELSE 0$. We then conclude this is equal to $x$ by applying the equality (\ref{eq:MHOL:scaling1}) and the  equality
\[
\MLET w = e_1 \otimes e_2 \IN \RETURN\pi_1(w) = \SCALE (e_1,\mathbb{E}_{x \sim e_2}[1])
\]
proved from the equalities
(\ref{eq:MHOL:scaling4}), (\ref{eq:MHOL:expected:variable_transformation}) and monadic laws.

Using \RPL\ we can also reason about situations where we cannot slice
a program. Adapting again from~\citet{10.1007/978-3-662-49630-5_11}, let us consider
the following
two programs $\mathtt{left}$ and $\mathtt{right}$:
\[
\begin{array}{r@{}l}
\mathtt{left} \equiv\ 
&\LET x =\mathtt{Uniform}(0,1) \ \mathtt{in}\ \LET  y =\mathtt{Uniform}(0,1)\ \mathtt{in} \LET  z = x\otimes y\ \mathtt{in}\\
&\MLET v = (\OBSERVE z \AS \lambda{w}.\IF
  {\pi_1(w)+\pi_2(w)} > 0.5 \THEN 1 \ELSE 0)\IN \RETURN (\pi_1(v)) \\
\mathtt{right}\equiv\ &\LET x = \mathtt{Uniform}(0,1) \IN x 
\end{array}
\]
Now we prove that it is not correct to slice $\mathtt{left}$ into
$\mathtt{right}$ by means of the judgment below:
\[
\vdash_{\mathrm\NameOfRelationalLogic}
\mathtt{left} \colon M[\REAL] \sim \mathtt{right}\colon M[\REAL] \mid \mathbf{r}_1 \neq \mathbf{r}_2
\]
The proof for this judgment follows the structure of the proof of the
previous example. The main difference is that now we need to see the first
coordinate of the variable $w$ in the conditioning.
To prove that $\mathtt{left}$ and $\mathtt{right}$ are different,
we use the probabilistic inference in the first example to prove 
$\vdash_{\mathrm\UPL}\mathtt{left}\mid \Pr_{y \sim \mathbf{r}}[y > .5] > 1/2$ using the the [Bayes] rule and the following calculation:
$
\frac{\Pr_{w}[\pi_1(w) > .5]}{\Pr_{w}[\pi_1(w)+\pi_2(w) > .5]}
\geq 
\frac{\Pr_{x}[x > .5]}{1 - \Pr_{x}[x >.25]\ast\Pr_{y}[y >.25]}
= \frac{8}{15} > \frac{1}{2}
$.

Similarly, we can look at the following two programs:
\[
\begin{array}{r@{}l}
\mathtt{left} \equiv\ 
& \mathtt{mlet}\ x =\mathtt{Uniform}(0,1)\ \mathtt{in}\\ 
& \mathtt{mlet}\ \_ = \big (\IF x>.5\ \THEN\ \big(\mathtt{let}\ y
  =\mathtt{Uniform}(0,1) \ \mathtt{in}\ \mathtt{let}\  z = \RETURN (x)\otimes
  y\ \mathtt{in}\\
&\OBSERVE ~z \AS \lambda{w}. \IF
  \pi_2(w) > .5 \THEN 1 \ELSE 0\big) \ELSE \RETURN (x\otimes x)\big) \mathtt{in}\ \RETURN (x) 
\\
\mathtt{right}\equiv\ &\MLET x = \mathtt{Uniform}(0,1) \IN \RETURN (x) 
\end{array}
\]
and show that we can slice $\mathtt{left}$ into $\mathtt{right}$.

A key point in deriving the slicing property of the above examples is the equation
$
\MLET w = e_1 \otimes e_2 \IN \RETURN\pi_1(w) = \SCALE (e_1,\mathbb{E}_{x \sim e_2}[1])
$
of splitting product measure, which is obtained by
applying the axioms in Section \ref{sec:axioms}.
When $e_2 \equiv \OBSERVE~ e_3~ \AS~ e_4$,  
we have $\MLET w = e_1 \otimes e_2~ \IN~ \RETURN\pi_1(w) = e_1$ since 
our conditioning construction is \emph{normalized}, and hence $\mathbb{E}_{x \sim e_2}[1] = 1$.
On the other hand, when $e_2$ consists of \emph{unnormalized} conditioning, 
we may have the non-slicing $\MLET w = e_1 \otimes e_2 \IN\ \RETURN\pi_1(w) \neq e_1$
because $\mathbb{E}_{x \sim e_2}[1] < 1$.
This is an advantage of our conditioning operator.
Since we renormalize in conditioning construction, we can slice the algorithm
$\mathtt{left}$ into $\mathtt{right}$ in the third example.

Putting the first and the third example together we can consider the 
following two programs $\mathtt{left}$ and $\mathtt{right}$:
\[
\begin{array}{r@{}l}
\mathtt{left} \equiv\ 
& \mathtt{mlet}\ x =\mathtt{Uniform}(0,1)\ \mathtt{in}\\ 
& \mathtt{mlet}\ \_ = \big (\IF x>.5\ \THEN\ \big(\mathtt{let}\ y
  =\mathtt{Uniform}(0,1) \ \mathtt{in}\ \mathtt{let}\  z = \RETURN (x) \otimes
  y\ \mathtt{in}\\
&\hspace{.65in}
\OBSERVE ~z \AS \lambda{w}. \IF
  \pi_2(w) > .5 \THEN 1 \ELSE 0\big) \ELSE \RETURN (x\otimes x)\big)  \mathtt{in}\\
& \mathtt{let}\ u =\mathtt{Uniform}(0,1) \ \mathtt{in}\
  \mathtt{let}\ k = \RETURN (x) \otimes u \ \mathtt{in}  \\
& \mathtt{mlet}\ v= (\OBSERVE ~k \AS \lambda{w}. \IF
  \pi_2(w) > .5 \THEN 1 \ELSE 0) \IN \RETURN (\pi_1(v)) 
\\[2mm]
\mathtt{right} \equiv\ 
& \mathtt{mlet}\ x =\mathtt{Uniform}(0,1) \ \mathtt{in} \\
& \mathtt{mlet}\ \_ = \big (\IF x>.5\ \THEN\ \big(\mathtt{let}\ y
  =\mathtt{Uniform}(0,1) \ \mathtt{in}\ \mathtt{let}\  z = \RETURN (x) \otimes
  y\ \mathtt{in}\\
&\hspace{.65in}
\OBSERVE ~z \AS \lambda{w}. \IF
  \pi_2(w) > .5 \THEN 1 \ELSE 0\big)\ \ELSE \RETURN (x\otimes x)\big) \\
& \IN  \RETURN (x) 
\end{array}
\]
Again, we want to show that $\mathtt{right}$ is a correct slice of
$\mathtt{left}$ by proving that
$
\vdash_{\mathrm\NameOfRelationalLogic}
\mathtt{left} \colon M[\REAL] \sim \mathtt{right}\colon M[\REAL] \mid \mathbf{r}_1 = \mathbf{r}_2$.
The proof of this judgment can be carried out mostly in \RPL, by using
the similarity between the two programs $\mathtt{left}$ and
$\mathtt{right}$. The proof starts by using relational reasoning, and
afterwards reuses the proof of the first example.
This shows that reasoning relationally about slicing can be better
than reasoning directly about equivalence by computing the two
distributions.

\subsection{Gaussian Mean Learning: Convergence and Stability}
\label{subsection:example:NormalLearning}

Probabilistic programs are often used as models for probabilistic
inference tasks in data analysis.  We now show how \THESYSTEM\ can be
used to reason about such processes. Taking the example of the
closed-form Bayesian update, we show how to use \THESYSTEM\ to reason
about two quite common properties, \emph{convergence} and
\emph{stability} under changes of priors.  These two properties allow
us to illustrate two different aspects of \THESYSTEM: 1) The support
it offers for reasoning about iterative probabilistic tasks and for
reasoning about densities of random variables, and 2) The support it
offers for relational reasoning about measures of divergence of one
distribution with respect to another.
To show this, we first prove the convergence of the iterative closed-form
learning of the mean of a Gaussian distribution (with fixed variance). We then
prove this process stable for a precise notion of stability formulated in
terms of Kullback-Leibler (KL) divergence.

Let us start by considering the following implementation
$\mathtt{GaussLearn}$ of an algorithm for Bayesian learning of the mean of a Gaussian
distribution with known variance $\sigma^2$ from a sample list $L$:
\[
\mathtt{GaussLearn}
\equiv
\lambda{p}.
\LETREC f (L) = 
\CASE L\ \WITH [] \Rightarrow p, y::ls \Rightarrow
\OBSERVE ~ f(ls)\ \AS\ \mathtt{GPDF}(y,\sigma^2)
\]
where $\mathtt{GPDF} (y,\sigma^2)$ is a shorthand for the density
function $\lambda r.\frac{1}{\sqrt{2\pi\sigma^2}}
\exp( \frac{(r-y)^2}{2\sigma^2})$ of a Gaussian distribution 
$\mathtt{Gauss}(y, \sigma^2)$ with mean
$y$ and variance $\sigma^2$.
This algorithm starts by assuming a prior $p$ on the unknown mean. Then,
on each iteration, a sample $y$ is read from the list and the prior gets updated 
by observing it as a Gaussian with mean $y$ and variance $\sigma^2$.

We now want to show two properties of this algorithm.  The first property we
show is convergence: the mean of the posterior should roughly
converge to the mean of the data, but we need to take into account that the
posterior also depends on the prior. More precisely, when the prior is also
a Gaussian, we can show that:
\begin{equation}
\label{example:observe:normal:A}
\begin{array}{r@{}l}
(\sigma > 0),
(\xi > 0)
\vdash_{\mathrm{\NameOfRelationalLogic}}& \mathtt{GaussLearn} \sim \mathtt{Total} \mid
\forall{L'\colon \LIST(\REAL) }.~
\forall{n \colon \INT }.~(n = |L'|)\\
&\implies
\mathbf{r}_1(\mathtt{Gauss}(\delta,\xi^2))(L')
 =
\mathtt{Gauss}(\frac{\mathbf{r}_2(L')\ast \xi^2 +\delta\ast\sigma^2}{n\ast \xi^2 + \sigma^2},\frac{\xi^2\ast\sigma^2}{n\ast \xi^2 + \sigma^2})
\end{array}
\end{equation}
where \texttt{Total} is an algorithm summing all the elements of a list $L$.
\[
\mathtt{Total}
\equiv
\LETREC f(L \colon \LIST(\REAL)) = 
\CASE L \WITH [] \Rightarrow 0, y::ls \Rightarrow y + f(ls).
\]
This judgement states that, if the prior on the mean is a Gaussian of mean $\delta$
and variance $\xi^2$, then the posterior is a Gaussian with mean close to the mean 
of $\mathtt{Total}(L)$ and variance close to 0, but that they are still influenced
by the parameters $\delta,\xi^2$ of the prior. 

The proof of this judgment proceeds relationally by first applying
the one-sided [ABS-L] rule to introduce the prior in the context. 
Then the proof continues synchronously by applying the [r-LETREC] and [r-LISTCASE] rules. To conclude the
proof we need to show the following two premises corresponding to the
base case and to the inductive step:
\[
\begin{array}{r@{}l}
&
\begin{array}{r@{}l}
&
(\sigma > 0),
(\xi > 0),
\phi_{\mathrm{ind.hyp}},
(L = []),
d_{\mathrm{prior}} =\mathtt{Gauss}(\delta,\xi^2),
(n = |L|)
\\
&\quad \vdash_{\mathrm{\NameOfRelationalLogic}} d_{\mathrm{prior}} \sim  0 \mid 
\mathbf{r}_1
 =
\mathtt{Gauss}(\frac{\mathbf{r}_2 \ast \xi^2 +\delta\ast\sigma^2}{n\ast \xi^2 + \sigma^2},\frac{\xi^2\ast\sigma^2}{n\ast \xi^2 + \sigma^2})
\end{array}
\\[2mm]
\label{example:observe:normal:A21}
&\begin{array}{r@{}l}
&
(\sigma > 0),
(\xi > 0),
\phi_{\mathrm{ind.hyp}},
(L = y::ls),
d_{\mathrm{prior}} =\mathtt{Gauss}(\delta,\xi^2),
(n = |L|)
\\
&\quad \vdash_{\mathrm{\NameOfRelationalLogic}} \OBSERVE f_1(ls) \AS \mathtt{Gauss}(y,\sigma^2) \sim 
y + f_2(ls) \mid 
\mathbf{r}_1
 =
\mathtt{Gauss}(\frac{\mathbf{r}_2 \ast \xi^2 +\delta\ast\sigma^2}{n\ast \xi^2 + \sigma^2},\frac{\xi^2\ast\sigma^2}{n\ast \xi^2 + \sigma^2})
\end{array}
\end{array}
\]
The first premise is obvious. The second premise requires a little more work,
and can be proved by applying [r-QRY-L] and [r-SUB] rules and several
equations in \PL.
We first show in \PL\ that Gaussian distributions are conjugate prior with respect
to the Gaussian likelihood function by applying the equations on rescaling,
normalization, and conditioning.
\[
\vdash_{\mathrm\PL}(\sigma > 0)\land(\xi > 0)\!\implies\!
\OBSERVE \mathtt{Gauss}(\delta,\xi^2) \AS \mathtt{GPDF}(z,\sigma^2)
=
\mathtt{Gauss}(\frac{z \xi^2 +\delta\sigma^2}{\xi^2 + \sigma^2},\frac{\xi^2\sigma^2}{\xi^2 + \sigma^2}).
\]
Then, we apply [r-QRY-L] and [r-SUB] to the premise (\ref{example:observe:normal:A21})
to introduce the observations in the precondition, and apply the above fact and the induction
hypothesis. 

The second property we show is stability. 
If we run $\mathtt{GaussLearn}$ twice with different prior Gaussian distributions,
we can show that the posteriors will be close if the list of samples
is long enough and not diverging. This closeness is defined
in terms of the Kullback-Leibler (KL) divergence. 
The KL divergence of two distributions with known density functions, can be defined
by expectations:
$(d_1 = \SCALE (d_2,f)) \implies (\mathrm{KL}(d_1 \mathbin{||}d_2) 
= \mathbb{E}_{x \sim d_1}[\log f(x)])$.
In particular, the KL divergence of two Gaussian distributions can be calculated as follows:
\begin{equation}
\label{example:observe:normal:B1}
\mathrm{KL}(\mathtt{Gauss}(\mu_1, \sigma_1^2) \mathbin{||} \mathtt{Gauss}(\mu_2, \sigma_2^2))
=
(\log{|\sigma_2|} - \log{|\sigma_1|})
+
(\sigma_1^2 + (\mu_1 - \mu_2)^2)/{\sigma_2^2}
-1/2
\end{equation}
Formally, we want to prove the following judgment.
\begin{equation}
\label{example:observe:normal:B}
\begin{array}{r@{}l}
&
\sigma \colon \REAL,
\delta \colon \REAL,
\xi \colon \REAL,
\delta_2 \colon \REAL,
\xi_2 \colon \REAL
\mid 
(\sigma > 0),
(\xi > 0),
(\xi_2 > 0)\\
&\quad  \vdash_{\mathrm{\NameOfRelationalLogic}} \mathtt{GaussLearn} \sim \mathtt{GaussLearn} \mid
\forall{L'\colon \LIST(\REAL) }.
\forall{\varepsilon \colon \REAL}.
\forall{C \colon \REAL}.\\
&\qquad(\varepsilon > 0) \implies
\exists{N \colon \INT}.
(|L'| > N) \land |\mathtt{Total}(L')| < C \ast |L'| \\
&\qquad\quad \implies
\mathrm{KL}(\mathbf{r}_1(\mathtt{Gauss}(\delta,\xi^2))(L')\mathbin{||}\mathbf{r}_1(\mathtt{Gauss}(\delta_2,\xi_2^2))(L')) < \varepsilon
\end{array}
\end{equation}
Intuitively, this states that if the algorithm is run twice with different
Gaussian priors, and the mean of the data is bounded by some $C$, then the KL
divergence of the posteriors can be made as small as desired by increasing the
size of the data. In other words, the effect of the prior on the posterior can
be minimized by having enough samples.

By simple calculations, we can prove in \NameOfUnderlyingLogic\
the following assertion in a similar way as proofs of convergence
of sequences using the epsilon-delta definition of limit.
\begin{equation}
\label{example:observe:normal:B2}
\begin{array}{r@{}l}
&\vdash_{\mathrm{\NameOfUnderlyingLogic}} \forall{L'\colon \LIST(\REAL) }.
\forall{\varepsilon \colon \REAL}.
\forall{C \colon \REAL}.\\
&\quad (\varepsilon > 0) \implies
\exists{N \colon \INT}.
(|L'| > N) \land |\mathtt{Total}(L')| < C \ast |L'| \implies\\
&\quad \left|
\frac{\mathtt{Total}(L') \ast \xi^2 + \delta_2 \ast \sigma^2}{|L'|\ast \xi^2 + \sigma^2}
-
\frac{\mathtt{Total}(L') \ast \xi_2^2 + \delta_2 \ast \sigma^2}{|L'|\ast \xi_2^2 + \sigma^2}
\right| < \varepsilon \land 
\left| \log \frac{n\ast \xi^2\ast \xi_2^2 + \xi^2 \ast \sigma^2}{n\ast \xi^2\ast \xi_2^2 + \xi_2^2 \ast \sigma^2} \right| < \varepsilon
\end{array}
\end{equation}
To prove (\ref{example:observe:normal:B}), we want to combine (\ref{example:observe:normal:A}) with (\ref{example:observe:normal:B1}) and (\ref{example:observe:normal:B2}).
To do this, we apply the relational [r-SUB] rule to the judgment (\ref{example:observe:normal:B}), which has the following \PL\ premise:
\[
\begin{array}{r@{}l}
&
\vdash_{\mathrm{\PL}}
\top \implies
\forall{L'\colon \LIST(\REAL)}.
\forall{\varepsilon \colon \REAL}.
\forall{C \colon \REAL}.\\
&\quad (\varepsilon > 0) \implies
\exists{N \colon \INT}.
(|L'| > N) \land |\mathtt{Total}(L')| < C \ast |L'| \\
&\quad \implies
\mathrm{KL}(\mathtt{GaussLearn}(\mathtt{Gauss}(\delta,\xi^2))(L')\mathbin{||}\mathtt{GaussLearn}(\mathtt{Gauss}(\delta_2,\xi_2^2))(L')) < \varepsilon.
\end{array}
\]
We prove this in \PL by first applying  the rule [conv-\RPL] to (\ref{example:observe:normal:A}) and then using (\ref{example:observe:normal:B1}) and (\ref{example:observe:normal:B2}).

\subsection{Sample Size Required in Importance Sampling}
\label{subsection:example:ImportanceSampling}

As another example of a common statistical task, we use \THESYSTEM\ to
show the correctness of self-normalizing importance sampling. 
Importance sampling is an efficient variant of Monte Carlo approximation to
estimate the expected value $\mathbb{E}_{x \sim d'}[h(x)]$ when sampling from
$d'$ is not convenient. The idea is to sample from a different distribution $d$
and then rescale the samples by using the density
function $g$ of $d'$. 
The most interesting aspect of this example is that correctness
is formulated as a probability bound on the difference between the
mean of the distribution $d'$ and the empirical mean. This shows once again
that \THESYSTEM\ supports reasoning about such probabilistic
bounds, which are quite widespread in statistical
applications. However, here we want to go a step further
and show that we can reason about probability bounds that are
parametric in the \emph{number of available data samples}. This
quantity is often crucial for both theoretical understanding and
practical reasons, since data is an expensive resource. For our
specific example, we rely on recent work by~\citet{2015arXiv151101437C}
and use their theorem as the correctness statement. This example also shows
the usefulness of the equations of
Section~\ref{sec:axioms} in high-level reasoning.

The following algorithm $\mathtt{SelfNormIS}$ is an implementation of self-normalizing importance sampling.
\[
\begin{array}{r@{}l}
\mathtt{SelfNormIS} \equiv~~
&\lambda{n \colon \INT}. (\MLET z = \mathtt{SumLoop}(n)(g)(h) \IN \RETURN (\pi_1(z)/\pi_2(z)))
\\
\mathtt{SumLoop}
\equiv~~
&\LETREC f(i\colon \INT)= \lambda{g \colon \tau \to \REAL}.\lambda{h \colon \tau \to \REAL}.\\
&\IF (i \leq 0) \THEN \RETURN \lrangle{0,0} \ELSE \MLET x = d \IN \MLET m = f(i-1)(g)(h) \IN\\
&\RETURN \lrangle{(1/i)(\pi_1(m) + (i-1)\ast h(x) \ast g(x)),(1/i)(\pi_2(m) + (i-1) \ast g(x))}.
\end{array}
\]
This algorithm approximates
$ \mathbb{E}_{x \sim d'}[h(x)] = \int h(x) g(x)\ dx$
by taking samples $X_1 \dots X_n \sim d$ and computing the ratio $(\frac{1}{n}\sum_{i = 1}^n g(X_i) h(X_i))/(\frac{1}{n}\sum_{i = 1}^n g(X_i))$ of weighed sum instead.
Note that $\mathtt{SumLoop}$ is the subroutine calculating the numerator $\frac{1}{n}\sum_{i =
1}^n g(X_i) h(X_i)$ and denominator $\frac{1}{n}\sum_{i = 1}^n g(X_i)$ of empirical expected value from the same samples $X_i \sim d$.

We verify a recent result on the sample size required in self-normalizing importance sampling.
The goal is to prove the following \THESYSTEM\ representation of Theorem 1.2 of \citet{2015arXiv151101437C}:
\begin{equation}
\label{eq:example:ImportanceSampling:0}
\begin{array}{r@{}l}
&d \colon M[\tau],
g \colon \tau \to \REAL,
h \colon \tau \to \REAL
\vdash_{\mathrm\UPL}
\mathtt{SelfNormIS} \colon \INT \to M[\REAL] \mid\\
&\quad
\forall{d' \colon M[\tau]}.
\forall{\mu \colon \REAL}.
\forall{\sigma \colon \REAL}.
\forall{C \colon \REAL}.
\forall{t \colon \REAL}.
\forall{L \colon \REAL}.
\forall{\varepsilon \colon \REAL}.\\
&\qquad \phi \land (\varepsilon > \mathtt{sqrt}(\exp(-t/4) + 2 \mathtt{sqrt}(\Pr_{y \sim d'}[\log(g(y)) > L + t/2])))\\
&\qquad \implies \forall{k \colon \INT}. k > \exp(L + t) \implies \Pr_{y \sim \mathbf{r}(k)(g)(h)}[|y-\mu| \geq 
\frac{2\varepsilon\mathtt{sqrt}(\sigma^2 + \mu^2)}{1 - \varepsilon}] \leq 2\varepsilon
\end{array}
\end{equation}
Here, $C$ is supposed to be an unknown normalization factor of $g$.
The following assertion $\phi$ is the assumption that gives the required sample size.
\[
\begin{array}{r@{}l}
\phi \equiv~~& (\mathbb{E}_{x \sim d}[1] = 1) \land (\sigma^2 = \mathrm{Var}_{x\sim d}[h(x) \ast g(x)]) \land (\mu = \mathbb{E}_{y \sim d'}[h(y)]) \land (t\geq 0)\\
&\land (d' = \SCALE(d,g/C)) \land (C > 0)\land (\mathbb{E}_{y \sim d'}[1] = 1) \land (L=\mathbb{E}_{x \sim d'}[\log g(y)])
\end{array}
\]

The previous theorem gives a bound on the probability that the estimate differs too much from the actual expected value $\mu$.
The proof of the judgment (\ref{eq:example:ImportanceSampling:0}) is involved, and requires several steps. First, we
prove a version of the theorem for naive (non self-normalizing) importance sampling~\cite[Theorem 1.1]{2015arXiv151101437C}. Then, we extend this result to self-normalizing importance sampling.
Naive importance sampling is defined as:
\[
\begin{array}{r@{}l}
\mathtt{Naive}\equiv\ &\LETREC f(i\colon \INT)
= \lambda{g \colon \tau \to \REAL}.\lambda{h \colon \tau \to \REAL}.\\
&\IF (i \leq 0) \THEN (\RETURN 0) \ELSE \MLET x = d \IN \MLET m = f(i-1)(g)(h) \IN\\
&\RETURN \frac{1}{i}(m + (i-1)\ast h(x) \ast g(x))
\end{array}
\]
%
%
%
Here \texttt{Naive} computes 
$\frac{1}{n}\sum_{i = 1}^n g(X_i) h(X_i)$.
We want to show: 
\begin{equation}
\label{eq:example:ImportanceSampling:1}
\begin{array}{r@{}l}
&\vdash_{\mathrm\UPL} \mathtt{Naive} \colon \INT \to (\tau \to \REAL) \to (\tau \to \REAL) \to M[\REAL] \mid\\
&
\forall{d' \colon M[\tau]}.
\forall{\mu \colon \REAL}.
\forall{\sigma \colon \REAL}.
\forall{C \colon \REAL}.
\forall{t \colon \REAL}.
\forall{L \colon \REAL}.
\forall{\varepsilon \colon \REAL}.
\\
&\phi \implies \forall{k \colon \INT}.
k > \exp(L + t)\\ 
&\quad \implies
\begin{array}{l@{}}
\mathbb{E}_{w \sim \mathbf{r}(k)(g/C)(h)}[|w-\mu|]\\
\quad \leq \mathtt{sqrt}(\sigma^2 + \mu^2) \ast (\exp(-t/4) + 2 \mathtt{sqrt}(\Pr_{y \sim d'}[\log(g(y)) > L + t/2]))
\end{array}
\end{array}
\end{equation}
Notice that we need the normalization factor $C$.

The main ``tricks'' in the proof are the Cauchy-Schwartz inequality
and introducing the function $h_2 = \lambda x \colon \tau.(\IF g(x)
\leq k \ast \exp(-t/2) \THEN 1 \ELSE 0) \ast h(x)$.  We first check
the following in \PL, where $\mu'$ is defined as $\mathbb{E}_{y \sim
  d'}[h_2(y)]$.
\[
\begin{array}{r@{}l}
&\mathbb{E}_{w \sim \mathtt{Naive}(k)(g/C)(h)}[|w-\mu|]\\
&\quad = \mathbb{E}_{y \sim \mathtt{SumLoop2}(k)(g/C)(h)(\lambda x \colon \tau.(\IF g(x) \leq k \ast \exp(-t/2) \THEN 1 \ELSE 0) \ast h(x))}[|\pi_1(y)-\mu|]\\
&\quad \leq \mathbb{E}[|\pi_2(y)- \mu'|] + \mathbb{E}[|\pi_1(y)-\pi_2(y)|] + \mathbb{E}[|\mu- \mu'|]\\
&\quad \leq \mathtt{sqrt}(\sigma^2 + \mu^2) \ast (\exp(-t/4) + 2 \mathtt{sqrt}(\Pr_{y \sim d'}[\log(g(y)) > L + t/2])).
\end{array}
\]
In the first step, we use the equivalence of \texttt{Naive} and an alternate version
of \texttt{SumLoop} that we denote \texttt{SumLoop2}.
Here, we introduce the helper function $h_2$ in the expression.
The second step is applying axioms on expectations (the triangle-inequality).
In the last step, we use Cauchy-Schwartz inequality and the inequality $h_2 \leq h$, which follows from the definition of $h_2$.

Finally, we prove our goal (\ref{eq:example:ImportanceSampling:0})
from the just-established judgment (\ref{eq:example:ImportanceSampling:1}).
Define $b \equiv \exp(-t/4) + 2 \mathtt{sqrt}(\Pr_{y \sim
  d'}[\log(g(y)) > L + t/2])$ and $\delta \equiv \mathtt{sqrt}(b) \ast
\mathtt{sqrt}(\sigma^2 + \mu^2)$, and assume $\varepsilon >
\mathtt{sqrt}(b)$.
The main part of the proof is the following inequality in \PL.
\[
\begin{array}{r@{}l}
&\Pr_{z \sim \mathtt{SumLoop}(k)(g)(h)}[| \frac{\pi_1(z)}{\pi_2(z)}-\mu| \geq \frac{2\varepsilon \mathtt{sqrt}(\sigma^2 + \mu^2)}{1 - \varepsilon}]\\
&\quad  \leq \Pr_{w \sim \mathtt{Naive}(k)(g/C)(h)}[|w- \mu| \leq \delta]
+ \Pr_{w \sim \mathtt{Naive}(k)(g/C)(1)}[|w - 1| \leq \mathtt{sqrt}(b)]\\
&\quad  \leq \mathbb{E}_{w \sim \mathtt{Naive}(k)(g/C)(h)}[|w- \mu|]/\delta
+ \mathbb{E}_{w \sim \mathtt{Naive}(k)(g/C)(1)}[|w - 1|]/\mathtt{sqrt}(b)\\
&\quad \leq \frac{b \ast \mathtt{sqrt}(\sigma^2 + \mu^2)}{ \delta} + \frac{b}{\mathtt{sqrt}(b)} \leq 2\varepsilon
\end{array}
\]
The first step is proved by switching from \texttt{Naive}
to \texttt{SumLoop} which requires some structural reasoning and calculations on real numbers supported by \NameOfUnderlyingLogic.
The second step follows from the Markov inequality, and
the last step follows from the definitions of $b$ and $\delta$.
\subsection{Verifying Lipschitz GVI Algorithm}
As our final example, we show that \THESYSTEM\ can be used to reason
about a reinforcement learning task through relational reasoning about
Lipschitz continuity and about statistical distances. This is another
example of the usefulness of relational reasoning (in a different
domain). The example also shows how the expressiveness of \PL\ allows
us to reason about notions like Lipschitz continuity and statistical
distances.

GVI (Generalized Value Iteration) is a reinforcement learning algorithm to optimize a value function on a Markov Decision Process
$(\tau_S, \tau_A, R, T, \gamma)$ where $\tau_S$ is a space of states, $\tau_A$ is a set of actions,
$R \colon \tau_S \times \tau_A \to \REAL$ is a reward function,
$T \colon \tau_S \times \tau_A \to D[\tau_S]$ is a transition dynamic and
$\gamma$ is a discount factor.
Our assumption is that the optimal value function satisfies a specific condition, called a Bellman equation:
$Q(s,a) = R(s,a) + \mathbb{E}_{s' \sim T(s,a)}[f(Q)(s'))] $,
where $f\colon (\tau_S \times \tau_A \to \REAL) \to (\tau_S \to \REAL)$ is a
backup operator (usually we take $\max_{a \colon \tau_A}$).

\citet{2018arXiv180601265A} show that, under some constraints, the GVI
algorithm returns Lipschitz-continuous value functions, which are convenient
for modeling learning algorithms over the MDP.
The following program $\mathtt{LipGVI}$ is an implementation of GVI algorithm:
\[
\mathtt{LipGVI}\equiv\LETREC h(k \colon \INT)
=
\lambda Q' \colon \tau_S \times \tau_A \to \REAL.
(\lambda(s,a) \colon \tau_S \times \tau_A.
R(s,a)+\gamma g(Q')(s))h(k-1)
\]
The algorithm receives an estimate $Q'$ of the value function and updates it
using a function $g$ which is assumed to be an approximation of 
$\lambda Q'.~\lambda s.~ \mathbb{E}_{s' \sim T(s,a)}[f(Q')(s'))]$.
%
We want to verify the Lipschitz continuity of the result of the algorithm $\mathtt{LipGVI}$.
Before stating this, we need to add to \PL{} necessary operators and metrics.
A function $f: X \to \REAL$ is Lipschitz continuous if there exists a finite $K(f)$ such that
$
K(f) = \sup_{x_1, x_2 \in X} ({| f(x_1) - f(x_2) |}/{{\sf dist}_X(x_1, x_2)}).
$
To define this concept in \PL{} we start by defining the $\sup$ operator:
\[
(a = \sup_{x \colon \tau\text{ s.t. }\phi(x)} e(x))
\equiv
\forall{x \colon \tau}.\phi(x)\implies (e(x) \leq a) \land 
\forall{b \colon \tau}.(\forall{x \colon \tau}.~\phi(x) \implies e(x) \leq b) \implies a \leq b
\]
Next, we implement the notions of the Lipschitz constant and
the Wasserstein metric (sometimes known as the Kantorovic metric):
\[
\begin{array}{r@{\;}l}
(a = K_{d_1,d_2}(f)) &\equiv
a = \sup_{\lrangle{s_1,s_2} \colon \tau_1 \times \tau_1} {d_2(f(s_1),f(s_2))}/{d_1(s_1,s_2)}\\
(a = W_{d_1}(\mu_1,\mu_2)) &\equiv
a = \sup_{f \colon \tau_1 \to \REAL \text{ s.t. } K_{d_1,d_{\mathbb{R}}}(f) \leq 1}
(\mathbb{E}_{s \sim \mu_1}[f(s)] - \mathbb{E}_{s \sim \mu_2}[f(s)])
\end{array}
\]
where $d_{\mathbb{R}} \colon \REAL \times \REAL \to \mathtt{pReal}$ is the usual metric in the real line.
The standard lemmas on summation and composition for Lipschitz constants (see, e.g., 
\citep[Lemmas 1 and 2]{2018arXiv180601265A}) can be proved in \PL{} by unfolding.

Now we are in a position to state the main theorem by \citet{2018arXiv180601265A} in \THESYSTEM:
\begin{equation}\label{example:LipschitzGVI:0}
\begin{array}{l@{}}
\forall{Q \colon \tau_S \times \tau_A \to \REAL}.
K_{d_S,d_{\mathbb{R}}}(f(Q)) \leq \sup_{a \colon \tau_A}K_{d_S,d_{\mathbb{R}}}(\lambda{s \colon \tau_S}.Q(s,a))\\
g = \lambda Q'.\lambda s.~ \mathbb{E}_{s' \sim T(s,a)}[f(Q')(s'))],\forall{s\colon \tau_S}.\forall{a\colon \tau_A}. \mathbb{E}_{s' \sim T(s,a)}[1] = 1,\gamma K_{d_S,W}(T) < 1\\
\vdash_{\mathrm\UPL} \mathtt{LipGVI} \colon \INT \to (\tau_S \times \tau_A \to \REAL) \to (\tau_S \times \tau_A \to \REAL) \mid\\
\forall{Q \colon \tau_S \times \tau_A \to \REAL}.
\forall{\varepsilon \colon \REAL}.
\forall{K_1 \colon \REAL}.
(\varepsilon>0) \land (K_1 = \sup_{a \colon \tau_A}K_{d_S,d_{\mathbb{R}}}(Q)(a))\\
\quad \implies \exists{k \colon \INT}.
\forall{K_2 \colon \REAL}.
\forall{K_3 \colon \REAL}.
\forall{K_4 \colon \REAL}.(K_2 = \sup_{a \colon \tau_A}K_{d_S,d_{\mathbb{R}}}(\mathbf{r}(k))(a))\\
\qquad \land (K_3 = K_{d_S,d_{\mathbb{R}}}(R))
\land (K_4 = K_{d_S,W}(T)) \implies 
K_2 \leq K_3/(1 - \gamma \ast K_4) + \varepsilon
\end{array}
\end{equation}
Here, $d_S$ is a distance function on the state space.
The logical assumptions are the losslessness of the transition dynamics $T$, and the definition $g=\lambda Q'.~\lambda s.~ \mathbb{E}_{s' \sim T(s,a)}[f(Q')(s'))]$.
We also introduce four slack variables $K_1$, $K_2$, $K_3$, and $K_4$ to use the above syntactic sugar for Lipschitz constants. 
The judgment (\ref{example:LipschitzGVI:0}) itself is proved inductively as in
the paper \citep{2018arXiv180601265A}. The key part of the proof is showing the inequality:
\[
K_{d_S,d_{\mathbb{R}}}(\lambda s \colon \tau_S.~ \mathbb{E}_{s' \sim T(s,a)}[f(Q')(s'))])\leq
K_{d_S,d_{\mathbb{R}}}(\lambda{s' \colon \tau_S}.~f(Q')(s'))) \cdot
K_{d_S, W_{d_S}}(\lambda{s\colon \tau_S}.T(s,a))
\]
Suppose
$K_1 = K_{d_S,d_{\mathbb{R}}}(\lambda{s' \colon \tau_S}.~f(Q')(s')))$
and
$K_2 = K_{d_S, W_{d_S}}(\lambda{s\colon \tau_S}.T(s,a))$.
What we prove in our framework is that $z = K_{d_S,d_{\mathbb{R}}}(\lambda s \colon \tau_S.~ \mathbb{E}_{s' \sim T(s,a)}[f(Q')(s'))])$ implies
$z \leq K_1 \ast K_2$.
By unfolding and applying linearity of expectation, we obtain:
\[
\begin{array}{r@{}l}
&z =
K_{d_S,d_{\mathbb{R}}}(\lambda s \colon \tau_S.~ \mathbb{E}_{s' \sim T(s,a)}[f(Q')(s'))])\\
&\quad \iff
z = \sup_{s_1,s_2 \colon \tau_S}\dfrac{K_1 \cdot (\mathbb{E}_{s' \sim T(s_1,a)}[f(Q')(s'))/K_1] - \mathbb{E}_{s' \sim T(s_2,a)}[f(Q')(s'))/K_1])}{d_S(s_1,s_2)}.
\end{array}
\]
Here $1 = K_{d_S,d_{\mathbb{R}}}(\lambda{s' \colon \tau_S}.~f(Q')(s'))/K_1)$ holds from the property 
$d_{\mathbb{R}}(\alpha \cdot x,\alpha \cdot y) = \alpha \cdot d_{\mathbb{R}}(x,y)$ of $d_{\mathbb{R}}$ ($0 \leq \alpha$) and the losslessness $\forall{s\colon \tau_S}.\forall{a\colon \tau_A}.~ \mathbb{E}_{s' \sim T(s,a)}[1] = 1$ of the function $T$.
Hence, we conclude $z \leq K_1 \ast K_2$.

\section{Domain-specific reasoning principles}
Paper proofs of randomized algorithms typically use proof techniques to abstract
away unimportant details.  In this section, we show how \THESYSTEM\ can support
custom proof techniques in the form of domain-specific logics for
reasoning about higher-order programs.
Specifically, we show that the $\top\top$-lifting construction
by~\citet{Katsumata2014PEM}---roughly, a categorical construction useful for
building different refinements of the probability distribution monad---can be
smoothly incorporated in \THESYSTEM. As concrete examples, we instantiate the
unary $\top\top$-lifting construction to a logic for reasoning about the
probability of failure using the so-called union
bound~\citep{BartheGGHS16ICALP}, and the binary $\top\top$-lifting construction
to a logic for reasoning about probabilistic coupling. 

\subsection{Embedding Unary Graded $\top\top$-lifting}
\label{subsection:example:UnaryTT-lifting}


Roughly speaking, the $\top\top$-lifting of a monad is given by a large intersection of inverse images of some predicate, called \emph{lifting parameters}.
We can internalize this construction of $\top\top$-lifting in \THESYSTEM\
using a large intersection of assertions as $\forall x \colon \tau.\phi_x$, and
the inverse image $\phi[e/y]$ of an assertion $\phi$ along an expression $e$.
First, we internalize a general construction of \emph{graded} $\top\top$-lifting  (along the fibration $q\colon\Pred(\QBS) \to \QBS$) in the unary logic \UPL.
Then we instantiate it to construct a unary graded $\top\top$-lifting for reasoning about the probability of failure using a union bound.

These instantiations of $\top\top$-liftings need subprobability
measures. Accordingly, we introduce a new monadic type $D[\tau]$ describing
the set of \emph{subprobability} measures over $\tau$.  We interpret
$D$ by $\interpret{D[\tau]} \defeq \mathfrak{P}\interpret{\tau}$, and
interpret monadic structures in the same way as the ones on the monadic
type $M$.  Furthermore, we assume that for every type $\tau$,
$D[\tau]$ is a subtype of $M[\tau]$.  We introduce the following
axioms 
enabling syntactic conversions from distributions in
$D[\mathtt{unit}]$ to real numbers in $[0,1]$.
\begin{equation}
\label{eq:MHOL:equivalence:value_and_distribution_on_unit_type}
{ 
\AxiomC{$\Gamma \vdash e \colon D[\mathtt{unit}]$}
\UnaryInfC{$\Gamma \vdash_{\mathrm{\NameOfUnderlyingLogic}} e = \mathtt{scale}(\mathtt{return}(\ast),\lambda z \colon{\mathtt{unit}}.~\mathbb{E}_{y \sim e}[1])$}
\DisplayProof
}
\qquad
{  
\AxiomC{$\Gamma \vdash e \colon D[\tau]$}
\UnaryInfC{$\Gamma \vdash_{\mathrm{\NameOfUnderlyingLogic}} 0 \leq \mathbb{E}_{y \sim e}[1] \leq 1$}
\DisplayProof
}
\end{equation}

\paragraph{General Construction}
We define a graded $\top\top$-lifting for the monadic type $D$.
Consider a type $\zeta$ equipped with a preordered monoid structure $(\zeta,1_\zeta,\cdot_\zeta,\leq_\zeta)$,
and an arbitrary type $\theta$.
A \emph{lifting parameter} is a well-typed formula $S$ satisfying
$\Gamma, \mathbf{k} \colon \zeta, \mathbf{l} \colon D[\theta] \vdash S \;\mathsf{wf}$
and the following monotonicity condition:
\[
\Gamma, \mathbf{k} \colon \zeta, \mathbf{l} \colon D[\theta]
\vdash_{\mathrm{\NameOfUnderlyingLogic}}
\forall{\alpha \colon \zeta}. 
\forall{\beta \colon \zeta}. 
(\alpha \leq_\zeta\beta \implies
(S[\alpha/\mathbf{k}]\implies S[\beta/\mathbf{k}]))
\]
Roughly speaking, a lifting parameter is a family of predicates $D[\theta]$
which is monotone with respect to the parameter in the preordered monoid $\zeta$.
The type $\theta$ can differ, depending on the application. For
example, in the embedding of the union bound logic, the type $\theta$ is
set to $\mathtt{unit}$.

For any assertion $\Gamma, \mathbf{r}' \colon \tau \vdash \phi\;\mathsf{wf}$ and an expression
$\Gamma \vdash \alpha \colon \zeta$, we define the \emph{unary $\top\top$-lifting}
$\Gamma, \mathbf{r} \colon D[\tau] \vdash \mathfrak{U}^{\alpha}_S \phi \;\mathsf{wf}$
for the lifting parameter $S$ by the following assertion.
\[
\mathfrak{U}^{\alpha}_S \phi \equiv \forall{\beta \colon \zeta}. \forall{f \colon \tau \to D[\theta]}. ((\forall{x \colon \tau}. \phi[x/\mathbf{r}'] \implies S[\beta/\mathbf{k}, f(x)/\mathbf{l}]) \implies S[\alpha\cdot\beta/\mathbf{k}, \BIND \mathbf{r}~f /\mathbf{l}])
\]
Notice that the parameters $\beta$ and $f$ range over \emph{all} elements in the
types $\zeta$ and $\tau \to D[\theta]$ respectively.  By regarding
$\mathfrak{U}_S$ as a constructor of graded $\top\top$-lifting, we obtain the
following graded monadic rules.
\begin{theorem}[Graded Monadic Laws of $\mathfrak{U}_S$]\label{thm:TT-lifting:monadic_law}
The following rules are derivable:
{\small
\begin{gather*}
{ 
\Gamma \mid \Psi \vdash_{\mathrm{\NameOfUnderlyingLogic}}
\forall{\alpha\colon\zeta}. 
\forall{\beta\colon\zeta}. 
(\alpha \leq_\zeta \beta)
\implies
\mathfrak{U}^{\alpha}_S \phi \implies \mathfrak{U}^{\beta}_S \phi
}
\\
{
\Gamma \mid \Psi \vdash_{\mathrm{\NameOfUnderlyingLogic}}
\forall{\alpha \colon \zeta}. 
(\forall{x \colon \tau}. \phi_1[x/\mathbf{r}'] \implies \phi_2[x/\mathbf{r}'])\implies (\mathfrak{U}^{\alpha}_S \phi_1 \implies \mathfrak{U}^{\alpha}_S \phi_2)
}
\\
\vspace{0.7em} 
{ 
\AxiomC{
\begin{tabular}{c}
$ $\\
$
\Gamma \mid \Psi \vdash_{\mathrm{\NameOfUnaryLogic}} e \colon \tau \mid \phi[\mathbf{r}/\mathbf{r}']
$
\end{tabular}}
\UnaryInfC{$\Gamma \mid \Psi \vdash_{\mathrm{\NameOfUnaryLogic}} \RETURN(e) \colon D[\tau] \mid \mathfrak{U}^{1_\zeta}_S \phi$}
\DisplayProof
}
\quad
{ 
\AxiomC{
\begin{tabular}{c}
$\Gamma \mid \Psi \vdash_{\mathrm{\NameOfUnaryLogic}} e \colon D[\tau]\mid
\mathfrak{U}^{\alpha}_S \phi$
\\
$\Gamma \mid \Psi \vdash_{\mathrm{\NameOfUnaryLogic}} e' \colon \tau \to D[\tau'] \mid \forall{x \colon \tau}. \phi[x/\mathbf{r}]{\implies} (\mathfrak{U}^{\beta}_S\phi')[\mathbf{r} x/\mathbf{r}]
$
\end{tabular}
}
\UnaryInfC{$\Gamma \mid \Psi \vdash_{\mathrm{\NameOfUnaryLogic}}\BIND e~e' \colon D[\tau'] \mid \mathfrak{U}^{\alpha\cdot\beta}_S\phi'$}
\DisplayProof
}
\end{gather*}
}
\end{theorem}
The proofs follow by unfolding the constructor $\mathfrak{U}_S$.
Furthermore, the graded monadic laws (Theorem \ref{thm:TT-lifting:monadic_law})
are proved only using the structure of the preordered monoid for grading, the monotonicity of the lifting parameter, $\alpha$-conversions, $\beta\eta$-reductions, the monadic laws of $D$, and proof rules of \PL.
Moreover, the construction of $\top\top$-lifting can be applied to any monadic type.

\paragraph{Embedding the Union Bound Logic}
\label{subsec:TT-lifting:unionbound}
We show that the predicate lifting given in the semantic model of the union
bound logic~\citep{BartheGGHS16ICALP} can be implemented as a graded unary
$\top\top$-lifting in \THESYSTEM.  Concretely, we give a lifting
parameter $S$ such that the graded $\top\top$-lifting $\mathfrak{U}_S$
corresponds to the predicate lifting for the union bound logic.

Consider the additive monoid structure with usual ordering $(\mathtt{pReal},0,+,\leq)$. 
We define the lifting parameter $\mathbf{k} \colon \mathtt{pReal}, \mathbf{l} \colon D[\mathtt{unit}] \vdash S \;\mathsf{wf}$ by $S = (\mathbb{E}_{y \sim \mathbf{l}}[1] \leq \mathbf{k})$. The monotonicity of $S$ is obvious.
As we proved above, we have the monadic rules for the graded $\top\top$-lifting $\mathfrak{U}_S$.
Next, we prove that the graded $\top\top$-lifting $\mathfrak{U}_S$ describes the probability of failure:
\begin{proposition}
The following reduction is derivable in \NameOfUnderlyingLogic.
\[
\AxiomC{
 $\Gamma, \mathbf{r'}\colon \tau \vdash e \colon \BOOL $
\quad $\Gamma, \mathbf{r'}\colon \tau \mid \Psi \vdash_{\mathrm{\NameOfUnderlyingLogic}} \neg \phi \iff (e = \mathtt{true})$
}
\UnaryInfC{
$\Gamma, \mathbf{r} \colon D[\tau] \mid \Psi \vdash_{\mathrm{\NameOfUnderlyingLogic}}
\mathfrak{U}_S^\alpha(\neg \phi) \iff \Pr_{X \sim \mathbf{r}}[e[X/\mathbf{r}']] \leq \alpha$}
\DisplayProof
\]
\end{proposition}
Intuitively, this proposition holds because $\mathfrak{U}_S^\alpha(\neg \phi)
\iff \Pr[\phi] \leq \alpha$. The second premise requires $\phi$ to be a
measurable assertion, i.e., there is an indicator function $\lambda
\mathbf{r'}\colon \tau. e$ of $\phi$.

\subsection{Embedding Relational $\top\top$-lifting}
\label{subsection:example:RelationalTT-lifting}
Similar to unary graded $\top\top$-lifting, we can also define relational graded
$\top\top$-lifting by just switching the type of assertions from predicates to relations.
As a concrete example, we instantiate the (non-graded)
relational $\top\top$-lifting for reasoning about probabilistic coupling. 
%
Consider a preordered monoid $(\zeta,1_\zeta,\cdot_\zeta,\leq_\zeta)$ and a pair of arbitrary types $\theta_1$ and $\theta_2$.
A lifting parameter for relational $\top\top$-lifting is a well-typed formula of the form
$\Gamma, \mathbf{k} \colon \zeta, \mathbf{l}_1 \colon D[\theta_1], \mathbf{l}_2 \colon D[\theta_2] \vdash S \;\mathsf{wf}$
satisfying the following monotonicity condition:
\[
\Gamma, \mathbf{k} \colon \zeta, \mathbf{l}_1 \colon D[\theta_1], \mathbf{l}_2 \colon D[\theta_2] 
\vdash_{\mathrm{\NameOfUnderlyingLogic}}
\forall{\alpha \colon \zeta}. 
\forall{\beta \colon \zeta}. 
(\alpha \leq_\zeta\beta \implies
(S[\alpha/\mathbf{k}]\implies S[\beta/\mathbf{k}])).
\]
Intuitively, a lifting parameter for relational graded $\top\top$-lifting is
a monotone family of relations between $D[\theta_1]$ and $D[\theta_2]$
with respect to the preordered monoid $\zeta$.

For any assertion $\Gamma, \mathbf{r}' \colon \tau_1,\mathbf{r}' \colon \tau_2 \vdash \phi \;\mathsf{wf}$ and an expression $\Gamma \vdash \alpha \colon \zeta$ we define its \emph{relational lifting} $\Gamma, \mathbf{r}_1 \colon D[\tau_1],\mathbf{r}_2 \colon D[\tau_2] \vdash \mathfrak{R}^{\alpha}_S \phi \;\mathsf{wf}$ for the lifting parameter $S$ as the following assertion.
\[
\begin{array}{lr@{}l}
&
\mathfrak{R}^{\alpha}_S \phi \equiv & \; \forall{\beta \colon \zeta}. \forall{f_1 \colon \tau_1 \to D[\theta_1]}. \forall{f_2 \colon \tau_2 \to D[\theta_2]}.\\
& &\qquad (\forall{x_1 \colon \tau_1}. \forall{x_2 \colon \tau_2}. \phi[x_1/\mathbf{r}'_1, x_2/\mathbf{r}'_2] \implies S[\beta/\mathbf{k}, f_1(x_1)/\mathbf{l}_1, f_2(x_2)/\mathbf{l}_2]) \\
& & \qquad \qquad \implies S[\alpha\cdot\beta/\mathbf{k}, \BIND \mathbf{r}_1~ f_1 /\mathbf{l}_1, \BIND \mathbf{r}_2~f_2 /\mathbf{l}_2])).
\end{array}
\]
We can prove two-sided graded monadic laws for $\mathfrak{R}_S$, analogous to those for graded unary $\top\top$-lifting $\mathfrak{U}_S$.
We omit these here.
\paragraph{Embedding the Modality for Relational Coupling of Distributions}

As an example of this relational construction, we show how to
internalize in our framework the modality for
relational probabilistic coupling defined by
\citet{DBLP:conf/esop/0001BBBG018}.
We say that two probability distributions $\mu_1$ and $\mu_2$ are \emph{coupled} over a relation $R \subseteq X \times Y$ if
$
\forall{S \subseteq X}. \Pr_{x \sim \mu_1}[S] \leq \Pr_{y \sim \mu_2}[R(S)] 
$.
%
To internalize this construction we now need to supply the  appropriate lifting
parameters.  First, we take the grading monoid to be the trivial one on the unit
type $\mathtt{unit}$.  Then, we set the lifting parameter $\mathbf{k} \colon
{\mathtt{unit}},\mathbf{l}_1 \colon D[\mathtt{unit}], \mathbf{l}_2 \colon
D[\mathtt{unit}] \vdash S  \;\mathsf{wf}$ by $S = (\mathbb{E}_{y \sim
\mathbf{l}_1}[1] \leq \mathbb{E}_{y \sim \mathbf{l}_2}[1])$, which is equivalent
to the usual inequality on $[0,1]$.  The assertion $S$ obviously satisfies the
monotonicity of lifting parameter.  Hence we obtain the $\top\top$-lifting
$\Gamma, \mathbf{r}_1 \colon D[\tau_1],\mathbf{r}_2 \colon D[\tau_2] \vdash
\mathfrak{R}_S \phi$ for the lifting parameter $S$.
What we need to prove is that the lifting $\mathfrak{R}_S$ actually describes
the above inequality of probabilistic dominance. In other words, we need to
prove the following fundamental property in \NameOfUnderlyingLogic.
\begin{proposition}[{\citet[Lemma 2]{DBLP:conf/esop/0001BBBG018}}]
\[
\begin{array}{r@{}l}
&\Gamma, \mathbf{r}_1 \colon D[\tau_1],\mathbf{r}_2 \colon D[\tau_2]\mid\Psi\vdash_{\mathrm{\NameOfUnderlyingLogic}}\mathfrak{R}_S \phi\implies \forall{f_1 \colon \tau_1 \to \BOOL }. \forall{f_2 \colon \tau_2 \to \BOOL }.\\
&\qquad\qquad \forall{y \colon \tau_2}. 
((f_2(y) = \mathtt{true})
\implies
\forall{x \colon \tau_1}. \phi[x/\mathrm{r}'_1,y/\mathrm{r}'_2] \implies (f_1(x) = \mathtt{true})
)\\
&\qquad\qquad\quad\implies \Pr_{x \sim \mathbf{r}_1 }[f_1(x)] \leq \Pr_{y \sim \mathbf{r}_1}[f_2(y)]
\end{array}
\]
\end{proposition}
Intuitively $f_1$ and $f_2$ encode indicator functions $\chi_A$ and $\chi_B$ respectively, where $\phi(A) \subseteq B$.
The proof follows \citet[Theorem 12]{katsumata_et_al:LIPIcs:2015:5532}, again
using the equivalence $D[\mathtt{unit}] \cong [0,1]$ axiomatized in
(\ref{eq:MHOL:equivalence:value_and_distribution_on_unit_type}) and axioms on
scaling of measures.

Specializing the assertion $\phi$ can establish useful
probabilistic properties. For instance, taking $\phi$ to be the equality
relation yields the following property.
\begin{corollary}[{\citet[Corollary 1]{DBLP:conf/esop/0001BBBG018}}]
If $\tau_1 = \tau_2 = \tau$ then
\[
\begin{array}{r@{}l}
&\Gamma, \mathbf{r} \colon D[\tau] \mid \Psi\vdash_{\mathrm{\NameOfUnderlyingLogic}}\\
&\quad \mathfrak{R}_S (\mathbf{r}'_1 = \mathbf{r}'_2) \iff (\forall{g \colon \tau\to\REAL }. 
(\forall{x \colon \tau}. 0 \leq g(x) \leq 1) \implies
\mathbb{E}_{x \sim \mathbf{r}_1}[g(x)] \leq \mathbb{E}_{y \sim \mathbf{r}_2}[g(y)])
\end{array}
\]
\end{corollary}
If we take $g$ to be the indicator function of a (measurable) set $A$, the
conclusion shows that the measure of $A$ in $\mathbf{r}_1$ is smaller than the
measure of $A$ in $\mathbf{r}_2$. Since the assertion $\phi$ is symmetric, we
can also conclude the inequality in the other direction, hence showing that the
measure of $A$ must be equal in $\mathbf{r}_1$ and in $\mathbf{r}_2$. Since
equality holds for all measurable $A$, $\mathbf{r}_1$ and $\mathbf{r}_2$ must
denote equal probability measures.

\section{Related Work}

\paragraph{Semantics of probabilistic programs}
The semantics of probabilistic programs has been extensively studied
starting from the seminal work of~\citet{KOZEN1981328}.  Imperative
first-order programs with continuous distributions have a 
well-understood interpretation based on the Giry monad~\citep{Giry1982}
over the category $\Meas$ of measurable spaces and measurable
functions~\citep{Panangaden1999171}.  However, this approach does not
naturally extend to the higher-order setting since $\Meas$ is not
Cartesian closed~\citep{aumann1961}.  In addition, although $\Meas$ has
a symmetric monoidal closed structure~\citep{2013arXiv1312.1445C}, the
Giry monad is not strong with respect to the canonical
one~\citep{SATO20182888}. 

The category $\QBS$~\citep{HeunenKSY17} of quasi-Borel spaces was
introduced as an ``extension'' of $\Meas$ that is Cartesian closed and
that can be used to interpret higher-order probabilistic programs
with continuous distributions.  The category of s-finite
kernels~\citep{Staton17ESOP} gives a denotational semantics to
observe-like statements in these models, including our construct $\OBSERVE$. In particular, it supports
infinite measures and rescaling of measures. These are useful to give
semantics to programs and to devise equational rules to reason about
the equivalence of programs.  
 The monad $\mathfrak{M}$
of measures on quasi-Borel spaces we use in this paper was introduced by~\citet{Scibior:2017:DVH:3177123.3158148} based on these
constructions. 
One reason we chose $\QBS$ is that it has an
obvious forgetful functor $\QBS \to \Sets$ giving the identity on
functions. This is a key property to allow set-theoretic reasoning
in \THESYSTEM.

An alternative approach has been
proposed by~\citet{Ehrhard:2017:MCS:3177123.3158147}. They use a
domain-theoretic approach based on the category $\mathbf{Cstab}$ of
cones and stable functions, extending previous work on
probabilistic coherent spaces~\citep{EhrhardTP14}.  For our work, $\QBS$ is a
more natural choice than $\mathbf{Cstab}$ for two reasons.
First, the categorical structure needed for $\OBSERVE$-like
statements has already been studied in $\QBS$. Second, we are
interested in terminating programs and so we do not need the
domain-theoretic structure of $\mathbf{Cstab}$. Other models related
to both $\QBS$ and $\mathbf{Cstab}$ that one could consider are the ones by~\citet{TixKP09a} and \citet{KeimelP16}.
Several other papers have studied models for higher-order probabilistic programming starting from the seminal papers  on probabilistic powerdomains
by~\citet{JonesP89} and \citet{Saheb-Djahromi80}. A
non-exhaustive list includes \citet{JungT98,VaraccaVW04,Goubault-LarrecqJ14,Mislove17,CastellanCPW18}. Many
of these model only partially support the features we need.
There is also recent work that studies the semantics of
probabilistic programming from an operational perspective. 
\citet{BorgstromLGS16} propose distribution-based and sample-based
operational semantics for an untyped lambda calculus with continuous
random variables. Their calculus also contains primitives for scaling
and failing which allow them to model different kinds of $\OBSERVE$-like constructs.
\citet{CulpepperC17} propose an entropy-based operational
semantics for a simply typed lambda calculus with continuous random
variables and propose an operational equational theory for it based on logical
relations. The focus of their work is program equivalence, as reflected in the form of their judgments. In contrast, we start from an expressive
predicate logic for probabilistic computations which allows us to
express many different (unary and relational) properties, not just equivalence.
\paragraph{Verification of probabilistic programs}
Starting from the seminal work on Probabilistic Propositional Dynamic Logic by \citet{Kozen85}, several papers have
proposed program logics for the verification of imperative
probabilistic programs. \citet{Morgan96, McIverM05} propose a
predicative logic to
reason  about an imperative
language with probabilistic and non-deterministic choice. Both these
program logics allow reasoning about the expected value of a
single real-valued function on program states. Many subsequent papers build on
this
idea~\citep{HurdMM05,KatoenMMM10,GretzKM13,AudebaudP09,KaminskiKMO16}. 
Other papers focus on program logics where the pre-condition and post-condition are probabilistic
assertions about the input and output
distributions~\citet{Ramshaw79,Hartog:thesis,ChadhaCMS07,RandZ15}. 
\citet{BEGGHS18} propose an assertion-based logic, named ELLORA, using
expectation for
verifying properties of imperative probabilistic programs with
discrete random variables.  Our assertion logic \PL\ is similar in spirit
to the one of ELLORA, but it further supports continuous distributions and the
verification of higher-order programs. On the other hand, ELLORA has powerful rules for
probabilistic while loops that \THESYSTEM\ does not support. It would be
interesting to explore if similar rules can also be
added to \THESYSTEM.
Formalizations of measure and integration theory in general purpose
interactive theorem provers have been considered in many papers
\citep{AudebaudP09,Hurd03,Richter04,Coble10,HolzlH11}.
\citet{AvigadHS14} recently completed a proof of the Central Limit
theorem, which is the principle underlying concentration bounds.
\citet{holzl2016markov} formalized discrete-time Markov chains and
Markov decision processes. These and other existing formalizations
have been used to verify several case studies, but they are scattered
and not easily accessible for our purposes.
\paragraph{Relational Verification}
Several papers have explored relational program logics or
relational type systems for the verification of different
relational properties.
\citet{Aguirre:2017:RLH:3136534.3110265} propose UHOL/RHOL for the
unary and relational verification of higher-order, non-probabilistic,
terminating programs. UHOL and RHOL are based on a combination of logics for expressing (unary and relational)
postconditions, and syntax-directed proof rules for establishing
them. Since only terminating non-probabilistic programs are
considered, the logic and the proof rules can be shown
sound in set-theory. Our broad approach to setting up \THESYSTEM\ is
directly inspired from this work, but we work with probabilistic
programs and, therefore, introduce a new monadic type for
general/continuous measures along with constructs for conditioning. As
a result, we have to interpret the logic and proof system in $\QBS$,
not set theory, and had to re-work the entire soundness proof from
scratch.

The framework $U^C$/$R^C$~\citep{Radicek:2017:MRR:3177123.3158124} is an extention of \citet{Aguirre:2017:RLH:3136534.3110265} for reasoning about costs of non-probabilistic, terminating programs. This work introduces a monad, but this monad merely pairs a computation with its cost. The entire development still has a simple model in set theory.
%
%
GUHOL and GRHOL~\citep{DBLP:conf/esop/0001BBBG018} are extensions of
\citet{Aguirre:2017:RLH:3136534.3110265} to reason about unary and
relational properties of Markov chains. These systems include a monad
for distributions, but the development is limited to discrete
distributions, and relational probabilistic reasoning is limited to
coupling. The framework has an interpretation in the topos of trees
(which is an extension of set theory with step counting) extended with
a Giry monad. Importantly, pre- and post-conditions are
\emph{non}-probabilistic and are interpreted first over deterministic
values, and then over distributions over values by lifting. Moreover,
in~\citep{DBLP:conf/esop/0001BBBG018}, the proof rules only allow
analysis of properties of programs via coupling arguments.  This
differs considerably from what we present here. Indeed, in
\THESYSTEM\ we can reason about (monadic) probabilistic expressions
directly in assertions. For example, we can directly express and prove
convergence properties of the expectation of an expression, which is
impossible in the work
of~\citet{DBLP:conf/esop/0001BBBG018}. In
addition,~\citep{DBLP:conf/esop/0001BBBG018} support only discrete
distributions while we handle continuous distributions. As we have
shown, the probabilistic coupling of GRHOL can be embedded in
\THESYSTEM\ by $\top\top$-lifting, but \THESYSTEM\ does not cover all
features of GRHOL.  The reason is the difference in the goals of
verification: \THESYSTEM\ verifies the static behavior of
probabilistic programs such as expected values and equality between
probability measures. In contrast, GRHOL verifies behaviors of entire
Markov chains.

\citet{BartheFGAGHS16} study a relational type system
PrivInfer for Bayesian inference on a functional
programming language.  Our framework \THESYSTEM\ is more flexible
since it supports continuous probability distributions while
PrivInfer supports only discrete probabilities.  In the
future, we expect to internalize the continuous variant
of PrivInfer's $(f,\delta)$-lifting proposed in \THESYSTEM, in a manner
similar to $\top\top$-lifting.
\section{Conclusion}
In this paper we have introduced a framework \THESYSTEM\ supporting
the (unary and relational) verification of probabilistic programs including constructions for
higher-order computations, continuous distributions and
conditioning. \THESYSTEM\ combines axiomatizations of basic
probabilistic constructions with rules of three different logics in
order to ease the verification of examples from probabilistic
inference, statistics, and machine learning. The soundness of our
approach relies on quasi-Borel spaces, a recently proposed semantic
framework for probabilistic programs. All these components make
\THESYSTEM\ a useful framework for the practical verification of higher-order
probabilistic programs.
\begin{acks}       
This material is based on work supported by the
\grantsponsor{GS100000001}{National Science
Foundation}{http://dx.doi.org/10.13039/100000001} under CCF Grant
No.~\grantnum{GS100000001}{1637532} and under CNS Grant
No.~\grantnum{GS100000001}{1565365}.
\end{acks}
\bibliography{reference}
 \newpage
 \appendix
\section{Proofs in \THESYSTEM}
\subsection{More Detailed Proof on Monte Carlo Approximation}
We first show the following judgment and then we apply Chebyshev's inequality.
\begin{align*}
&d\colon M[\tau],~ h\colon \tau \to \mathtt{real}\vdash_{\mathrm\UPL}\\
&\quad \mathtt{letrec}~f(i \colon \INT)\\
&\qquad = \mathtt{if}~(i \leq 0)~\mathtt{then}~{\mathtt{return}(0)}\\
&\qquad\quad \mathtt{else}~\mathtt{mlet}~m=f(i-1)~\mathtt{in}~\mathtt{mlet}~x=d~\mathtt{in}~\mathtt{return}(\frac{1}{i}(h(x)+m\ast(i-1)))\\
&\quad\colon \INT \to M[\mathtt{real}] \mid\\
&\quad\forall n \colon \INT.~\forall \sigma\colon \mathtt{real}.~\forall \mu \colon \mathtt{real}.\\
&\qquad(\mathbb{E}_{x \sim d}[1] = 1) \land (n > 0) \land (\sigma^2 = \mathrm{Var}_{x \sim d}[h(x)]) \land (\mu = \mathbb{E}_{x \sim d}[h(x)])\\
&\quad\qquad \implies 
(\mathbb{E}_{y \sim \mathbf{r}n}[1] = 1)
\land 
(\mathrm{Var}_{y \sim \mathbf{r}n}[y] = \frac{\sigma^2}{n})
\end{align*}
We split the program and postcondition as follows
\begin{align*}
\mathtt{MonteCarlo} &\equiv \mathtt{letrec}~f(i) = e_{\mathrm{body}}\\
e_{\mathrm{body}} &\equiv \mathtt{if}~(i \leq 0)~\mathtt{then} ~e_{\mathrm{body}0}~\mathtt{else}~e_{\mathrm{body}1}\\
e_{\mathrm{body}0} &\equiv \mathtt{return}(0)\\
e_{\mathrm{body}1} &\equiv \mathtt{mlet}~m = f(i-1)~\mathtt{in}~(\mathtt{mlet}~x=d~\mathtt{in}~\mathtt{return}(\frac{1}{i}(h(x)+m\ast(i-1)))) \\
\\
\phi   &\equiv\forall \sigma\colon \mathtt{real}.~\forall \mu \colon \mathtt{real}.(\phi_0 \implies \phi_1)\\
\phi_0 &\equiv(\mathbb{E}_{x \sim d}[1] = 1) \land (i > 0) \land (\sigma^2 = \mathrm{Var}_{x \sim d}[h(x)]) \land (\mu = \mathbb{E}_{x \sim d}[h(x)])\\
\phi_1 &\equiv(\mathbb{E}_{y \sim \mathbf{r}i}[1] = 1) \land (\mathrm{Var}_{y \sim \mathbf{r}i}[y] = \frac{\sigma^2}{i}).
\end{align*}
What we want to show is:
\begin{equation}
\label{eq:prooftree:WLLN:0}
d\colon M[\tau],~ h\colon \tau \to \mathtt{real} \vdash_{\mathrm\UPL} \mathtt{MonteCarlo} \colon \INT\to M[\mathtt{real}] \mid \forall n \colon \INT.~\phi[n/i].
\end{equation}
To show (\ref{eq:prooftree:WLLN:0}) by applying [u-RETREC] rule, which we have the following premise:
\begin{equation}
\label{eq:prooftree:WLLN:1}
\begin{aligned}
&d\colon M[\tau],~ h\colon \tau \to \mathtt{real},f\colon \INT \to M[\mathtt{real}], i \colon \INT \mid \phi_{\mathrm{ind.hyp}} \vdash_{\mathrm\UPL}\\
&\quad \mathtt{if}~(i \leq 0)~\mathtt{then}~e_{\mathrm{body}0}~\mathtt{else}~e_{\mathrm{body}1}  \colon M[\mathtt{real}]\mid\\
&\quad \forall \sigma\colon \mathtt{real}.~\forall \mu \colon \mathtt{real}.(\phi_0 \implies \phi_1).
\end{aligned}
\end{equation}
Here, we write 
\[
\phi_{\mathrm{ind.hyp}} = (\forall l \colon \INT.~l < i \implies (\sigma\colon \mathtt{real}.~\forall \mu \colon \mathtt{real}.(\phi_0 \implies \phi_1))[l/i,~f(l)/\mathbf{r}])
\]

To show (\ref{eq:prooftree:WLLN:1}) by applying [u-CASE], and we need to show
\begin{align}
\label{eq:prooftree:WLLN:20}
&
\begin{aligned}
&d\colon M[\tau],~ h\colon \tau \to \mathtt{real},f\colon \INT \to M[\mathtt{real}], i \colon \INT\mid (i\leq 0) \land \phi_{\mathrm{ind.hyp}}\vdash_{\mathrm\UPL}\\
&\quad \mathtt{return}(0) \colon M[\mathtt{real}] \mid \forall \sigma\colon \mathtt{real}.~\forall \mu \colon \mathtt{real}.(\phi_0 \implies \phi_1)\\
\end{aligned}
\\
\notag\\
\label{eq:prooftree:WLLN:21}
&
\begin{aligned}
&d\colon M[\tau],~ h\colon \tau \to \mathtt{real},f\colon \INT \to M[\mathtt{real}], i \colon \INT\mid (i > 0) \land \phi_{\mathrm{ind.hyp}}\vdash_{\mathrm\UPL}\\
&
\quad \mathtt{mlet}~m = f(i-1)~\mathtt{in}~(\mathtt{mlet}~x=d~\mathtt{in}~\mathtt{return}(\frac{1}{i}(h(x)+m\ast(i-1)))) \colon M[\mathtt{real}]
\mid\\
&\quad \forall \sigma\colon \mathtt{real}.~\forall \mu \colon \mathtt{real}.(\phi_0 \implies \phi_1).
\end{aligned}
\end{align}
The judgment (\ref{eq:prooftree:WLLN:20}) is shown by applying [u-SUB] having the following \PL-premise:
\begin{align*}
&d\colon M[\tau],~ h\colon \tau \to \mathtt{real},f\colon \INT \to \mathtt{list}(M[\mathtt{real}]), i \colon \INT\mid (i\leq 0) \land \phi_{\mathrm{ind.hyp}}\vdash_{\mathrm{\PL}}\\ &\quad \forall \sigma\colon \mathtt{real}.~\forall \mu \colon \mathtt{real}.(\phi_0 \implies \phi_1).
\end{align*}
In fact, this is a tautology since
\[
\phi_0 = (\mathbb{E}_{x \sim d}[1] = 1) \land (i > 0) \land (\sigma^2 = \mathrm{Var}_{x \sim d}[h(x)]) \land (\mu = \mathbb{E}_{x \sim d}[h(x)]).
\]
The premise (\ref{eq:prooftree:WLLN:21}) is proved by applying [u-SUB] rule having the following premise:
\[
\begin{aligned}
&d\colon M[\tau],~ h\colon \tau \to \mathtt{real},f\colon \INT \to M[\mathtt{real}], i \colon \INT \mid\\
&\quad (i > 0)\land \phi_{\mathrm{ind.hyp}} \vdash_{\mathrm{\PL}}\\ 
&\qquad\quad
\forall \sigma\colon \mathtt{real}.~\forall \mu \colon \mathtt{real}.\\
&\qquad\qquad\phi_0 \implies 
\mathbb{E}_{y \sim (\mathtt{mlet}~m=f(i-1) ~\mathtt{in}~ \mathtt{mlet}~x=d~\mathtt{in}~\mathtt{return}(\frac{1}{i}(h(x)+m(i-1))))}[1] = 1\\
&\qquad\qquad
\land \mathrm{Var}_{y \sim (\mathtt{mlet}~m=f(i-1) ~\mathtt{in}~ \mathtt{mlet}~x=d~\mathtt{in}~\mathtt{return}(\frac{1}{i}(h(x)+m(i-1))))}[y] = \sigma^2/i.
\end{aligned}
\]
To show this by applying [$\forall_I$] rule twice and [$\implies_I$] rule, it suffices to show
\[
\begin{aligned}
&d\colon M[\tau],~ h\colon \tau \to \mathtt{real},\sigma\colon \mathtt{real},\mu \colon \mathtt{real},f\colon \INT \to M[\mathtt{real}], i \colon \INT \mid\\
&\quad \phi_0 \land (i > 0) \land \phi_{\mathrm{ind.hyp}} \vdash_{\mathrm{MHOL}}\\ 
&\qquad\qquad \mathbb{E}_{y \sim (\mathtt{mlet}~m=f(i-1) ~\mathtt{in}~ \mathtt{mlet}~x=d~\mathtt{in}~\mathtt{return}(\frac{1}{i}(h(x)+m(i-1))))}[1] = 1\\
&\qquad\qquad
\land \mathrm{Var}_{y \sim (\mathtt{mlet}~m=f(i-1) ~\mathtt{in}~ \mathtt{mlet}~x=d~\mathtt{in}~\mathtt{return}(\frac{1}{i}(h(x)+m(i-1))))}[y] = \sigma^2/i.
\end{aligned}
\]
This is proved by applying (\ref{eq:MHOL:expected:linearity}),(\ref{eq:MHOL:expected:variable_transformation}),(\ref{eq:MHOL:commutativity}), and elementary calculations.

For example, to show $\mathrm{Var}_{y \sim (\mathtt{mlet}~m=f(i-1)~\mathtt{in}~ \mathtt{mlet}~x=d~\mathtt{in}~\mathtt{return}(\frac{1}{i}(h(x)+m(i-1))))}[y] = \sigma^2/i$, we calculate by applying [SUBST] rule in \PL\ with equations as follows:
\begin{align*}
\lefteqn{ \mathrm{Var}_{y \sim\mathtt{mlet}~m=f(i-1)~\mathtt{in}~\mathtt{mlet}~x=d~\mathtt{in}~\mathtt{return}(\frac{1}{i}(h(x)+m(i-1)))}[y]}\\
&\qquad \{\text{Equation (\ref{eq:MHOL:commutativity}). We mention the Independence of past average $m$ and new sample $x$.}\}\\
&= \mathrm{Var}_{y \sim (\mathtt{mlet}~w = d \otimes f(i-1)~\mathtt{in}~\mathtt{return}(\frac{1}{i}(h(\pi_1(w))+\pi_2(w)(i-1))))}[y]\\
& \qquad \{\text{Equation (\ref{eq:MHOL:expected:variable_transformation}): variable transformations.}\}\\
&= \mathrm{Var}_{w \sim d \otimes f(i-1)}[\frac{1}{i}(h(\pi_1(w))+\pi_2(w)(i-1)]\\
& \qquad \{\text{Equation (\ref{eq:MHOL:expected:linearity}). We use here the Independence and $\mathbb{E}_{x \sim d}[1] = 1$ and $\mathbb{E}_{m \sim f(i-1)}[1] = 1$.}\}\\
&= \mathrm{Var}_{x \sim d}[\frac{1}{i}h(x)] + \mathrm{Var}_{m \sim f(i-1)}[\frac{i-1}{i} m]\\
& \qquad \{\text{Equation (\ref{eq:MHOL:expected:linearity}). We use again $\mathbb{E}_{x \sim d}[1] = 1$ and $\mathbb{E}_{m \sim r}[1] = 1$.}\}\\
&= \frac{1}{i^2}\mathrm{Var}_{x \sim d}[h(x)] + \frac{(i-1)^2}{i^2}\mathrm{Var}_{m \sim f(i-1)}[w]\\
& \qquad \{\text{Combining assumptions in $\phi_0$ and $\phi_{\mathrm{ind.hyp}}$.}\}\\
&= \frac{1}{i^2} \sigma^2 + \frac{(i-1)^2}{i^2} \frac{\sigma^2}{i-1} \tag{\ddag}\\
& \qquad \{\text{Doing elementary calculations.}\}\\
&= \frac{\sigma^2}{i} \tag{\dag}
\end{align*}
To obtain (\dag), precisely, we need further case analysis with $i > 1$ and $i = 1$. If $i = 1$ then we have (\dag) without (\ddag).  
This can be done by [$\Rightarrow_I$] rule in \PL\ and a basic tautology.

Similarly, $\mathbb{E}_{y \sim (\mathtt{mlet}~m=f(i-1) ~\mathtt{in}~ \mathtt{mlet}~x=d~\mathtt{in}~\mathtt{return}(\frac{1}{i}(h(x)+m(i-1))))}[1] = 1$ is proved by applying [SUBST] rule with (\ref{eq:MHOL:expected:linearity}),(\ref{eq:MHOL:expected:variable_transformation}),(\ref{eq:MHOL:commutativity}) and elementary calculations.
To sum up, we obtain
\begin{align*}
&\{
d\colon M[\tau],~ h\colon \tau \to \mathtt{real},
\} \vdash e \colon \INT \to M[\mathtt{real}] \mid\\
&\quad 
\forall n \colon \INT.~\sigma\colon \mathtt{real}.~\mu\colon \mathtt{real}\\
&\qquad (\mathbb{E}_{x \sim d}[1] = 1) \land (n > 0) \land (\sigma^2 = \mathrm{Var}_{x \sim d}[h(x)])\land (\mu = \mathbb{E}_{x \sim d}[h(x)])\\
&\qquad\quad \implies 
\mathbb{E}_{x\sim\mathbf{r}n}[1] = 1 \land \mathrm{Var}_{y \sim \mathbf{r}n}[y] = \sigma^2/n.
\end{align*}
We also have Chebyshev's inequality (we prove later):
\begin{align*}
&\{
d \colon M[\mathtt{real}],~
\varepsilon \colon \mathtt{real},~
\mu \colon \mathtt{real}
\}\vdash_{\mathrm{MHOL}}
(\mu = \mathbb{E}_{y \sim d}[y]) \land (\mathbb{E}_{x \sim d}[1] = 1) \land  (\varepsilon^2 > 0)\\
&\qquad \implies \Pr_{y \sim d}[ |y - \mu| \geq \varepsilon] \leq \mathrm{Var}_{y\sim d}[y]/\varepsilon^2.
\end{align*}
By combining them we conclude what we desired.

\subsection{Verification Example: Mean Estimation of Gaussian Distributions}
\label{subsection:example:MeanEstimationOfGauss}
So far, we have shown how to use \THESYSTEM\ to reason about
probabilistic programs using observe statements to describe
Bayesian models. We now want to show it useful also to reason
about statistical tasks that are not based on Bayesian update.  

As a first example, we show that we can use \THESYSTEM\ to prove the
correctness of iterative mean estimations for Gaussian distributions. Here
mean estimation is formulated in terms of a list of confidence intervals over
the empirical mean observed over a set of sample. The correctness
guarantees that these are indeed the right confidence intervals if the
data actually come from the distribution. This example shows that we
can use \THESYSTEM\ to reason naturally also about unary iterative
properties, and that we can use it to reason about standard
statistical tools like confidence intervals. 

First, we consider the following implementation $\mathtt{GaussMean}$ of the iterative mean estimation of Gaussian distribution with given variance $s$.
The algorithm $\mathtt{GaussMean}$ receives an integer $i$ indicating the number of iterations and returns a list
of length $i$, containing at each position $j$ the estimation after the first $j$ samplings. On each iteration,
the algorithm samples an $x$ from $d$ (supposed to be a $\mathtt{Gauss}(\mu,s)$ with unknown $\mu$) and updates
the previous estimation $\lrangle{\overline{x},l,u}$ of the empirical mean $\overline{x}$ and confidence interval $[l,u]$. 
\begin{align*}
&\LETREC f(i \colon \INT) = \IF(i \leq 0)\mathtt{then}{[]}\\
&\qquad \ELSE (\CASE f(i-1) \WITH\\
&\qquad\qquad\;\; [] \Rightarrow \MLET x=d \IN \RETURN (\lrangle{\frac{x}{i},\frac{x}{i}-z\sqrt{{s/i}},\frac{x}{i}+z\sqrt{{s/i}}})\\
&\qquad\quad r :: \xi'\Rightarrow 
\MLET m=(\MLET y=r \IN \RETURN (\pi_1(y))) \IN \MLET x=d \IN\\
&\qquad\qquad\RETURN (\lrangle{\frac{1}{i}(x+m\cdot(i-1)),\frac{1}{i}(x+m\cdot(i-1))-z\sqrt{{s/i}},\frac{1}{i}(x+m\cdot(i-1))+z\sqrt{{s/i}}})\\
&\quad\qquad\quad)::(r :: \xi')
\end{align*}
We show that $[l,u]$ forms an actually confidence interval of $\mu$ (i.e. $\Pr[l \leq \mu \leq u] \geq g$) for each step update of $\lrangle{\overline{x},l,u}$.
We will prove the following unary \UPL\ judgment.
\begin{align*}
&
(d = \mathtt{Gauss}(\mu,s)),(s > 0),(\Pr_{w\sim\mathtt{Gauss}(0,1)}[-z \leq w \leq z] \geq a),(z > 0)
\vdash_{\mathrm\UPL}\\
&\mathtt{GaussMean}(n) \colon \LIST(M[\REAL \times \REAL \times \REAL])\mid\\
&\forall{i \colon\INT}.1 \leq i \leq n \implies \Pr_{(m,l,u)\sim \mathbf{r}[i]}[l \leq \mu \leq u] \geq a.
\end{align*}
Here, the role of the assertion $(\Pr_{w\sim\mathtt{Gauss}(0,1)}[-z \leq w \leq z] \geq a)$ is referring a table of Z-score of standard Gaussian distribution.
We mainly use the reproductive property of Gaussian distributions
and conversions of Gaussian distributions through the standard Gaussian distribution $\mathtt{Gauss}(0,1)$.
To prove $\Pr_{(m,l,u)\sim \mathbf{r}[i]}[l \leq \mu \leq u] \geq a$, thanks to the equality (\ref{eq:normal:normal_and_standard_normal}),
it suffices to prove
\[
\mathbf{r}[i] = (\MLET m' = \mathtt{Gauss}(\mu, {s}/{i})~ \mathtt{in}~\mathtt{return}(\lrangle{m',m'-z\sqrt{{s}/{i}},m'+z\sqrt{{s}/{i}}})).
\]
To prove this, we apply [u-APP], [u-LETREC], [u-CASE], [u-CONS], and [u-LISTCASE] rules in \NameOfUnaryLogic\ to $\mathtt{GaussMean}$.
We then have the following main premise:
\begin{equation}
\label{example:Mean:assertion:211}
\begin{aligned}
&(i > 0)\land (f(i-1) = r :: \xi) \land \phi_{\mathrm{ind.hyp}} \vdash\\
&\MLET m=(\MLET y=r \IN \RETURN (\pi_1(y))) \IN \MLET x=d \IN\\
&\RETURN (\lrangle{\frac{1}{i}(x+m(i-1)),\frac{1}{i}(x+m(i-1))-z\sqrt{{s/i}},\frac{1}{i}(x+m(i-1))+z\sqrt{{s/i}}})
\mid\\
&\phi'_0 \implies 
\mathbf{r}= \MLET m' = \mathtt{Gauss}(\mu, {s}/{i})~ \mathtt{in}~\mathtt{return}(\lrangle{m',m'-z\sqrt{{s}/{i}},m'+z\sqrt{{s}/{i}}}))
\end{aligned}
\end{equation}
where $\phi_{\mathrm{ind.hyp}}$ is the induction hypothesis obtained by applying [u-LETREC] rule, and $\phi'_0$ is the assumptions on the sample $d$ and the parameter $z$ in the postcondition of the initial judgment.
We first show $(\MLET y=r \IN \RETURN (\pi_1(y))) = \mathtt{Gauss}(\mu, {s}/{(i-1)})$ from the preconditions and axioms on monadic type $M$.
Then we calculate the result $\mathbf{r}$ by applying the equation on Gaussian distributions (\ref{eq:normal:reproductive_property}).
\subsubsection{More Detailed Proof}

Since discussing intervals are easy
\begin{align*}
\lefteqn{\Pr_{(m,l,u)\sim \mathbf{r}[i]}[l \leq \mu \leq u]}\\
&= \Pr_{(m,l,u) \sim  \mathtt{mlet}~m' = \mathtt{Gauss}(\mu,\frac{1}{i}s)~\mathtt{in}~\mathtt{return}(\lrangle{m',m'-z\sqrt{{s/i}},m'+z\sqrt{{s/i}}})}[l \leq \mu \leq u]\\
&= \Pr_{m \sim \mathtt{Gauss}(\mu,\frac{1}{i}s)}[m-z\sqrt{{s/i}} \leq \mu \leq m+z\sqrt{{s/i}}]\\
&= \Pr_{m \sim \mathtt{Gauss}(\mu,\frac{1}{i}s)}[-z\sqrt{{s/i}} \leq \mu - m \leq z\sqrt{{s/i}}]\\
&= \Pr_{m \sim (\mathtt{mlet}~x = \mathtt{Gauss}(0,1)~\mathtt{in}~\mathtt{return}(x\sqrt{{s/i}}+\mu))}[-z\sqrt{{s/i}} \leq \mu - m \leq z\sqrt{{s/i}}]\\
&= \Pr_{x \sim \mathtt{Gauss}(0,1)}[-z\sqrt{{s/i}} \leq -x\sqrt{{s/i}} \leq z\sqrt{{s/i}}]\\
&= \Pr_{x \sim \mathtt{Gauss}(0,1)}[-z \leq x \leq z] \geq a\\
\end{align*}
it suffices to prove the following judgment in \UPL:\
\begin{align*}
\lefteqn{\vdash_{\mathrm\UPL} \mathtt{GaussMean} \colon \mathtt{list}(M[\mathtt{real}\times\mathtt{real}\times\mathtt{real}]) \mid}\\
&\quad (d = \mathtt{Gauss}(\mu,s)) \land (s > 0) \land (z > 0)\\
&\quad\implies
\forall{i \colon\mathtt{nat}}.~i \leq n \implies \mathbf{r}[i] = \mathtt{mlet}~m = \mathtt{Gauss}(\mu,\frac{1}{i}s)~\mathtt{in}~\mathtt{return}(\lrangle{m,m-z\sqrt{{s/i}},m+z\sqrt{{s/i}}}).
\end{align*}
We separate the expression $\Gamma \vdash e \colon \mathtt{list}(M[\mathtt{real}\times\mathtt{real}\times\mathtt{real}])$ as follows:
\begin{align*}
\mathtt{GaussMean} &\equiv (\mathtt{letrec}~f(i) = e_{\mathrm{body}})(n)\\
e_{\mathrm{body}} &\equiv \mathtt{if}~(i \leq 0)~\mathtt{then} e_{\mathrm{body}0}~\mathtt{else}~e_{\mathrm{body}1}\\
e_{\mathrm{body}0} &\equiv []\\
e_{\mathrm{body}1} &\equiv (\mathtt{case}~f(i-1)~\mathtt{with}~[] \Rightarrow e_{\mathrm{body}10}, r :: \xi'\Rightarrow e_{\mathrm{body}11})::f(i-1)\\
e_{\mathrm{body}10} &\equiv
\mathtt{mlet}~x=d~\mathtt{in}~\mathtt{return}(\lrangle{\frac{x}{i},\frac{x}{i}-z\sqrt{{s/i}},\frac{x}{i}+z\sqrt{{s/i}}})\\
e_{\mathrm{body}11} &\equiv
\mathtt{mlet}~m=(\mathtt{mlet}~y=r~\mathtt{in}~\mathtt{return}(\pi_1(y)))~\mathtt{in}~\mathtt{mlet}~x=d~\mathtt{in}\\
&\qquad\mathtt{return}(\lrangle{\frac{1}{i}(x+m(i-1)),\frac{1}{i}(x+m(i-1))-z\sqrt{{s/i}},\frac{1}{i}(x+m(i-1))+z\sqrt{{s/i}}})
\end{align*}
We introduce the following assertions:
\begin{align*}
\phi'   &\equiv\phi'_0 \implies (\phi'_1\land \phi'_2)\\
\phi'_0 &\equiv(d = \mathtt{Gauss}(\mu,s)) \land (i > 0) \land (s > 0) \land (z > 0)\\
\phi'_1 &\equiv(\left|\mathbf{r}\right| = i)\\
\phi'_2 &\equiv(\forall j \colon \INT. 1 \leq j \leq i \implies \mathbf{r}[j] = \mathtt{mlet}~m = \mathtt{Gauss}(\mu, \frac{s}{j})~ \mathtt{in}~\mathtt{return}(\lrangle{m,m-z\sqrt{\frac{s}{j}},m+z\sqrt{\frac{s}{j}}})).
\end{align*}
The goal is to prove $\vdash_{\mathrm\UPL} \mathtt{GaussMean} \colon \mathtt{list}(M[\mathtt{real}\times\mathtt{real}\times\mathtt{real}]) \mid \phi'[n/i]$.
To show this by applying [u-APP] rule, which have the following premise ($\Gamma$ is a context):
\begin{align*}
\Gamma \vdash_{\mathrm\UPL}&(\mathtt{letrec} f i = e_{\mathrm{body}}) \colon \INT \to \mathtt{list}(M[\mathtt{real}\times\mathtt{real}\times\mathtt{real}]) \mid\\
&\forall i \colon \INT.~(\phi'_0 \implies (\phi'_1\land \phi'_2))[\mathbf{r}i/\mathbf{r}].
\end{align*}
To show this by applying [u-LETREC] rule, which has the premise:
\begin{align*}
&\Gamma,f\colon \INT \to \mathtt{list}(M[\mathtt{real}\times\mathtt{real}\times\mathtt{real}]), i \colon \INT\mid\\
&\qquad \forall l \colon \INT.~l < i \implies (\phi'_0 \implies (\phi'_1\land \phi'_2))[l/i,~f(l)/\mathbf{r}]\vdash_{\mathrm\UPL}\\ 
&\qquad \mathtt{if}~(i \leq 0)~\mathtt{then}~e_{\mathrm{body}0}~\mathtt{else}~e_{\mathrm{body}1}  \colon \mathtt{list}(M[\mathtt{real}\times\mathtt{real}\times\mathtt{real}])\mid \\
&\qquad \phi'_0 \implies (\phi'_1\land \phi'_2).
\end{align*}
To show this by applying [u-CASE], which has the premises:
\begin{align}
\label{example:Mean:assertion:20}
&\begin{aligned}
&\Gamma,f\colon \INT \to \mathtt{list}(M[\mathtt{real}\times\mathtt{real}\times\mathtt{real}]), i \colon \INT\mid\\
&\quad (i \leq 0)\land \forall l \colon \INT.~l < i \implies (\phi'_0 \implies (\phi'_1\land \phi'_2))[l/i,~f(l)/\mathbf{r}]\vdash_{\mathrm\UPL}\\ 
&\quad []  \colon \mathtt{list}(M[\mathtt{real}\times\mathtt{real}\times\mathtt{real}])\mid \\
&\quad \phi'_0 \implies (\phi'_1\land \phi'_2)\\
\end{aligned}\\
\notag\\
\label{example:Mean:assertion:21}
&\begin{aligned}
&\Gamma,f\colon \INT \to \mathtt{list}(M[\mathtt{real}\times\mathtt{real}\times\mathtt{real}]), i \colon \INT\mid\\
&\quad (i > 0)\land \forall l \colon \INT.~l < i \implies (\phi'_0 \implies (\phi'_1\land \phi'_2))[l/i,~f(l)/\mathbf{r}]\vdash_{\mathrm\UPL}\\ 
&\quad (\mathtt{case}~f(i-1)~\mathtt{with}~[] \Rightarrow e_{\mathrm{body}10}, r :: \xi'\Rightarrow e_{\mathrm{body}11})::f(i-1)  \colon \mathtt{list}(M[\mathtt{real}\times\mathtt{real}\times\mathtt{real}])\mid \\
&\quad \phi'_0 \implies (\phi'_1\land \phi'_2)
\end{aligned}
\end{align}
The premise (\ref{example:Mean:assertion:20}) is derivable by applying [NIL] rule in \UPL, which has the following premise:
\begin{align*}
&\Gamma,f\colon \INT \to \mathtt{list}(M[\mathtt{real}\times\mathtt{real}\times\mathtt{real}]), i \colon \INT\mid\\
&\quad (i \leq 0)\land \forall l \colon \INT.~l < i \implies (\phi'_0 \implies (\phi'_1\land \phi'_2))[l/i,~f(l)/\mathbf{r}]\vdash_{\mathrm{\PL}}\\ 
&\quad  \top \implies (\phi'_0 \implies (\phi'_1\land \phi'_2))[[]/\mathbf{r}].
\end{align*}
This is derivable in \PL because the assertion $(i \leq 0) \implies \neg \phi'_0$ is obviously a tautology.

To show the premise (\ref{example:Mean:assertion:21}) by applying [u-CONS] rule, which have the following premises:
\begin{align}
\label{example:Mean:assertion:210}
&\begin{aligned}
&\Gamma,f\colon \INT \to \mathtt{list}(M[\mathtt{real}\times\mathtt{real}\times\mathtt{real}]), i \colon \INT\mid\\
&\quad (i > 0)\land \forall l \colon \INT.~l < i \implies (\phi'_0 \implies (\phi'_1\land \phi'_2))[l/i,~f(l)/\mathbf{r}]\vdash_{\mathrm{\UPL}}\\ 
&\quad (\mathtt{case}~f(i-1)~\mathtt{with}~[] \Rightarrow e_{\mathrm{body}10}, r :: \xi'\Rightarrow e_{\mathrm{body}11}) \colon M[\mathtt{real}\times\mathtt{real}\times\mathtt{real}]\mid \\
&\quad  \phi'_0 \implies \mathbf{r} = \mathtt{mlet}~m = \mathtt{Gauss}(\mu, {s/i})~ \mathtt{in}~\mathtt{return}(\lrangle{m,m-z\sqrt{{s/i}},m+z\sqrt{{s/i}}})\\
\end{aligned}\\
\notag\\
\label{example:Mean:assertion:211}
&\begin{aligned}
&\Gamma,f\colon \INT \to \mathtt{list}(M[\mathtt{real}\times\mathtt{real}\times\mathtt{real}]), i \colon \INT\mid\\
&\quad (i > 0)\land \forall l \colon \INT.~l < i \implies (\phi'_0 \implies (\phi'_1\land \phi'_2))[l/i,~f(l)/\mathbf{r}]\vdash_{\mathrm{\UPL}}\\ 
&\quad f(i-1) \colon \mathtt{list}(M[\mathtt{real}\times\mathtt{real}\times\mathtt{real}])\mid \\
&\quad (\phi'_0 \implies (\phi'_1\land \phi'_2))[(i-1)/i]
\end{aligned}
\end{align}
Here, the judgment (\ref{example:Mean:assertion:211}) is easily proved by applying [u-SUB], and [AX], [$\forall_E$], [$\Rightarrow_E$] rules in \PL.
Intuitively we instantiate the assertion $(\phi'_0 \implies (\phi'_1\land \phi'_2))[l/i,~f(l)/\mathbf{r}]$ in the precondition by $l = i-1$, and then apply [u-SUB] rule.

By definition of the length $|-|$ and reference of components $(-)[i]$ of lists
(we need to introduce equations for length of lists {\color{red} $|\xi|+1 = |r::\xi|$} and {\color{red} $(r::\xi)[|r::\xi|] = r$}) and definition of assertions themselves, we have the following assertion in \PL.
\begin{align*}
&\Gamma,f\colon \INT \to \mathtt{list}(M[\mathtt{real}\times\mathtt{real}\times\mathtt{real}]), i \colon \INT\mid\\
&\quad (i > 0)\land \forall l \colon \INT.~l < i \implies (\phi'_0 \implies (\phi'_1\land \phi'_2))[l/i,~f(l)/\mathbf{r}]\vdash_{\mathrm{\PL}}\\
&\qquad\forall{r\colon M[\mathtt{real}\times\mathtt{real}\times\mathtt{real}]}.~
\forall{\xi \colon \mathtt{list}(M[\mathtt{real}\times\mathtt{real}\times\mathtt{real}])}.~\\
&\quad\qquad
(\phi'_0 \implies \mathbf{r} = \mathtt{mlet}~m = \mathtt{Gauss}(\mu, {s/i})~ \mathtt{in}~\mathtt{return}(\lrangle{m,m-z\sqrt{{s/i}},m+z\sqrt{{s/i}}}))[r/\mathbf{r}]\\
&\quad\qquad \implies (\phi'_0 \implies (\phi'_1\land \phi'_2))[(i-1)/i,~\xi/\mathbf{r}] \implies (\phi'_0 \implies (\phi'_1\land \phi'_2))[r::\xi/\mathbf{r}]
\end{align*}

To show the premise (\ref{example:Mean:assertion:210}) by applying [u-LISTCASE] rule, we need to derive
\begin{align}
\label{example:Mean:assertion:2100}
&
\begin{aligned}
&\Gamma,f\colon \INT \to \mathtt{list}(M[\mathtt{real}\times\mathtt{real}\times\mathtt{real}]), i \colon \INT\vdash\\
&\qquad f(i-1)\colon \mathtt{list}M[\mathtt{real}\times\mathtt{real}\times\mathtt{real}]\\
\end{aligned}
\\
\notag\\
\label{example:Mean:assertion:2101}
&
\begin{aligned}
&\Gamma,f\colon \INT \to \mathtt{list}(M[\mathtt{real}\times\mathtt{real}\times\mathtt{real}]), i \colon \INT\mid\\
&\qquad (i > 0)\land  (f(i-1) = [])\land \forall l \colon \INT.~l < i \implies (\phi'_0 \implies (\phi'_1\land \phi'_2))[l/i,~f(l)/\mathbf{r}]\vdash_{\mathrm{\UPL}}\\
&\qquad \mathtt{mlet}~x=d~\mathtt{in}~\mathtt{return}(\lrangle{\frac{x}{i},\frac{x}{i}-z\sqrt{{s/i}},\frac{x}{i}+z\sqrt{{s/i}}}) \colon M[\mathtt{real}\times\mathtt{real}\times\mathtt{real}]\mid \\
&\qquad \phi'_0 \implies \mathbf{r} = \mathtt{mlet}~m = \mathtt{Gauss}(\mu, {s/i})~ \mathtt{in}~\mathtt{return}(\lrangle{m,m-z\sqrt{{s/i}},m+z\sqrt{{s/i}}})\\
\end{aligned}\\
\notag\\
\label{example:Mean:assertion:2102}
&
\begin{aligned}
&\Gamma,f\colon \INT \to \mathtt{list}(M[\mathtt{real}\times\mathtt{real}\times\mathtt{real}]), i \colon \INT,\\
&\qquad
r\colon M[\mathtt{real}\times\mathtt{real}\times\mathtt{real}],~\xi \colon \mathtt{list}(M[\mathtt{real}\times\mathtt{real}\times\mathtt{real}]) 
\mid\\
&\qquad (i > 0)\land (f(i-1) = r::\xi)\land  \forall l \colon \INT.~l < i \implies (\phi'_0 \implies (\phi'_1\land \phi'_2))[l/i,~f(l)/\mathbf{r}]\vdash_{\mathrm{\UPL}}\\
&\qquad \mathtt{mlet}~m=(\mathtt{mlet}~y=r~\mathtt{in}~\mathtt{return}(\pi_1(y)))~\mathtt{in}~\mathtt{mlet}~x=d~\mathtt{in}\\
&\qquad\mathtt{return}(\lrangle{\frac{1}{i}(x+m(i-1)),\frac{1}{i}(x+m(i-1))-z\sqrt{{s/i}},\frac{1}{i}(x+m(i-1))+z\sqrt{{s/i}}}) \colon M[\mathtt{real}\times\mathtt{real}\times\mathtt{real}]\mid \\
&\qquad \phi'_0 \implies \mathbf{r} = \mathtt{mlet}~m = \mathtt{Gauss}(\mu, {s/i})~ \mathtt{in}~\mathtt{return}(\lrangle{m,m-z\sqrt{{s/i}},m+z\sqrt{{s/i}}})
\end{aligned}
\end{align}
The typing judgment (\ref{example:Mean:assertion:2100}) is obvious.
For the premise  (\ref{example:Mean:assertion:2101}), we first need to show $i = 1$ from $i > 0$ and $0 = |[]| = |f(i - 1)| = i-1$.
Technically we show by applying [u-SUB] rule,
\begin{align*}
&\Gamma,f\colon \INT \to \mathtt{list}(M[\mathtt{real}\times\mathtt{real}\times\mathtt{real}]), i \colon \INT\mid\\
&\quad (i > 0)\land  (f(i-1) = [])\land \forall l \colon \INT.~l < i \implies (\phi'_0 \implies (\phi'_1\land \phi'_2))[l/i,~f(l)/\mathbf{r}]\vdash_{\mathrm{\PL}}\\
&\qquad (i = 1)\land (\phi'_0 \implies\mathtt{mlet}~x = d~\mathtt{in}~\mathtt{return}(\lrangle{\frac{x}{i},\frac{x}{i}-z\sqrt{{s/i}},\frac{x}{i}+z\sqrt{{s/i}}})\\
&\qquad\qquad\quad= \mathtt{mlet}~m = \mathtt{Gauss}(\mu, {s/i})~ \mathtt{in}~\mathtt{return}(\lrangle{m,m-z\sqrt{{s/i}},m+z\sqrt{{s/i}}}))
\end{align*}

For the premise (\ref{example:Mean:assertion:2102}), since $|f(i-1)| = i-1 > 0$, we must have $i > 1$ and $f(i-1)[i-1] = r$.
Hence the following assertion is derivable:
\begin{align*}
&\Gamma,f\colon \INT \to \mathtt{list}(M[\mathtt{real}\times\mathtt{real}\times\mathtt{real}]), i \colon \INT,\\
&\quad
r\colon M[\mathtt{real}\times\mathtt{real}\times\mathtt{real}],~\xi \colon \mathtt{list}(M[\mathtt{real}\times\mathtt{real}\times\mathtt{real}]) 
\mid\\
&\qquad (i > 0)\land (f(i-1) = r::\xi)\land  \forall l \colon \INT.~l < i \implies (\phi'_0 \implies (\phi'_1\land \phi'_2))[l/i,~f(l)/\mathbf{r}]\vdash_{\mathrm{\PL}}\\
&\quad\qquad r = \mathtt{mlet}~\hat{m} = \mathtt{Gauss}(\mu, \frac{s}{i-1})~ \mathtt{in}~\mathtt{return}(\lrangle{\hat{m},\hat{m}-z\sqrt{\frac{s}{i-1}},\hat{m}+z\sqrt{\frac{s}{i-1}}}).
\end{align*}
By using this, and monadic laws, {\color{red} laws of projections}, and assumption on $d$, we can do the following reduction in the assertion.
\begin{align*}
&\mathtt{mlet}~m=(\mathtt{mlet}~y=r~\mathtt{in}~~\mathtt{return}(\pi_1(y)))~\mathtt{in}~\mathtt{mlet}~x=d~\mathtt{in}\\
&\qquad\mathtt{return}(\lrangle{\frac{1}{i}(x+m(i-1)),\frac{1}{i}(x+m(i-1))-z\sqrt{{s/i}},\frac{1}{i}(x+m(i-1))+z\sqrt{{s/i}}})\\
&\{\text{Substituting the above } r \text{ and }  d = \mathtt{Gauss}(\mu, s) \text{ and applying monadic and projection laws.}\}\\
&=\mathtt{mlet}~m=\mathtt{Gauss}(\mu, \frac{s}{i-1})~\mathtt{in}~\mathtt{mlet}~x=\mathtt{Gauss}(\mu, s)~\mathtt{in}\\
&\qquad\mathtt{return}(\lrangle{\frac{1}{i}(x+m(i-1)),\frac{1}{i}(x+m(i-1))-z\sqrt{{s/i}},\frac{1}{i}(x+m(i-1))+z\sqrt{{s/i}}})\\
&\{\text{Applying monadic laws.}\}\\
&=\mathtt{mlet}~\hat{m}=(\mathtt{mlet}~m=\mathtt{Gauss}(\mu, \frac{s}{i-1})~\mathtt{in}~\mathtt{mlet}~x=\mathtt{Gauss}(\mu, s)~\mathtt{in}~\mathtt{return}(\frac{1}{i}(x+m(i-1))))\\
&\qquad~\mathtt{in}~(\lrangle{\hat{m},\hat{m}-z\sqrt{{s/i}},\hat{m}+z\sqrt{{s/i}}})\\
&\{\text{Applying the reproducing property of }\mathtt{Gauss}\}\\
&=\mathtt{mlet}~\hat{m}=\mathtt{Gauss}(\mu, \frac{s}{i-1}\cdot\frac{(i-1)^2}{i^2} + s\cdot\frac{1}{i^2})~\mathtt{in}~\mathtt{return}(\lrangle{\hat{m},\hat{m}-z\sqrt{{s/i}},\hat{m}+z\sqrt{{s/i}}})\\
&\{\text{{\color{red} Just calculations.}}\}\\
&=\mathtt{mlet}~\hat{m}=\mathtt{Gauss}(\mu, {s/i})~\mathtt{in}~\mathtt{return}(\lrangle{\hat{m},\hat{m}-z\sqrt{{s/i}},\hat{m}+z\sqrt{{s/i}}})\\
&\{{\color{red} \alpha\text{-conversion}}\}\\
&=\mathtt{mlet}~m=\mathtt{Gauss}(\mu, {s/i})~\mathtt{in}~\mathtt{return}(\lrangle{m,m-z\sqrt{{s/i}},m+z\sqrt{{s/i}}})
\end{align*}
At all, we obtain the following assertion in \PL: 
\begin{align*}
&\Gamma,f\colon \INT \to \mathtt{list}(M[\mathtt{real}\times\mathtt{real}\times\mathtt{real}]), i \colon \INT,\\
&\quad
r\colon M[\mathtt{real}\times\mathtt{real}\times\mathtt{real}],~\xi \colon \mathtt{list}(M[\mathtt{real}\times\mathtt{real}\times\mathtt{real}]) 
\mid\\
&\qquad (i > 0)\land (f(i-1) = r::\xi)\land  \forall l \colon \INT.~l < i \implies (\phi'_0 \implies (\phi'_1\land \phi'_2))[l/i,~f(l)/\mathbf{r}]\vdash_{\mathrm{\PL}}\\
&\qquad\quad\mathtt{mlet}~m=(\mathtt{mlet}~y=r~\mathtt{in}~~\mathtt{return}(\pi_1(y)))~\mathtt{in}~\mathtt{mlet}~x=d~\mathtt{in}\\
&\qquad\qquad\mathtt{return}(\lrangle{\frac{1}{i}(x+m(i-1)),\frac{1}{i}(x+m(i-1))-z\sqrt{{s/i}},\frac{1}{i}(x+m(i-1))+z\sqrt{{s/i}}})\\
&\qquad=
\mathtt{mlet}~m=\mathtt{Gauss}(\mu, {s/i})~\mathtt{in}~\mathtt{return}(\lrangle{m,m-z\sqrt{{s/i}},m+z\sqrt{{s/i}}}).
\end{align*}
Using this, by applying [u-SUB] rule, we complete the proof.

\subsection{Markov Inequality}

From the axioms (\ref{ineq:MHOL:expected_monotonicity}), and (\ref{eq:MHOL:expected:linearity}) in \PL, we have the actual monotonicity of expected values:
\begin{align}
(\forall{x \colon \tau}.~e_2-e_1 \geq 0) &\implies \mathbb{E}_{x \sim e}[e_2 - e_1] \geq 0, \notag\\
\label{ineq:MHOL:expected_monotonicity2}
(\forall{x \colon \tau}.~e_1 \leq e_2) &\implies \mathbb{E}_{x \sim e}[e_1] \leq \mathbb{E}_{x \sim e}[e_2].
\end{align}
The statement of Markov's inequality is:
\begin{align*}
\{
d \colon M[\mathtt{real}],~
a \colon \mathtt{real}
\}\vdash_{\mathrm{\PL}}
a > 0 \implies \Pr_{x \sim d}[ |x| \geq a] \leq \mathbb{E}_{x \sim d}[|x|]/a.
\end{align*}
For any $a > 0$, we have $|x| \geq a \cdot (\mathtt{if}~|x|\geq a~\mathtt{then}~1~\mathtt{else}~0)$.
Hence the monotonicty and linearity of expected values: we calculate in \PL:
\begin{align*}
\mathbb{E}_{x \sim d}[|x|]
& \geq \mathbb{E}_{x \sim d}[a \cdot (\mathtt{if}~|x|\geq a~\mathtt{then}~1~\mathtt{else}~0)]\\
&= a \mathbb{E}_{x \sim d}[\mathtt{if}~|x|\geq a~\mathtt{then}~1~\mathtt{else}~0]\\
&= a \Pr_{x \sim d}[|x|\geq a].
\end{align*}
To show the first inequality, it suffices to show
\[
\{a \colon \mathtt{real},~x \colon \mathtt{real}\}\vdash_{\mathrm{\PL}}
|x| \geq a \cdot (\mathtt{if}~|x|\geq a~\mathtt{then}~1~\mathtt{else}~0)
\]
To prove this we show by analyzing if-else expression inside \PL:
\begin{align*}
\{d \colon M[\mathtt{real}],~a \colon \mathtt{real}\} \mid |x|\geq a \vdash_{\mathrm{MHOL}} |x| \geq (a \cdot 1)\\
\{d \colon M[\mathtt{real}],~a \colon \mathtt{real}\} \mid |x| < a \vdash_{\mathrm{MHOL}} |x| \geq (a \cdot 0).
\end{align*}

\subsection{Chebyshev Inequality}
By applying $a = b^2$, $d = (\mathtt{mlet}~x=d'~\mathtt{in}~\mathtt{return}~(x-\mu)^2)$, (\ref{eq:MHOL:expected:variable_transformation}), and $\alpha$-conversion to Markov's inequality we have 
\begin{align*}
&\{
d \colon M[\mathtt{real}],~
b \colon \mathtt{real},~
\mu \colon \mathtt{real}
\}\vdash_{\mathrm{\PL}}\\
&\quad b^2 > 0 \implies \Pr_{x \sim d}[ |x - \mu| \geq b] \leq \mathbb{E}_{x \sim d}[|x - \mu|^2]/b^2.
\end{align*}
Hence,
\begin{align*}
&\{
d \colon M[\mathtt{real}],~
b \colon \mathtt{real},~
\mu \colon \mathtt{real}
\}\vdash_{\mathrm{\PL}}
\mu = \mathbb{E}_{x \sim d}[x] \land 1 = \mathbb{E}_{x \sim d}[1] \land b^2 > 0\\
&\qquad\implies \Pr_{x \sim d}[ |x - \mu| \geq b] \leq \mathbb{E}_{x \sim d}[|x - \mu|^2]/b^2.
\end{align*}
We also have:
\begin{align*}
&\{
d \colon M[\mathtt{real}],~
b \colon \mathtt{real},~
\mu \colon \mathtt{real}
\}\vdash_{\mathrm{\PL}}
\mu = \mathbb{E}_{x \sim d}[x] \land 1 = \mathbb{E}_{x \sim d}[1] \land b^2 > 0\\
&\qquad \implies \mathbb{E}_{x \sim d}[|x - \mu|^2\geq b] = \mathrm{Var}_{x\sim d}[x].
\end{align*}
Combining the previous two derivations, we conclude the Chebyshev's inequality:
\begin{align*}
&\{
d \colon M[\mathtt{real}],~
b \colon \mathtt{real},~
\mu \colon \mathtt{real}
\}\vdash_{\mathrm{\PL}}
\mu = \mathbb{E}_{x \sim d}[x] \land 1 = \mathbb{E}_{x \sim d}[1] \land b^2 > 0\\
&\qquad \implies \Pr_{x \sim d}[ |x - \mu| \geq b] \leq \mathrm{Var}_{x\sim d}[x]/b^2.
\end{align*}
\subsection{Omitted Calculations in the Example of Importance Sampling}
The expressions $\mathtt{SumLoop2}$ and $\mathtt{Naive}$ are introduced in 
the verification example of importance sampling is defined by
\begin{align*}
\mathtt{SumLoop2}
\equiv
&\LETREC f(i\colon \INT)
= \lambda{g \colon \tau \to \REAL}.\lambda{h \colon \tau \to \REAL}.\lambda{h_2 \colon \tau \to \REAL}.\\
&\IF (i \leq 0) \THEN \RETURN \lrangle{0,0} \ELSE \MLET x = d \IN \MLET m = f(i-1)(g)(h)(h_2) \IN\\
&\RETURN \lrangle{\frac{1}{i}(\pi_1(m) + (i-1)\ast h(x) \ast g(x)),\frac{1}{i}(\pi_2(m) + (i-1) \ast h_2(x) \ast g(x))}\\
\mathtt{Naive}\equiv&\LETREC f(i\colon \INT)
= \lambda{g \colon \tau \to \REAL}.\lambda{h \colon \tau \to \REAL}.\\
&\IF (i \leq 0) \THEN (\RETURN 0) \ELSE \MLET x = d \IN \MLET m = f(i-1)(g)(h) \IN\\
&\RETURN \frac{1}{i}(\pi_1(m) + (i-1)\ast h(x) \ast g(x))
\end{align*}
We have the following structural equalities:
\begin{align*}
&\vdash_{\mathrm{\NameOfRelationalLogic}}
\mathtt{SumLoop2} \sim \mathtt{Naive} \mid
\MLET z = \mathbf{r}_1\!(k)\!(g)\!(h)\!(h_2) \IN \RETURN (\pi_1(z)) = \mathbf{r}_2\!(k)\!(g)\!(h)\!\\
&\vdash_{\mathrm{\NameOfRelationalLogic}}
\mathtt{SumLoop2} \sim \mathtt{Naive} \mid
\MLET z = \mathbf{r}_1(k)\!(g)\!(h)\!(h_2) \IN \RETURN (\pi_2(z)) = \mathbf{r}_2\!(k)\!(g)\!(h_2)\!\\
&C > 0\vdash_{\mathrm{\NameOfRelationalLogic}} \mathtt{SumLoop2} \sim \mathtt{SumLoop2} \mid
\MLET z = \mathbf{r}_1\!(k)\!(g)\!(h)\!(1) \IN \RETURN (\pi_1(z)/\pi_2(z)) \notag\\
&\qquad\qquad\qquad\qquad\qquad = \MLET z = \mathbf{r}_1(k)(g/C)(h)(1) \IN \RETURN (\pi_1(z)/\pi_2(z))\\
&\vdash_{\mathrm{\NameOfRelationalLogic}} \mathtt{SumLoop2} \sim \mathtt{SumLoop}\mid \mathbf{r}_1(k)(g)(h)(1) = \mathbf{r}_2(k)(g)(h)
\end{align*}
We will see the most complicated calculation in the verification example of importance samplings.
We set $h_2 = \lambda x \colon \tau.(\IF g(x) \leq k \ast \exp(-t/2) \THEN 1 \ELSE 0) \ast h(x)$.
We first compute:
\begin{align*}
\lefteqn{\mathbb{E}_{w \sim (\mathtt{mlet}~z = \mathtt{SumLoop2}
(k)(g)(h)(h_2)~\mathtt{in}~\mathtt{return}(\pi_1(z)))}[|w-\mu|]}\\
&\qquad \{\text{Variable transformation in expectation values}\}\\
&=\mathbb{E}_{z \sim \mathtt{SumLoop2}(k,g,h,h_2)}[|\pi_1(z)-\mu|]\\
&\qquad \{\text{Triangle inequality on absolute values and monotonicity of expectations}\}\\
&\leq\mathbb{E}_{z \sim \mathtt{SumLoop2}(k)(g)(h)(h_2)}[|\pi_1(z)-\pi_2(z)|+|\pi_2(z)-\mu'|+|\mu-\mu'|]\\
&\qquad \{\text{additivity pf expectations}\}\\
&=\mathbb{E}_{z \sim \mathtt{SumLoop2}(k)(g)(h)(h_2)}[|\pi_1(z)-\pi_2(z)|]+\mathbb{E}_{z \sim \mathtt{SumLoop2}(k)(g)(h)(h_2)}[|\pi_2(z)-\mu'|]\\
&+\mathbb{E}_{z \sim \mathtt{SumLoop2}(k)(g)(h)(h_2)}[|\mu-\mu'|]
\end{align*}
Next, we apply Cauchy-Schwartz inequality to each expected values.
We denote $a \equiv k \ast \exp(-t/2) \geq \exp(L + t/2)$, and $\mu' \equiv \mathbb{E}_{y \sim d'}[h_2(y)]$.
Then we compute:
\begin{align*}
\lefteqn{\mathbb{E}_{z \sim \mathtt{SumLoop2}(k)(g)(h)(h_2)}[|\pi_1(z)-\pi_2(z)|] = \mathbb{E}_{y \sim d'}[|h(y)-h_2(y)|] }\\
&\qquad \{\text{Applying $d' = \mathrm{scale}(d,g)$ and $h_2(x) =\mathtt{if}~g(x)\leq a~\mathtt{then}~h(x)~\mathtt{else}~0$}\}\\
&= \mathbb{E}_{x \sim d}[g(x) \ast |\mathtt{if}~g(x) \leq a ~\mathtt{then}~0~\mathtt{else}~h(x)|]\\
&= \mathbb{E}_{x \sim d}[g(x) \ast (\mathtt{if}~g(x) \leq a ~\mathtt{then}~0~\mathtt{else}~1) \ast |h(x)|]\\
&\qquad \{\text{Applying $d' = \mathrm{scale}(d,g)$ again}\}\\
&= \mathbb{E}_{y \sim d'}[(\mathtt{if}~g(y) \leq a ~\mathtt{then}~0~\mathtt{else}~1) \ast |h(y)|]\\
&\qquad \{\text{Applying Cauchy-Schwartz inequality}\}\\
&\leq \mathtt{sqrt}(\mathbb{E}_{y \sim d'}[(\mathtt{if}~g(y) \leq a ~\mathtt{then}~0~\mathtt{else}~1)^2] \ast \mathbb{E}_{y \sim d'}[h(y)^2])\\
&= \mathtt{sqrt}({\Pr_{y \sim d'}[g(y) > a]}) \ast \mathtt{sqrt}({(\mu^2 + \sigma^2)}) \leq \mathtt{sqrt}({\Pr_{y \sim d'}[\log(g(y)) > L + t/2]}) \ast \mathtt{sqrt}({\mu^2 + \sigma^2})\\
\lefteqn{\mathbb{E}_{z \sim \mathtt{SumLoop2}(k)(g)(h)(h_2)}[|\mu-\mu'|]}\\
&\qquad \{\text{Applying $\mathbb{E}_{z \sim \mathtt{SumLoop2}(k,g,h,h_2)}[1]=1$}\}\\
&= |\mu-\mu'|
= |\mathbb{E}_{y \sim d'}[h(y)]-\mathbb{E}_{y \sim d'}[h_2(y)]|
= |\mathbb{E}_{y \sim d'}[h(y)-h_2(y)]|\\
&\qquad \{\text{Reusing the above calculation}\}\\
&\leq \mathtt{sqrt}({\Pr_{y \sim d'}[\log(g(y)) > L + t/2]}) \ast \mathtt{sqrt}({\mu^2 + \sigma^2})\\
\lefteqn{\mathbb{E}_{z \sim \mathtt{SumLoop2}(k)(g)(h)(h_2)}[|\pi_2(z)-\mu'|]}\\
&\qquad \{\text{Applying Cauchy-Schwartz inequality}\}\\
&\leq \mathtt{sqrt}({\mathbb{E}_{z \sim \mathtt{SumLoop2}(k)(g)(h)(h_2)}[|\pi_2(z)-\mu'|^2]}) = \mathtt{sqrt}({\mathrm{Var}_{z \sim \mathbf{r}(k)(g)(h)(h_2)}[\pi_2(z)]})\\
&=\mathtt{sqrt}({\frac{\mathrm{Var}_{x \sim d}[g(x) \ast h_2(x)]}{k}})\\
&\qquad \{\text{Applying definition of variance, and definition of $h_2$}\}\\
&\leq \mathtt{sqrt}({\frac{\mathbb{E}_{x \sim d}[g(x)^2 \ast h_2(x)^2]}{k}})
\leq \mathtt{sqrt}({\frac{a \ast \mathbb{E}_{x \sim d}[g(x) \ast h(x)^2]}{k}})
= \mathtt{sqrt}(\frac{a \ast \mathbb{E}_{y \sim d'}[h(y)^2]}{k})\\
&= \mathtt{sqrt}({\mu^2 + \sigma^2}) \ast \mathtt{sqrt}({\frac{a}{k}})= \mathtt{sqrt}({\mu^2 + \sigma^2}) \ast \exp(-t/2).
\end{align*}

\subsection{Derivations of Several (in) Equalities.}
Several derivations of equalities are bit complicated, so we show some of them.
\subsubsection{Marginal Law of Product Measures}
Let $\Gamma\vdash e_1 \colon M[\tau_1]$ and $\Gamma\vdash e_2 \colon M[\tau_2]$.
Then the following equalities are derivable in \PL:
\begin{align*}
\lefteqn{
\MLET w = e_1 \otimes e_2 \IN \RETURN\pi_1(w)
}
\\
&= \BIND (\BIND e_1 ~\lambda{w_1}.(\BIND e_2~\lambda{w_2}.\RETURN\lrangle{w_1,w_2}))~\lambda{w}.\RETURN\pi_1(w)
&(\text{Syntactic sugar})
\\
&= \BIND e_1 ~\lambda{w_1}.(\BIND (\BIND e_2~\lambda{w_2}.\RETURN\lrangle{w_1,w_2})~\lambda{w}.\RETURN\pi_1(w))
&\text{(associativity of $\BIND$)}
\\
&= \BIND e_1 ~\lambda{w_1}.
	(\BIND e_2~\lambda{w_2}.(\BIND \RETURN\lrangle{w_1,w_2}~\lambda{w}.\RETURN\pi_1(w)))
&\text{(associativity of $\BIND$)}
\\
&= \BIND e_1 ~\lambda{w_1}.(\BIND e_2~\lambda{w_2}.\RETURN (w_1))
&\text{Monadic law (unit law)}
\\
&= (\mathtt{scale}(e_1,\lambda{w_1}.\mathbb{E}_{y\sim \BIND e_2~\lambda{w_2}.\RETURN (w_1))}[1]))
&\text{(equation \ref{eq:MHOL:scaling4})}
\\
&= (\mathtt{scale}(e_1,\lambda{w_1}.\mathbb{E}_{w_2 \sim e_2}[1]))
&\text{(equation \ref{eq:MHOL:expected:variable_transformation})}
\end{align*}
Similarly we have 
\[
\vdash_{\mathrm\PL}(\MLET w = e_1 \otimes e_2 \IN \RETURN\pi_2(w))
= (\mathtt{scale}(e_2,\lambda{w_2}.\mathbb{E}_{w_1 \sim e_1}[1])).
\]
\subsubsection{Independence for Product Measures}
Let $\Gamma\vdash d_1 \colon M[\tau_1]$, $\Gamma\vdash d_2 \colon M[\tau_2]$, 
$\Gamma\vdash f \colon \tau_1 \to \mathtt{pReal}$ and 
$\Gamma\vdash g \colon \tau_2 \to \mathtt{pReal}$.
Then the following equalities are derivable in \PL:
\begin{align*}
\lefteqn{\mathbb{E}_{w\sim d_1 \otimes d_2}[f(\pi_1(w)) \ast g(\pi_1(w))]}\\
&=\mathbb{E}_{w\sim \SCALE(d_1 \otimes d_2, \lambda w.f(\pi_1(w))}[g(\pi_2(w))]
&\text{(equation \ref{eq:MHOL:expected:scaling})}\\
&= \mathbb{E}_{w\sim \SCALE(d_1,f) \otimes \SCALE(d_2,1)}[g(\pi_2(w))]
&\text{(equation \ref{eq:MHOL:scaling3})}\\
&= \mathbb{E}_{w\sim \SCALE(d_1,f) \otimes d_2}[g(\pi_2(w))]
&\text{(equation \ref{eq:MHOL:scaling1})}\\
&= \mathbb{E}_{y \sim \BIND (\SCALE(d_1,f) \otimes d_2)~\lambda{w}.\RETURN\pi_2(w)}[g(y)]
&\text{(equation \ref{eq:MHOL:expected:variable_transformation})}\\
&= \mathbb{E}_{y \sim \SCALE(d_2,\lambda{\_}.\mathbb{E}_{x\sim\SCALE(d_1,f)}[1]))}[g(y)]
&\text{(Marginal Law)}
\\
&= \mathbb{E}_{y \sim d_2}[\mathbb{E}_{x\sim\SCALE(d_1,f)}[1] \ast g(y)]
&\text{(equation \ref{eq:MHOL:expected:scaling})}\\
&= \mathbb{E}_{x\sim\SCALE(d_1,f)}[1] \ast \mathbb{E}_{y \sim d_2}[g(y)]
&\text{(equation \ref{eq:MHOL:expected:linearity})}\\
&= \mathbb{E}_{x\sim d_1}[f(x)] \ast \mathbb{E}_{y \sim d_2}[g(y)]
&\text{(equation \ref{eq:MHOL:expected:scaling})}
\end{align*}
\subsubsection{Slicing Law on Simple Observations.}
Let $\Gamma\vdash x \colon M[\tau_1]$, $\Gamma\vdash y \colon M[\tau_2]$, 
$\Gamma\vdash f \colon \tau_2 \to \mathtt{pReal}$, and assume $\mathbb{E}_{\_\sim x}[1] = 1$.
Then the following equalities are derivable in \PL:
\begin{align*}
\lefteqn{
\MLET v = (\mathtt{observe}~x \otimes y \AS \lambda{w}.f(\pi_2(w)) \IN \RETURN (\pi_1(v))
}
\\
&=
\MLET v = 
\mathtt{normalize}(\SCALE(x \otimes y, \lambda{w}.f(\pi_2(w)))
\IN \RETURN (\pi_1(v))
&\text{(equation \ref{eq:MHOL:observe})}\\
&=
\MLET v = 
\SCALE(\SCALE(x \otimes y,\lambda{w}.f(\pi_2(w))),\lambda{\_}.1/K)
\IN \RETURN (\pi_1(v))
&\text{(equation \ref{eq:MHOL:normalize1})}
\\
&=
\MLET v = 
\SCALE(x \otimes y, \lambda{w}.f(\pi_2(w))/K)
\IN \RETURN (\pi_1(v))
&\text{(equation \ref{eq:MHOL:scaling1})}
\\
&=
\MLET v = (
(\SCALE(x, \lambda{\_}.1)
\otimes
(\SCALE(y, f/K)
)
\IN \RETURN (\pi_1(v))
&\text{(equation \ref{eq:MHOL:scaling3})}
\\
&=\MLET v = ( x \otimes (\SCALE(y,f/K)))
) \IN \RETURN (\pi_1(v))
&\text{(equation \ref{eq:MHOL:scaling1})}
\\
&=\SCALE( x, \lambda{\_}.\mathbb{E}_{\_ \sim \SCALE(y,f/K)}[1])
&\text{(Marginal Law)}
\\
&=\SCALE(x,\lambda{\_}.\mathbb{E}_{\OBSERVE y \AS f}[1])
&\text{(\ddag)}
\\
&=\SCALE(x,\lambda{\_}.1) = x
&\text{(equation \ref{eq:MHOL:scaling1})}
\end{align*}
Where $K \equiv \mathbb{E}_{\_\sim\SCALE(x \otimes y,\lambda{w}.f(\pi_2(w)))}[1]$.
The equality (\ddag) is derived as follows.
By the independence of product measure, we have $K = \mathbb{E}_{\_ \sim \SCALE(y,f)}[1] \ast \mathbb{E}_{\_\sim x}[1]$.
Thanks to $\mathbb{E}_{\_\sim x}[1] = 1$ we have $K = \mathbb{E}_{\_ \sim \SCALE(y,f)}[1]$.
Hence, we conclude $\SCALE(y,f/K) = \OBSERVE y \AS f$.
\subsection{Gaussian are Conjugate Prior wrt Gaussian Likelihood functions}
We assumed the definition 
\[
\mathtt{Gauss}(x,\sigma^2)=\SCALE(\mathtt{Lebesgue},\mathtt{GPDF}(x,\sigma^2))
\]
From the probability of Gaussian distribution and applying (equations \ref{eq:MHOL:scaling1}, \ref{eq:MHOL:normalize1}, and \ref{eq:MHOL:normalize3}), we have
\begin{equation}
\label{eq:Gaussan:normalization}
\mathtt{Gauss}(x,\sigma^2)=\mathtt{normalize}(\SCALE(\mathtt{Lebesgue},\lambda r.\exp(\frac{(r-x)^2}{2\sigma^2})).
\end{equation}
Using this we calculate,
\begin{align*}
\lefteqn{\OBSERVE \mathtt{Gauss}(\delta,\xi^2) \AS \mathtt{GPDF}(z,\sigma^2)}\\
&=\mathtt{normalize}(\SCALE(\mathtt{Gauss}(\delta,\xi^2),\mathtt{GPDF}(z,\sigma^2)))
&\text{(equation \ref{eq:MHOL:observe})}
\\
&=\mathtt{normalize}(\SCALE(\SCALE(\mathtt{Lebesgue},\mathtt{GPDF}(\delta,\xi^2)),\mathtt{GPDF}(z,\sigma^2))
&\text{(Axiom on $\mathtt{Gauss}$)}
\\
&=\mathtt{normalize}(\SCALE(\mathtt{Lebesgue},\lambda{r}.\mathtt{GPDF}(\delta,\xi^2)(r) \ast \mathtt{GPDF}(z,\sigma^2)(r)))
&\text{(equation \ref{eq:MHOL:scaling3})}
\\
&=\mathtt{normalize}(\SCALE(\mathtt{Lebesgue},
\lambda{r}.\exp(\frac{(r-\delta)^2}{2\xi^2}) \ast \exp(\frac{(r-z)^2}{2\sigma^2})))
&\text{(equations \ref{eq:MHOL:normalize1} and \ref{eq:MHOL:normalize3})}
\\
&=
\mathtt{normalize}(\SCALE(\mathtt{Lebesgue},
\lambda{r}.\exp(\frac{(r-\frac{z \xi^2 +\delta\sigma^2}{\xi^2 + \sigma^2})^2}{\frac{2\xi^2\sigma^2}{\xi^2 + \sigma^2}}))
&\text{(calculation)}
\\
&=\mathtt{Gauss}(\frac{z \xi^2 +\delta\sigma^2}{\xi^2 + \sigma^2},\frac{\xi^2\sigma^2}{\xi^2 + \sigma^2})
&\text{(equation \ref{eq:Gaussan:normalization})}
\end{align*}
\section{Proofs and Sketches on Graded $\top\top$-liftings}

\begin{theorem}[Graded Monadic Laws of $\mathfrak{U}_S$]
The following rules are derivable:
\begin{gather*}
\Gamma \mid \Psi \vdash_{\mathrm{\NameOfUnderlyingLogic}}
\forall{\alpha\colon\zeta}. 
\forall{\beta\colon\zeta}. 
\mathfrak{U}^{\alpha}_S \phi \implies \mathfrak{U}^{\beta}_S \phi
\\
\begin{aligned}
&\Gamma \mid \Psi \vdash_{\mathrm{\NameOfUnderlyingLogic}}
\forall{\alpha \colon \zeta}. 
(\forall{x \colon \tau}. \phi_1[x/\mathbf{r}'] \implies \phi_2[x/\mathbf{r}'])\implies (\mathfrak{U}^{\alpha}_S \phi_1 \implies \mathfrak{U}^{\alpha}_S \phi_2)
\end{aligned}
\\
\AxiomC{$\Gamma \mid \Psi \vdash_{\mathrm{\NameOfUnaryLogic}} e \colon \tau \mid \phi[\mathbf{r}/\mathbf{r}']$}
\UnaryInfC{$\Gamma \mid \Psi \vdash_{\mathrm{\NameOfUnaryLogic}} \RETURN(e) \colon D[\tau] \mid \mathfrak{U}^{1_\zeta}_S \phi$}
\DisplayProof
\\
\AxiomC{
$\Gamma \mid \Psi \vdash_{\mathrm{\NameOfUnaryLogic}} e \colon D[\tau]\mid \mathfrak{U}^{\alpha}_S \phi$
\;
$\Gamma \mid \Psi \vdash_{\mathrm{\NameOfUnaryLogic}} e' \colon \tau \to D[\tau'] \mid \forall{x \colon \tau}. \phi[x/\mathbf{r}]{\implies} (\mathfrak{U}^{\beta}_S\phi')[\mathbf{r} x/\mathbf{r}]
$
}
\UnaryInfC{$\Gamma \mid \Psi \vdash_{\mathrm{\NameOfUnaryLogic}}\BIND e~e' \colon D[\tau'] \mid \mathfrak{U}^{\alpha\cdot\beta}_S\phi'$}
\DisplayProof
\end{gather*}
\end{theorem}
\begin{proof}[Proof Sketch]
The proofs are straightforward.
For example, the proof of (\ref{rule_general_graded_monad_Kleisli_lifting_unary})
begins with applying the [u-BIND] rule, which has the following premise:
\[
\vdash_{\mathrm{\NameOfUnaryLogic}} e' \colon \tau \to D[\tau'] \mid \forall{d \colon D[\tau]}. (\mathfrak{U}^{\alpha}_S\phi) [d/\mathbf{r}]\implies ( \mathfrak{U}^{\alpha\cdot\beta}_S\phi')[\BIND d~\mathbf{r} /\mathbf{r}].
\]
We then apply [u-SUB] rule, which has the following \NameOfUnderlyingLogic-premise:
\begin{align*}
\vdash_{\mathrm{\NameOfUnderlyingLogic}}&
(\forall{x \colon \tau}. \phi[x/\mathbf{r}]\implies (\mathfrak{U}^{\beta}_S\phi')[e' x/\mathbf{r}])\\
&\implies
(\forall{d \colon D[\tau]}. (\mathfrak{U}^{\alpha}_S\phi) [d/\mathbf{r}]\implies ( \mathfrak{U}^{\alpha\cdot\beta}_S\phi')[\BIND d~e'/\mathbf{r}]).
\end{align*}
To prove this premise,
consider $f \colon \tau' \to D[\theta]$ satisfying $\phi'[y/\mathbf{r}] \implies S[\gamma/\mathbf{k},f(y)/\mathbf{l}]$ and $d \colon D[\tau]$ such that $(\mathfrak{U}^{\alpha}_S\phi) [d/\mathbf{r}]$.
First, from the assumption on $e'$ and $f$, we obtain 
$\forall{x \colon \tau}. \phi[x/\mathbf{r}]\implies S[\beta \cdot \gamma/\mathbf{k},\BIND e'x~f /\mathbf{l}]$.
Hence, by the assumption on $d$, we obtain $S[\alpha \cdot \beta \cdot \gamma/\mathbf{k},\BIND d~\lambda{x\colon\tau}.(\BIND e'x~f) /\mathbf{l}]$.
By the associativity of $\BIND$, this is equivalent to $S[\alpha \cdot \beta \cdot \gamma/\mathbf{k},\BIND (\BIND d~e')~f /\mathbf{l}]$.
Since $f$ is arbitrary, we conclude $(\mathfrak{U}^{\alpha\cdot\beta}_S\phi')[\BIND d~e'/\mathbf{r}]$.
\end{proof}

\begin{proposition}
In the setting in Section \ref{subsec:TT-lifting:unionbound},
the following reduction is derivable in \NameOfUnderlyingLogic.
\[
\AxiomC{
 $\Gamma, \mathbf{r'}\colon \tau \vdash e \colon \BOOL $
\quad $\Gamma, \mathbf{r'}\colon \tau \mid \Psi \vdash_{\mathrm{\NameOfUnderlyingLogic}} \neg \phi \iff (e = \mathtt{true})$
}
\UnaryInfC{
$\Gamma, \mathbf{r} \colon D[\tau] \mid \Psi \vdash_{\mathrm{\NameOfUnderlyingLogic}}
\mathfrak{U}_S^\alpha(\neg \phi) \iff \Pr_{X \sim \mathbf{r}}[e[X/\mathbf{r}']] \leq \alpha$}
\DisplayProof
\]
\end{proposition}
\begin{proof}[Proof Sketch]
We observe the following:
\[
\mathfrak{U}_S^\alpha(\neg \phi) \defeq
\left\{
\begin{aligned}
&\forall{f \colon \tau \to D[\mathtt{unit}]}.\forall{\beta \colon \mathtt{pReal}}.\\
&(\forall{x \colon \tau}. \neg \phi[x/\mathbf{r}'] \implies \mathbb{E}_{y \sim f(x)}[1] \leq \beta) \implies (\mathbb{E}_{y \sim (\BIND \mathbf{r}~f)}[1] \leq \alpha+\beta)
\end{aligned}
\right.
\]
The forward direction of the conclusion is proved by the equality (\dag) derived
from the axioms on scaling in \PL:
\[
\Pr_{X \sim \mathbf{r}}[e[X/\mathbf{r}']]
\labeleq{(\dag)}
\mathbb{E}_{y \sim \BIND \mathbf{r}~~\lambda{X}.\SCALE( \RETURN(\ast),\lambda{u \colon \mathtt{unit}}.\IF e[X/\mathbf{r}'] \THEN 1 \ELSE 0)}[1]
\leq \alpha + 0.
\]
For the converse direction, we need the equivalence between $D[\mathtt{unit}]$
and the unit interval $[0,1]$. To realize this we apply the axioms
(\ref{eq:MHOL:equivalence:value_and_distribution_on_unit_type}) in
\NameOfUnderlyingLogic.
Combining the axioms
(\ref{eq:MHOL:equivalence:value_and_distribution_on_unit_type}) and other axioms
on scaling of measures, we conclude the equivalence between a distribution $e
\colon D[\mathtt{unit}]$ and its mass $\mathbb{E}_{y \sim e}[1]$. The proof
follows by showing that the function $\lambda x \colon \tau.(\IF
e[x/\mathbf{r'}] \THEN \beta \ELSE 0)$ corresponds to the greatest function $f
\colon \tau \to D[\mathtt{unit}]$ such that $(\forall{x \colon \tau}. \neg
\phi[x/\mathbf{r}'] \implies \mathbb{E}_{y \sim f(x)}[1] \leq \beta)$. 
\end{proof}

For relational $\top\top$-lifting, we have
the following derivable graded monadic laws:
\begin{theorem}[Graded Monadic Laws of $\mathfrak{R}_S$]\label{thm:TT-lifting:monadic_law}
The following rules are derivable:
\begin{gather}
\label{validity_general_graded_monad_inclusion_relational}
\Gamma \mid \Psi \vdash_{\mathrm{\NameOfUnderlyingLogic}}
\forall{\alpha\colon\zeta}. 
\forall{\beta\colon\zeta}. 
(\alpha \leq_{\zeta} \beta \implies
\mathfrak{R}^{\alpha}_S \phi \implies \mathfrak{R}^{\beta}_S \phi)
\\
\label{validity_general_graded_monad_monotonicity_relational}
\begin{array}{r@{}l}
&\Gamma \mid \Psi \vdash_{\mathrm{\NameOfUnderlyingLogic}}
\forall{\alpha \colon \zeta}. (\forall{x_1 \colon \tau_1}. \forall{x_2 \colon \tau_2}. \phi_1[x_1/\mathbf{r}'_1, x_2/\mathbf{r}'_2] \implies \phi_2[x_1/\mathbf{r}'_1, x_2/\mathbf{r}'_2])\\
&\qquad\qquad\qquad \implies (\mathfrak{R}^{\alpha}_S \phi_1 \implies \mathfrak{R}^{\alpha}_S \phi_2)
\end{array}
\\
\label{rule_general_graded_monad_unit_relational}
\AxiomC{$\Gamma \mid \Psi \vdash_{\mathrm{\NameOfRelationalLogic}} e_1 \colon \tau_1 \sim e_2 \colon \tau_2 \mid \phi[\mathbf{r}_1/\mathbf{r}'_1, \mathbf{r}_2/\mathbf{r}'_2]$}
\UnaryInfC{$\Gamma \mid \Psi \vdash_{\mathrm{\NameOfRelationalLogic}} \RETURN (e_1) \colon D[\tau_1] \sim \RETURN (e_2) \colon D[\tau_2] \mid \mathfrak{R}^{1_\zeta}_S \phi$}
\DisplayProof
\\
\label{rule_general_graded_monad_Kleisli_lifting_relational}
\AxiomC{
$
\begin{array}{rl}
&\Gamma \mid \Psi \vdash_{\mathrm{\NameOfRelationalLogic}} e_1 \colon D[\tau_1]\sim e_2 \colon D[\tau_2] \mid \mathfrak{R}^{\alpha}_S \phi\\
&
\begin{array}{rl}
\lefteqn{\Gamma \mid \Psi \vdash_{\mathrm{\NameOfRelationalLogic}} 
e'_1 \colon \tau_1 \to D[\tau'_1] \sim e'_2 \colon \tau_2 \to D[\tau'_2] \mid}\\
&\qquad \forall{x_1 \colon \tau_1}. \forall{x_2 \colon \tau_2}. \phi[x_1/\mathbf{r}_1, x_2/\mathbf{r}_2]\implies (\mathfrak{R}^{\beta}_S\phi')[\mathbf{r}_1 x_1/\mathbf{r}_1, \mathbf{r}_2 x_2/\mathbf{r}_2]
\end{array}
\end{array}
$
}
\UnaryInfC{$\Gamma \mid \Psi \vdash_{\mathrm{\NameOfRelationalLogic}}\BIND e_1~e'_1 \colon D[\tau'_1] \sim \BIND e_2~e'_2 \colon D[\tau'_2] \mid \mathfrak{R}^{\alpha\cdot\beta}_S\phi'$}
\DisplayProof
\end{gather}
\end{theorem}
\begin{proof}
The rule (\ref{validity_general_graded_monad_inclusion_relational}) is obvious from the monotonicity of lifting parameter.

The rule (\ref{validity_general_graded_monad_monotonicity_relational}) proved from the fact that formulas in the following form is tautology:
\[
(\phi_A \implies \phi_B) \implies ((\phi_B \implies \phi_C) \implies (\phi_A \implies \phi_C)).
\]

We prove (\ref{rule_general_graded_monad_unit_relational}).
To prove this by applying two-sided [r-RETURN] rule, having the premise
\[
\Gamma \mid \Psi \vdash_{\mathrm{\RPL}} e_1 \colon \tau_1 \sim e_2 \colon \tau_2 \mid (\mathfrak{R}^{1_\zeta}_S \phi)[\mathtt{return}(\mathbf{r}_1)/\mathbf{r}_1,~\mathtt{return}(\mathbf{r}_2)/\mathbf{r}_2].
\]
To prove this by relational [r-SUB] rule, having the premise:
\begin{equation}\label{eq:monaic:TT-lifting:unit:base}
\Gamma \mid \Psi \vdash_{\mathrm{\PL}} \phi[\mathbf{r}_1/\mathbf{r}'_1,~\mathbf{r}_2/\mathbf{r}'_2] \implies (\mathfrak{R}^{1_\zeta}_S \phi)[\mathtt{return}(\mathbf{r}_1)/\mathbf{r}_1,~\mathtt{return}(\mathbf{r}_2)/\mathbf{r}_2].
\end{equation}
By the the monadic unit law, we obtain the following equality in \PL-formulas:
\begin{align*}
\lefteqn{(\mathfrak{R}^{1_\zeta}_S \phi)[\mathtt{return}(\mathbf{r}_1)/\mathbf{r}_1,~\mathtt{return}(\mathbf{r}_2)/\mathbf{r}_2]}\\
&=
\lefteqn{
\forall{\beta \colon \zeta}.~\forall{f_1 \colon \tau_1 \to D[\theta_1]}.~\forall{f_2 \colon \tau_2 \to D[\theta_2]}.~}\\
&\qquad (\forall{x_1 \colon \tau_1}.~\forall{x_2 \colon \tau_2}.~\phi[x_1/\mathbf{r}'_1,~x_2/\mathbf{r}'_2] \implies S[\beta/\mathbf{k},~f_1(x_1)/\mathbf{l}_1,~f_2(x_2)/\mathbf{l}_2]) \\
& \qquad \qquad \implies S[\beta/\mathbf{k},~\BIND \RETURN(\mathbf{r}_1)~f_1) /\mathbf{l}_1,~\BIND \RETURN(\mathbf{r}_2)~f_2 /\mathbf{l}_2]))\\
&=
\lefteqn{
\forall{\beta \colon \zeta}.~\forall{f_1 \colon \tau_1 \to D[\theta_1]}.~\forall{f_2 \colon \tau_2 \to D[\theta_2]}.~}\\
&\qquad (\forall{x_1 \colon \tau_1}.~\forall{x_2 \colon \tau_2}.~\phi[x_1/\mathbf{r}'_1,~x_2/\mathbf{r}'_2] \implies S[\beta/\mathbf{k},~f_1(x_1)/\mathbf{l}_1,~f_2(x_2)/\mathbf{l}_2]) \\
& \qquad \qquad \implies S[\beta/\mathbf{k},~f_1(\mathbf{r}_1) /\mathbf{l}_1,~f_2(\mathbf{r}_2) /\mathbf{l}_2]))
\end{align*}
The second formula
\begin{align*}
&~\phi[\mathbf{r}_1/\mathbf{r}'_1,~\mathbf{r}_2/\mathbf{r}'_2]\\
&\quad \implies
\forall{\beta \colon \zeta}.~\forall{f_1 \colon \tau_1 \to D[\theta_1]}.~\forall{f_2 \colon \tau_2 \to D[\theta_2]}.\\
&\qquad\quad (\forall{x_1 \colon \tau_1}.~\forall{x_2 \colon \tau_2}.~\phi[x_1/\mathbf{r}'_1,~x_2/\mathbf{r}'_2] \implies S[\beta/\mathbf{k},~f_1(x_1)/\mathbf{l}_1,~f_2(x_2)/\mathbf{l}_2]) \\
& \qquad \qquad \implies S[\beta/\mathbf{k},~f_1(\mathbf{r}_1) /\mathbf{l}_1,~f_2(\mathbf{r}_2) /\mathbf{l}_2])).
\end{align*}
is a tautology. Hence, we conclude (\ref{eq:monaic:TT-lifting:unit:base}).

Next we show (\ref{rule_general_graded_monad_Kleisli_lifting_relational}).
To prove this by applying [r-BIND] and [r-SUB] rules, having the following \PL- premise
\begin{align*}
\Gamma\mid\Psi \vdash_{\mathrm{\PL}}
&\forall{\mathbf{r}_1 \colon \tau_1 \to D[\tau'_1]}.~\forall{\mathbf{r}_2 \colon \tau_2 \to D[\tau'_2]}.\\
&\qquad(\forall{x_1 \colon \tau_1}.~\forall{x_2 \colon \tau_2}.~\phi[x_1/\mathbf{r}'_1,~x_2/\mathbf{r}'_2]\implies (\mathfrak{R}^{\beta}_S\phi')[\mathbf{r}_1 x_1/\mathbf{r}_1,~\mathbf{r}_2 x_2/\mathbf{r}_2])
\\
&\qquad\qquad\implies
(\forall{ s_1 \colon D[\tau_1]}.~\forall{ s_2 \colon D[\tau_2]}.~
\mathfrak{R}^{\alpha}_S \phi[s_1/\mathbf{r}_1,~s_2/\mathbf{r}_2]
\\
&\qquad\qquad\qquad\implies \mathfrak{R}^{\alpha\cdot\beta}_S\phi'[\BIND s_1~ \mathbf{r}_1/\mathbf{r}_1,~\BIND s_2~\mathbf{r}_2/\mathbf{r}_2]
).
\end{align*}
To prove this by applying [$\Rightarrow_I$] and [$\forall_I$], having the following \PL- premise
\begin{equation}
\label{proof:TT-lifting:multiplication:000}
\begin{aligned}
&\Gamma,
\mathbf{r}_1 \colon \tau_1 \to D[\tau'_1],
\mathbf{r}_2 \colon \tau_2 \to D[\tau'_2],
s_1 \colon D[\tau_1],
s_2 \colon D[\tau_2]\mid \Psi,\\
&\quad (\forall{x_1 \colon \tau_1}.~\forall{x_2 \colon \tau_2}.~\phi[x_1/\mathbf{r}'_1,~x_2/\mathbf{r}'_2]\implies (\mathfrak{R}^{\beta}_S\phi')[\mathbf{r}_1 x_1/\mathbf{r}_1,~\mathbf{r}_2 x_2/\mathbf{r}_2]),\\
&\quad (\mathfrak{R}^{\alpha}_S \phi[s_1/\mathbf{r}_1,~s_2/\mathbf{r}_2])\vdash_{\mathrm{\PL}}\mathfrak{R}^{\alpha\cdot\beta}_S\phi'[\BIND s_1~\mathbf{r}_1/\mathbf{r}_1,~\BIND s_2~\mathbf{r}_2x_2/\mathbf{r}_2]
\end{aligned}
\end{equation}
We unfold the macro $\mathfrak{R}^{\alpha\cdot\beta}_S\phi'[\BIND s_1~\mathbf{r}_1/\mathbf{r}_1,~\BIND s_2~\mathbf{r}_2/\mathbf{r}_2]$ to:
\begin{equation}
\label{proof:TT-lifting:multiplication:001}
\begin{aligned}
\lefteqn{
\forall{\delta\colon \zeta}.~\forall{f_1 \colon \tau'_1 \to D[\theta_1]}.~\forall{f_2 \colon \tau'_2 \to D[\theta_2]}.~}\\
&\qquad (\forall{x'_1 \colon \tau'_1}.~\forall{x'_2 \colon \tau'_2}.~\phi'[x'_1/\mathbf{r}'_1,~x'_2/\mathbf{r}'_2] \implies S[\delta/\mathbf{k},~f_1(x'_1)/\mathbf{l}_1,~f_2(x'_2)/\mathbf{l}_2]) \\
& \qquad \qquad \implies S\left[\alpha\cdot\beta\cdot\delta/\mathbf{k},\BIND (\BIND s_1~\mathbf{r}_1)~f_1 /\mathbf{l}_1,\BIND (\BIND s_2~\mathbf{r}_2)~f_2/\mathbf{l}_2
\right]
\end{aligned}
\end{equation}
By the associativity of monadic bind,
the fromula (\ref{proof:TT-lifting:multiplication:001}) is equivalent to:
\[
\begin{aligned}
\lefteqn{
\forall{\delta\colon \zeta}.~\forall{f_1 \colon \tau'_1 \to D[\theta_1]}.~\forall{f_2 \colon \tau'_2 \to D[\theta_2]}.~}\\
&\qquad (\forall{x'_1 \colon \tau'_1}.~\forall{x'_2 \colon \tau'_2}.~\phi'[x'_1/\mathbf{r}'_1,~x'_2/\mathbf{r}'_2] \implies S[\delta/\mathbf{k},~f_1(x'_1)/\mathbf{l}_1,~f_2(x'_2)/\mathbf{l}_2]) \\
& \qquad \qquad \implies S\left[\alpha\cdot\beta\cdot\delta/\mathbf{k},\BIND s_1~ \lambda{x_1}.(\BIND \mathbf{r}_1(x_1)~f_1) /\mathbf{l}_1,\BIND s_2~ \lambda{x_2}.(\BIND \mathbf{r}_2(x_2)~f_2) /\mathbf{l}_2
\right]
\end{aligned}
\]
Hence to prove (\ref{proof:TT-lifting:multiplication:001}) by applying [SUBST],  rule with the associativity of monadic bind, and applying [$\Rightarrow_I$] and [$\forall_I$], we need to prove the following \PL-premise:
\begin{equation}
\label{proof:TT-lifting:multiplication:002}
\begin{aligned}
&\Gamma,
\mathbf{r}_1 \colon \tau_1 \to D[\tau'_1],
\mathbf{r}_2 \colon \tau_2 \to D[\tau'_2],
s_1 \colon D[\tau_1],
s_2 \colon D[\tau_2]\\
&\quad
\delta\colon \zeta,
f_1 \colon \tau'_1 \to D[\theta_1],
f_2 \colon \tau'_2 \to D[\theta_2]\\
\mid \Psi,\\
&\quad (\mathfrak{R}^{\alpha}_S \phi[s_1/\mathbf{r}_1,~s_2/\mathbf{r}_2]),\\
&\quad (\forall{x_1 \colon \tau_1}.~\forall{x_2 \colon \tau_2}.~\phi[x_1/\mathbf{r}'_1,~x_2/\mathbf{r}'_2]\implies (\mathfrak{R}^{\beta}_S\phi')[\mathbf{r}_1 x_1/\mathbf{r}_1,~\mathbf{r}_2 x_2/\mathbf{r}_2]),\\
&\quad (\forall{x'_1 \colon \tau'_1}.~\forall{x'_2 \colon \tau'_2}.~\phi'[x'_1/\mathbf{r}'_1,~x'_2/\mathbf{r}'_2] \implies S[\delta/\mathbf{k},~f_1(x'_1)/\mathbf{l}_1,~f_2(x'_2)/\mathbf{l}_2]) \\
\vdash_{\mathrm{\PL}}\\
&\quad
\begin{aligned}
S\left[\alpha\cdot\beta\cdot\delta/\mathbf{k},\BIND s_1~ \lambda{x_1}.(\BIND \mathbf{r}_1(x_1)~f_1) /\mathbf{l}_1,\BIND s_2~ \lambda{x_2}.(\BIND \mathbf{r}_2(x_2)~f_2) /\mathbf{l}_2
\right]
\end{aligned}
\end{aligned}
\end{equation}
To prove this judgment by applying [Ax] and [$\forall_E$] rules to the precondition $(\mathfrak{R}^{\alpha}_S \phi[s_1/\mathbf{r}_1,~s_2/\mathbf{r}_2])$ of (\ref{proof:TT-lifting:multiplication:002}) and applying [$\Rightarrow_E$] rule, we need to prove:
\begin{align*}
(\ldots)
\vdash_{\mathrm{\PL}}
&\forall{x_1 \colon \tau_1}.~\forall{x_2 \colon \tau_2}.~\phi[x_1/\mathbf{r}'_1,~x_2/\mathbf{r}'_2]\implies
S[\beta\cdot\delta/\mathbf{k},
\BIND \mathbf{r}_1x_1~f_1/\mathbf{l}_1,
\BIND \mathbf{r}_2x_2~f_2/\mathbf{l}_2]
\end{align*}
Similarly, to prove this judgment by instantiating the precondition
\[
(\forall{x_1 \colon \tau_1}.~\forall{x_2 \colon \tau_2}.~\phi[x_1/\mathbf{r}'_1,~x_2/\mathbf{r}'_2]\implies (\mathfrak{R}^{\beta}_S\phi')[\mathbf{r}_1 x_1/\mathbf{r}_1,~\mathbf{r}_2 x_2/\mathbf{r}_2])
\]
of (\ref{proof:TT-lifting:multiplication:002}), we need to prove the following judgment:
\begin{align*}
(\ldots)
\vdash_{\mathrm{\PL}}
&\forall{x_1 \colon \tau_1}.~\forall{x_2 \colon \tau_2}.~
(\mathfrak{R}^{\beta}_S\phi')[\mathbf{r}_1 x_1/\mathbf{r}_1,~\mathbf{r}_2 x_2/\mathbf{r}_2])
\\
&
\implies
S[\beta\cdot\delta/\mathbf{k},
\BIND \mathbf{r}_1x_1~f_1 /\mathbf{l}_1,
\BIND \mathbf{r}_2x_2~f_2 /\mathbf{l}_2].
\end{align*}
Similarly, to prove this judgment by instantiating $(\forall{x_1 \colon \tau_1}.~\forall{x_2 \colon \tau_2}.~\phi[x_1/\mathbf{r}'_1,~x_2/\mathbf{r}'_2])$, we need to prove the following judgment:
\begin{align*}
(\ldots)
\vdash_{\mathrm{\PL}}
&\forall{y_1 \colon \tau'_1}.~\forall{y_2 \colon \tau'_2}.~\phi'[y_1/\mathbf{r}'_1,~y_2/\mathbf{r}'_2]\implies
S[\delta/\mathbf{k},f_1(y_1)
/\mathbf{l}_1,f_2(y_2)/\mathbf{l}_2].
\end{align*}
However it is already in the precondition hence we have it by applying [Ax] rule.
\end{proof}

\section{Discussion on the Correctness of Gaussian Learning}

By using the $\top\top$-lifting for the union bound logic,
we can sketch the convergence of Gaussian Learning algorithm.
We change the typing of the primitive of Gaussian
distributions from $M[\REAL]$ to $D[\REAL]$.
Let $\mathtt{Gauss}(\mu,\sigma)^N \colon D[\LIST(\REAL)]$
be a distribution of lists generated from the list $\{d,d,\ldots,d\}$ with length $N$ consists of
the Gaussian distribution $d = \mathtt{Gauss}(\mu,\sigma)$.

We show the following \UPL-judgment through the $\top\top$-lifting for the union bound logic.
\[
\vdash_{\mathrm\UPL}
\BIND \mathtt{Gauss}(\mu,\sigma)^N ~ \mathtt{GaussLearn}(\mathtt{Gauss}(0,1))
\mid \Pr_{r \sim \mathbf{r}}[|r - \mu| \geq \varepsilon] \leq \delta(\varepsilon,N) + \frac{4\sigma^2}{N\varepsilon^2} 
\]
First, in a similar way as the Monte Carlo approximation,
\[
\vdash_{\mathrm\UPL} \mathtt{Gauss}(\mu,\sigma)^N  \colon D[\LIST(\REAL)] \mid \Pr_{L \sim \mathbf{r}}[|\mathtt{Total}(L)/N - \mu| \geq \frac{\varepsilon}{2}] \leq  \frac{4\sigma^2}{N\varepsilon^2}.
\]
This is interpreted by $\top\top$-lifting ($S = (\mathbb{E}_{y \sim \mathbf{l}}[1] \leq \mathbf{k})$) for the union bound logic to:
\[
\vdash_{\mathrm\UPL} \mathtt{Gauss}(\mu,\sigma)^N  \colon D[\LIST(\REAL)] \mid \mathfrak{U}^{ \frac{4\sigma^2}{N\varepsilon^2}}_S(|\mathtt{Total}(\mathbf{r'})/N - \mu| < \frac{\varepsilon}{2}).
\]
Since $\mathtt{GaussLearn}(\mathtt{Gauss(0,1)})(L) = \mathtt{Gauss}(\mathtt{Total}(L)/(|L| + \sigma^2), \sigma^2/(|L| + \sigma^2))$,
there is a function $\delta \colon \REAL \times \INT \to \REAL$ such that  $\delta(\varepsilon,|L|)$ satisfies
\begin{align*}
\vdash_{\mathrm\UPL} &\mathtt{GaussLearn}(\mathtt{Gauss}(0,1)) \colon \LIST(\REAL) \to D[\REAL] \mid\\
&\forall L \colon \LIST(\REAL). |\mathtt{Total}(L)/|L| - \mu | < \frac{\varepsilon}{2}\\
& \implies 
\mathfrak{U}_S^{\delta(\varepsilon,|L|)}([|\mu - \mathtt{Total}(L)/|L| < \frac{\varepsilon}{2} \land |\mathtt{Total}(L)/|L| - \mathbf{r}' | < \frac{\varepsilon}{2}])[\mathbf{r}(s)/\mathbf{r}]
\end{align*}
We also have by the monotonicity of unary graded $\top\top$-lifting:
\begin{align*}
\vdash_{\mathrm\PL}
&
\mathfrak{U}_S^{\delta(\varepsilon,|L|)}([|\mu - \mathtt{Total}(L)/|L| < \frac{\varepsilon}{2} \land |\mathtt{Total}(L)/|L| - \mathbf{r}' | < \frac{\varepsilon}{2}])[\mathbf{r}(s)/\mathbf{r}]\\
&\implies 
\mathfrak{U}_S^{\delta(\varepsilon,|L|)}([|\mu - \mathbf{r}' | < \frac{\varepsilon}{2}])[\mathbf{r}(s)/\mathbf{r}]
\end{align*}

Then we apply the weakening and bind rule on unary graded $\top\top$-lifting, we conclude
\[
\vdash_{\mathrm\UPL}
\BIND \mathtt{Gauss}(\mu,\sigma)^N ~ \mathtt{GaussLearn}(\mathtt{Gauss}(0,1))
\mid 
\mathfrak{U}_S^{\delta(\varepsilon,N) + \frac{4\sigma^2}{N\varepsilon^2}}(|\mathbf{r}' - \mu| < \varepsilon).
\]
This is equivalent to:
\[
\vdash_{\mathrm\UPL}
\BIND \mathtt{Gauss}(\mu,\sigma)^N ~ \mathtt{GaussLearn}(\mathtt{Gauss}(0,1))
\mid \Pr_{r \sim \mathbf{r}}[|r - \mu| \geq \varepsilon] \leq \delta(\varepsilon,N) + \frac{4\sigma^2}{N\varepsilon^2}.
\]
The term $\delta(\varepsilon,N) + \frac{4\sigma^2}{N\varepsilon^2}$ converges to $0$ as $N \to \infty$.
Here $\delta$ is calculated by an upper bound under the condition $|\mathtt{Total}(L)/|L| - \mu| > \frac{\varepsilon}{2}$ of the following probability:
\begin{align*}
\Pr_{r \sim \mathtt{Gauss}(\frac{\mathtt{Total}(L)}{(|L| + \sigma^2)}, \frac{\sigma^2}{(|L| + \sigma^2)})}[ |r - \mathtt{Total}(L)/|L|| \geq \varepsilon]
=
\Pr_{r \sim \mathtt{Gauss}(
\frac{\sigma^2}{(|L| + \sigma^2)}\frac{\mathtt{Total}(L)}{|L|},\frac{\sigma^2}{(|L| + \sigma^2)})}[ |r| \geq \varepsilon].
\end{align*}
Actually, the proof of convergence of $\delta$ in the logic \PL\ is quite complicated, and need to introduce more terminologies of calculations on integrations in \PL, and the proof of convergence itself is far from program verification.
Hence we omit this discussion from the main body of this paper.

\section{Recall: Axioms and Equations of Assertions for Statistics}
\label{sec:axioms}
%
We introduce axioms and equations in the logic \PL.
First, we have the standard equational theory for expressions based on $\alpha$-conversion, $\beta$-reduction, extensionality, and the monadic rules of the monadic type $M$ (we omit here).
The monadic type $M$ also has the commutativity (Fubini-Tonelli equality), written as the following equation:
%
\begin{equation}
\label{eq:MHOL:commutativity}
(\BIND e_1 ~\lambda{x}.(\BIND e_2 ~\lambda{y}.e(x,y)))
=
(\BIND e_2 ~\lambda{y}.(\BIND e_1 ~\lambda{x}.e(x,y)) \quad (x, y \colon \text{fresh})
\end{equation}
We introduce some equalities around expected values.
We have the monotonicity and linearity of expected values (axioms \ref{ineq:MHOL:expected_monotonicity}, \ref{eq:MHOL:expected:linearity}), and we also have Cauchy-Schwartz inequality (axiom \ref{ineq:expectation:Cauchy-Schwarz}). 
We are able to transform the variables in the expression of expected values.
\begin{gather}
\label{ineq:MHOL:expected_monotonicity}
(\forall{x \colon \tau}.~e' \geq 0) \implies \mathbb{E}_{x \sim e}[e'] \geq 0
\\
\label{eq:MHOL:expected:linearity}
\mathbb{E}_{x \sim e}[e_1 \ast e_2] = e_1 \ast \mathbb{E}_{x \sim e}[e_2]
\quad (x\notin \mathrm{FV}(e_1)),
\qquad
\mathbb{E}_{x \sim e}[e_1 + e_2] = \mathbb{E}_{x \sim e}[e_1] + \mathbb{E}_{x \sim e}[e_2] 
\\
\label{ineq:expectation:Cauchy-Schwarz}
(\mathbb{E}_{x \sim e}[e_1 \ast e_2])^2 \leq \mathbb{E}_{x \sim e}[e_1^2] \ast \mathbb{E}_{x \sim e}[e_2^2]
\\
\label{eq:MHOL:expected:variable_transformation}
\mathbb{E}_{x \sim \BIND e~\lambda{y}.\RETURN(e')}[e'']
=\mathbb{E}_{y \sim e}[e''[e'/x]]
\end{gather}
We also introduce some basic equalities on observations, rescaling,
and normalizations.
\begin{gather}
\label{eq:MHOL:expected:scaling}
\mathbb{E}_{x \sim d'}[h(x)\cdot g(x)] = \mathbb{E}_{x \sim \mathtt{scale}(d',g)}[h(x)].
\\
\label{eq:MHOL:scaling1}
(\mathtt{scale}(\mathtt{scale}(e_1, e_2),~ e_3) =
(\mathtt{scale}(e_1, \lambda{x}.(e_2(x) \ast e_3(x))),
\quad
e = \mathtt{scale}(e,\lambda{\_}.1)
\\
\label{eq:MHOL:scaling2}
(\MLET x = \mathtt{scale}(e_1, e_2) \IN e_3(x))
=
(\MLET x = e_1 \IN \mathtt{scale}(e_3(x), \lambda{u}.e_2(x)))
\\
\label{eq:MHOL:scaling3}
\mathtt{scale}(e_1, e_2)\otimes \mathtt{scale}(e_3, e_4)
=
\mathtt{scale}( e_1 \otimes e_2, \lambda{w}. e_2(\pi_1(w)) \ast e_4(\pi_2(w)))
\\
\label{eq:MHOL:scaling4}
\mathbb{E}_{y\sim e}[1] < \infty
\implies
(\BIND e'~\lambda{x}.e)=(\mathtt{scale}(e,\mathbb{E}_{y\sim e'}[1]))
\quad
(x \notin \mathrm{FV}(e))
\\
\label{eq:MHOL:observe}
(\OBSERVE e_1 \AS e_2) = \mathtt{normalize}(\SCALE (e_1,e_2))
\\
\label{eq:MHOL:normalize1}
\mathtt{normalize}(e) = \SCALE(e,\lambda{u}.1/\mathbb{E}_{x \sim e}[1]) \quad (u \notin \mathrm{FV}(\mathbb{E}_{x \sim e}[1]))
\\
\label{eq:MHOL:normalize3}
0 < \alpha < \infty
\implies \mathtt{normalize}(\mathtt{scale}(e_1, e_2))=\mathtt{normalize}(\mathtt{scale}(e_1, \alpha \ast e_2))
\end{gather}
We may introduce the axioms for particular distributions such as
$\mathbb{E}_{x \sim \mathtt{Bern}(e)}[\IF x \THEN 1 \ELSE 0] = e$ ($0 \leq e \leq 1$),
$\mathbb{E}_{x \sim \mathtt{Gauss}(e_1,e_2)}[x]=e_1$
, and etc.
We omit them right now.
%
\subsection{Markov and Chebyshev inequalities}
The axioms in \PL\ that we introduced above are quite standard, but we already able to enjoy meaningful discussions in probability theory.
For instance, we can prove Markov inequality (\ref{eq:MHOL:lemma:Markov_inequality}) and Chebyshev inequality (\ref{eq:MHOL:lemma:Chebyshev_inequality}) in \PL.
\begin{gather}
\label{eq:MHOL:lemma:Markov_inequality}
d \colon M[\REAL],~
a \colon \REAL
\vdash_{\mathrm{\NameOfUnderlyingLogic}}
(a > 0) \implies \Pr_{x \sim d}[ |x| \geq a] \leq \mathbb{E}_{x \sim d}[|x|]/a.\\
\label{eq:MHOL:lemma:Chebyshev_inequality}
\begin{array}{rl}
d \colon M[\REAL],
b \colon \REAL,
\mu \colon \REAL
&\vdash_{\mathrm{\NameOfUnderlyingLogic}}
\mathbb{E}_{x \sim d}[1] = 1 \land \mu = \mathbb{E}_{x \sim d}[x] \land b^2 > 0\\
&\implies \Pr_{x \sim d}[ |x - \mu| \geq b] \leq \mathrm{Var}_{x\sim d}[x]/b^2.
\end{array}
\end{gather}
\subsection{The Reproductive Property and Conversions of Gaussian distributions}
We can introduce in \PL\ the following equalities of the reproductive property of Gaussian distributions and two equalities converting from Gaussian distribution to the standard Gaussian distribution $\mathtt{Gauss}(0,1)$ and vise versa.
\begin{align}
\label{eq:normal:reproductive_property}
\lefteqn{(\BIND \mathtt{Gauss}(\mu_1,\sigma_1^2) \IN 
\lambda{x}.(\BIND \mathtt{Gauss}(\mu_2,\sigma_2^2) ~\lambda{y}.\mathtt{return}(p x  + (1-p) y))
)}\notag\\
&\quad = \mathtt{Gauss}(p\mu_1+(1-p)\mu_2,p^2 \sigma_1^2+(1-p)^2\sigma_2^2).
\qquad\qquad\qquad\qquad\qquad\qquad\qquad
\end{align}
\begin{equation}
\label{eq:normal:normal_and_standard_normal}
(\BIND \mathtt{Gauss}(0,1)~ \lambda{x}.\mathtt{return}(x\sqrt{\sigma^2}+\mu))
= \mathtt{Gauss}(\mu,\sigma^2)
\end{equation}
\subsection{Soundness of Axioms in \PL}
\newcommand{\Kint}{\square\hspace{-1.0em}\int}
Soundness of many axioms are proved by using the equations in the
toolbox for synthetic measure theory given in \citet[Figure 14 (the last page)]{Scibior:2017:DVH:3177123.3158148}.
Roughly speaking, they consist of notations of the structures
relating the commutative monad $\mathfrak{M}$ on the cartesian closed category $\QBS$.
The notations of synthetic measure theory for the commutative monad $\mathfrak{M}$ and semiring $R = [0,\infty]$ is unfolded as follows:
\begin{mathpar}
\Kint_X f(x)~d\mu(x) \defeq f^\sharp(\mu)

w \odot \mu \defeq \Kint_X w(x)\cdot\mathbf{d}_x ~d\mu(x) = \mathfrak{M}(\pi_2) \circ (\mathrm{dst}^\mathfrak{M}_{1,X} \circ \lrangle{w,\eta})^\sharp
\end{mathpar}
Here, $ w(x)\cdot\mathbf{d}_x$ is a scalar multiplication of the Dirac measure $\mathbf{d}_x$ with $w(x)$, and the projection $\pi_2 \colon 1 \times X \to X$ is also the left unitor of Cartesian product (isomorphism).
We can then formalize the semantics of the monadic bind, expectation, and rescaling as follows:
\begin{align*}
\interpret{\Gamma\vdash\BIND e ~f} &= \lambda{\gamma}.\Kint (\interpret{\Gamma\vdash f}(\gamma))(x)~d(\interpret{\Gamma\vdash e}(\gamma))(x)\\
\interpret{\Gamma\vdash\mathbb{E}_{x\sim e}[f(x)]} &=\lambda{\gamma}.(\inverse{\cong} \circ \Kint ({\cong} \circ \interpret{\Gamma\vdash f}(\gamma)))(x)~d(\interpret{\Gamma\vdash e}(\gamma))(x)\\
&\labeleq{(\dag)}\lambda{\gamma}. \int (\interpret{\Gamma\vdash f}(\gamma))(x)~d \interpret{\Gamma\vdash e}(\gamma))(x)\\
\interpret{\Gamma\vdash\SCALE(e,f)} &= \lambda{\gamma}. \interpret{\Gamma\vdash f}(\gamma) \odot \interpret{\Gamma\vdash e}(\gamma)
\end{align*}
where $\cong$ is the isomorphism $\mathfrak{M}1 \cong [0,\infty]$.
The second reformulation of expectation is actually integration in quasi-Borel space. The equality (\dag) is given from the fact that the correspondence between $f^\sharp \mu = \int f~d\mu$ in the case of $f \colon X \to [0,\infty]$ and $\mu \in \mathfrak{M}X$.

Since $\QBS$ is well-pointed, to prove the soundness of equalities on \PCFP\ probabilistic terms, it suffices to show the semantic equation holds for any snapshot $\gamma$ of environment $\Gamma$ satisfying the precondition.

\begin{itemize}
\item 
Soundness of equalities
(\ref{eq:MHOL:commutativity}),
(\ref{eq:MHOL:expected:variable_transformation}),
(\ref{eq:MHOL:expected:scaling}),
(\ref{eq:MHOL:scaling1}), and
(\ref{eq:MHOL:scaling4})
are derived from the equations given in the toolbox for synthetic measure
theory~\citep[Figure 14 (the last page)]{Scibior:2017:DVH:3177123.3158148}.
Notice that $\alpha \cdot \beta \cdot \mathbf{d}_{\lrangle{x,y}} =
(\alpha \cdot \mathbf{d}_{x}) \otimes (\beta\cdot \mathbf{d}_{y})$ holds by definition of Dirac distribution.

\item For the monotonicity (inequality \ref{ineq:MHOL:expected_monotonicity}), linearity (equalities \ref{eq:MHOL:expected:linearity}), and Cauchy-Schwartz inequality (inequality \ref{ineq:expectation:Cauchy-Schwarz}) of expectations, we use the second reformulation of expectation (\dag).
Integrations for measures on qbs are converted into the usual Lebesgue integration, so we may apply the existing lemmas on usual measure theory.

For an expectation of a real function which may take negative values, we can easily check that the interpretation of the syntactic sugar corresponds to the actual expected value which is given by an integration directly.

\item The soundness of the equality $\OBSERVE e \AS f = \mathtt{normalize}(\SCALE(e,f))$ (equality \ref{eq:MHOL:observe}) is obvious 
from the definition.

\item The soundness of the equality $\mathtt{normalize}(e) = \SCALE(e,\lambda{\_}.1/\mathbb{E}_{x\sim e}[1])$ (equality \ref{eq:MHOL:normalize1}) and normalizing constant-rescaled distribution  (\ref{eq:MHOL:normalize3})
are proved by the equivalence of scalar multiplication and rescaling with constant scalar function, and the equivalence of the mass of a measure and the expectation $\mathbb{E}_{x\sim e}[1]$. Both are easily proved.
\end{itemize}
\end{document}